\documentclass[a4paper]{article}[11pts]

\usepackage[english]{babel}
\usepackage[utf8x]{inputenc}
\usepackage[T1]{fontenc}

\normalsize

\usepackage[a4paper,top=1in,bottom=1in,left=1in,right=1in,marginparwidth=1in]{geometry}

\usepackage{setspace}
\usepackage{amsmath}
\usepackage{amssymb} 
\usepackage{amsthm}
\usepackage{amsfonts, mathtools}
\usepackage{algorithm,algpseudocode}
\usepackage{bm}
\usepackage{graphicx}
\usepackage{multirow, booktabs}
\usepackage{caption}
\usepackage{subcaption}
\usepackage{xcolor}

\usepackage[colorinlistoftodos]{todonotes}
\usepackage[colorlinks=true, allcolors=blue]{hyperref}

\newcommand{\p}{\bm{p}}
\newcommand{\Bb}{\bm{b}}
\newcommand{\Dd}{\bm{d}}
\newcommand{\Aa}{\bm{a}}
\newcommand{\uu}{\bm{u}}

\newcommand{\E}{\mathbb{E}}
\newcommand{\R}{\mathbb{R}}
\newcommand{\prob}{\mathbb{P}}
\DeclareMathOperator*{\argmax}{arg\,max}
\DeclareMathOperator*{\argmin}{arg\,min}
\DeclarePairedDelimiter\ceil{\lceil}{\rceil}
\DeclarePairedDelimiter\floor{\lfloor}{\rfloor}
\DeclarePairedDelimiter\Bfloor{\Bigg\lfloor}{\Bigg\rfloor}

\newtheorem{proposition}{Proposition}
\newtheorem{lemma}{Lemma}
\newtheorem{assumption}{Assumption}
\newtheorem{theorem}{Theorem}
\newtheorem{corollary}{Corollary}
\newtheorem{example}{Example}
\newtheorem*{theorem*}{Theorem}

\usepackage{natbib}

\title{Online Linear Programming: Dual Convergence, New Algorithms, and Regret Bounds}
\author{Xiaocheng Li$^\dagger$ \and  Yinyu Ye$^\dagger$}
\date{\small 
$^\dagger$Department of Management Science and Engineering, Stanford University\\
$\{$chengli1, yyye$\}$@stanford.edu
}

\begin{document}
\maketitle

\begin{abstract}
We study an online linear programming (OLP) problem under a random input model in which the columns of the constraint matrix along with the corresponding coefficients in the objective function are generated i.i.d. from an unknown distribution and revealed sequentially over time. Virtually existing online algorithms were based on learning the dual optimal solutions/prices of the linear programs (LP), and their analyses were focused on the aggregate objective value and solving the packing LP where all coefficients in the constraint matrix and objective are nonnegative. However, two major open questions were: (i) Does the set of LP optimal dual prices learned in the existing algorithms converge to those of the ``offline'' LP, and (ii) Could the results be extended to general LP problems where the coefficients can be either positive or negative. We resolve these two questions by establishing convergence results for the dual prices under moderate regularity conditions for general LP problems. Specifically, we identify an equivalent form of the dual problem which relates the dual LP with a sample average approximation to a stochastic program. Furthermore, we propose a new type of OLP algorithm, Action-History-Dependent Learning Algorithm, which improves the previous algorithm performances by taking into account the past input data as well as the past decisions/actions. We derive an $O(\log n \log \log n)$ regret bound (under a locally strong convexity 
and smoothness condition) for the proposed algorithm, against the $O(\sqrt{n})$ bound for typical dual-price learning algorithms, where $n$ is the number of decision variables. Numerical experiments demonstrate the effectiveness of the proposed algorithm and the action-history-dependent design. 
\end{abstract}

\section{Introduction}

Sequential decision making has been an increasingly attractive research topic with the advancement of information technology and the emergence of new online marketplaces. As a key concept appearing widely in the fields of operations research, management science, and artificial intelligence, sequential decision making concerns the problem of finding the optimal decision/policy in a dynamic environment where the knowledge of the system, in the form of data and samples, amasses and evolves over time. In this paper, we study the problem of solving linear programs in a sequential setting, usually referred to as online linear programming (OLP) (See e.g., \citep{agrawal2014dynamic}). The formulation of OLP has been widely applied in the context of online Adwords/advertising \citep{mehta2005adwords}, online auction market \citep{buchbinder2007online}, resource allocation \citep{asadpour2019online}, packing and routing \citep{buchbinder2009online}, and revenue management \citep{talluri2006theory}. One common feature in these application contexts is that the customers, orders, or queries arrive in a forward sequential manner, and the decisions need to be made on the fly with no future data/information available at the decision/action point. 

The OLP problem takes a standard linear program as its underlying form (\textbf{with $n$ decision variables and $m$ constraints}), while the constraint matrix is revealed column by column with the corresponding coefficient in the linear objective function. In this paper, we consider the standard random input model (See \citep{goel2008online, devanur2019near}) where the orders, represented by the columns of the constraint matrix together with the corresponding objective coefficients, are sampled independently and identically from \textbf{an unknown distribution} $\mathcal{P}$. At each timestamp, the value of the decision variable needs to be determined based on the past observations and cannot be changed afterward. The goal is to minimize the gap (formally defined as regret) between the objective value solved in this online fashion and the ``offline'' optimal objective value where one has the full knowledge of the linear program data.

There were many algorithms and research results on OLP in the past decade due to its  wide applications. Virtually all existing online algorithms were based on learning the LP dual optimal solutions/prices, and their analyses of OLP were focused on the aggregate objective value and solving the packing LP where all coefficients in the constraint matrix and objective are nonnegative. Two major open questions in the literature were: (1) Does the set of LP optimal dual prices of OLP converge to those of the offline LP, and (2) Could the results be extended to general LP problems where the coefficients can be either positive or negative. We resolve these two questions in this paper as part of our results. Moreover, we propose a new type of OLP algorithm and develop tools to analyze the regret of OLP algorithms. Our key results and main contributions are summarized as follows.

\subsection{Key Results and Main Contributions}

\textbf{Dual convergence of online linear programs}. 
We establish convergence results for the dual optimal solutions of a sequence of linear programs in Section \ref{dualconvergence}. We first derive an equivalent form of the dual LP and discover that the sampled dual LP, under the random input model, can be viewed as a \textit{Sample Average Approximation} (SAA) \citep{kleywegt2002sample, shapiro2009lectures} of a constrained stochastic programming problem. The stochastic program is defined by the LP constraint capacity and the distribution $\mathcal{P}$ that generates the input of the LP. Our key result states that, under moderate regularity conditions, the optimal solution of the sampled dual LP will converge to the optimal solution of the stochastic program as the number of (primal LP) decision variables goes to infinity. Specifically, we establish that the L$_2$ distance between the two solutions is $\tilde{O}\left(\frac{\sqrt{m}}{\sqrt{n}}\right)$ under the random input model where $m$ is the number of constraints and $n$ is the number of decision variables. Moreover, the convergence results are not only pertaining to online packing LPs, but also hold for general LPs where the input data coefficients can be either positive or negative. 



\textbf{Action-history-dependent learning algorithm.} We develop a new type of OLP algorithm – Action-history-dependent Learning Algorithm in Section \ref{AHDLA}. This new algorithm is a dual-based algorithm (as the algorithms in \citep{devanur2011near, agrawal2014dynamic, gupta2014experts}, etc.) and it utilizes our results on the convergence of the sampled dual optimal solution. One common pattern in the design of most existing OLP algorithms is that the choice of the decision variable at time $t$ only depends on the past input data, i.e., the coefficients in the constraints and the objective function revealed, but not the decisions already made (until time $t-1$). Our new action-history-dependent algorithm considers both the past input data and the past choice of decision variables. Similar idea was considered in a few specific problems such as network revenue management and online auction. Compared to the existing OLP algorithms, our new algorithm is more conscious of the constraints/resources consumed by the past decisions, and thus the decisions can be made in a more dynamic, closed-loop, and non-stationary way. We demonstrate in both theory and numerical experiments that this actions-history-dependent mechanism significantly improves the online performance than existing OLP algorithms without this mechanism.

\textbf{Regret bounds for OLP.} We analyze the worst-case gap (regret) between the expected online objective value and the ``offline'' optimal objective value. Specifically, we study the regret in an asymptotic regime where the number of constraints $m$ is fixed as a constant and the LP right-hand-side input scales linearly with the number of decision variables $n$. As far as we know, this is the first regret analysis result in the general OLP formulation. We derive an $O(\log n \log \log n)$ regret upper bound (under a locally strong convexity and smoothness assumption) for the proposed action-history-dependent learning algorithm, which has a similar order of magnitude as the best achievable lower bound $\Omega(\log n)$ of the problem even with the exact knowledge of the underlying distribution \citep{bray2019does}. Our regret analysis provides an algorithmic insight for the constrained online optimization problem: a successful algorithm should have good control of the binding constraints (resources) consumption -- not exhausting those constraints too early or having too much remaining at the end, an aspect usually overlooked by typical dual-price learning algorithms in OLP literature. The analysis extends the findings in the network revenue management literature (for example, the adaptive resource control in \citep{jasin2012re, jasin2015performance}) to a more general non-parametric context. Moreover, the results and methodologies (such as Theorem \ref{representation}) are  potentially applicable to other online learning and online decision making problems.


\subsection{Literature Review}

Online optimization/learning/decision problems have been long studied in the community of operations research and theoretical computer science. We refer readers to the papers \citep{borodin2005online, buchbinder2009design, hazan2016introduction} for a general overview of the topic and the recent developments. The OLP problem has been studied mainly under two models: the random input model \citep{goel2008online, devanur2019near} and the random permutation model \citep{molinaro2013geometry, agrawal2014dynamic, gupta2014experts}. In this paper, we consider the random input  model (also known as stochastic input model) where the columns of constraint matrix are generated i.i.d. from an unknown distribution $\mathcal{P}$. In comparison, the random permutation model assumes the columns are arriving in  random order and the arrival order is uniformly distributed over all the  permutations. Technically, the i.i.d. assumption in the random input model is stronger than the random permutation assumption in that the random input model can be viewed as a special case of the random permutation model (\citep{mehta2013online}). Practically, the random input model can be motivated from the online advertising problem, network revenue management problems in flight and hotel bookings, or the online auction problem. In these application contexts, each column in the constraint matrix together with the corresponding coefficient in the objective function represents an individual order/bid/query. In this sense, the random input model can be interpreted as an independence assumption across different customers.

Our paper differs from the existing OLP literature in the right-hand-side assumption, i.e., the constraint capacity. A stream of OLP papers \citep{devanur2009adwords, molinaro2013geometry, agrawal2014dynamic, kesselheim2014primal, gupta2014experts} studied the trade-off between the algorithm competitiveness and the constraint capacity. They investigated the necessary and sufficient condition on the right-hand-side of the LP and the number of constraints $m$ for the existence of a $(1-\epsilon)$-competitive OLP algorithm. Specifically, \cite{agrawal2014dynamic} established the necessary part by constructing a worst-case example, stating that the right-hand-side should be no smaller than $\Omega\left(\frac{\log m}{\epsilon^2}\right)$. \cite{kesselheim2014primal, gupta2014experts} developed algorithms that achieve $(1-\epsilon)$-competitiveness under this necessary condition and thus completed the sufficient part. In this paper, we research an alternative question, when the right-hand-side grows linearly with the number of decision variables $n$, whether the algorithm could achieve a better performance than $(1-\epsilon)$-competitiveness. In general, this linear growth regime will render the optimal objective value growing linearly with $n$ as well. Consequently, a $(1-\epsilon)$-competitiveness performance guarantee will potentially incur a gap that is linear in $n$ between the online objective value and the offline optimal value. Thus we consider regret instead of competitiveness ratio as the performance measure and will analyze three algorithms that achieve sublinear regret under the linear growth regime. Moreover, a typical assumption in the OLP literature requires the data entries in the constraint matrix and the objective to be non-negative. We do not make this assumption in our model so that our model and analysis can capture a double-sided market with both buying and selling orders. 

Another stream of research originates from the revenue management literature and studies a parameterized type of the OLP problem. It models the network revenue management problem where heterogeneous customers arrive sequentially to the system and the customers can be divided into finitely many classes based on their requests and prices. In the language of OLP, the columns of the constraint matrix with the corresponding coefficients in the objective follow a well-parameterized distribution and have a finite and known support. \cite{reiman2008asymptotically, jasin2012re, jasin2013analysis, bumpensanti2018re} among others studied the problem under a setting where the model parameters are known and discussed the performance of a re-solving technique that dynamically solves the certainty-equivalent problem according to the current state of the system. This line of work highlights the effectiveness of the re-solving technique in an environment with known parameters and investigates the desired re-solving frequency. In a similar spirit, \cite{jasin2015performance} analyzed the performance of the re-solving based algorithm in an unknown parameter setting, and \cite{ferreira2018online} studied a slightly different pricing-based revenue management problem. The OLP model generalizes the network revenue management problem in that it does not impose any parametric structure on the distribution. For the OLP problem, the distribution of the customer request and price may have an infinite support (such as secretary problem/Adwords problem) and negative values are allowed (such as a two-sided auction market). Comparatively, the revenue management literature focuses on the case where the underlying distribution has a parametric structure and finite support. For example, \cite{jasin2015performance} studied an unknown setting, but the paper still assumed the knowledge of the distribution's support. Consequently, the parameter learning in \citep{jasin2015performance} was reduced to a simple intensity estimation problem for a homogeneous Poisson process and the algorithm therein relied on the estimation together with the re-solving technique. In fact, the algorithms developed along this line of literature fail for the more general OLP problem because when the distribution is fully non-parametric and unknown, there is no way to \textit{first estimate} the distribution parameter and \textit{then to solve} a certainty-equivalent optimization problem based on the estimated parameter. The dual-based algorithms developed in our paper can thus be viewed as a combination of this first-estimate-then-optimize procedure into one single step. In particular, our action-history-dependent algorithm implements the idea of re-solving technique in a non-parametric setting, and our algorithm analysis reinforces the effectiveness of the adaptive and re-solving design (mainly discussed in the revenue management literature) in a more general online optimization context.  

Another line of research investigated the multi-secretary problem (\citep{kleinberg2005multiple, arlotto2019uniformly, bray2019does} among others). The multi-secretary problem is a special form of OLP problem that has only one constraint and all the coefficients in the constraint matrix are one. \cite{arlotto2019uniformly} showed that when the reward distribution is unknown and if no additional assumption is imposed on the distribution, the regret lower bound of the multi-secretary problem is $\Omega(\sqrt{n})$. A subsequent work \citep{bray2019does} further noted that even when the distribution is known, the regret lower bound is $\Omega(\log n).$ \cite{balseiro2019learning} studied a similar one-constraint but multi-agent online learning formulation motivated from the repeated auction problem. The results in our paper are positioned in a more general context and consistent with this line of works. We identify a group of assumptions that admits an $\tilde{O}(\log n)$ regret upper bound for the more general OLP problem. Our action-history-dependent algorithm can also be viewed as a generalization of the elegant adaptive algorithms developed in \citep{arlotto2019uniformly, balseiro2019learning}.
Together with \citep{bray2019does}, our results indicate that the action-history-dependent algorithm achieves a near-optimal regret performance for the multi-secretary problem even in a comparison with the optimal dynamic algorithm developed with knowledge of the distribution.

The OLP problem is also related to the general online optimization problem. Compared to the standard online optimization problem, our OLP problem is special in two aspects: (i) the presence of the constraints and (ii) a dynamic oracle as the regret benchmark. For the literature on online optimization with constraints \citep{mahdavi2012trading, agrawal2014fast, yu2017online,yuan2018online}, the common approach is to employ a bi-objective performance measure and report the regret and constraint violation separately. We contribute to this line of research by developing a machinery to analyze the regret of a feasible online algorithm. For the second aspect, the dynamic oracle allows the decision variables (in OLP) of different time steps to take different values. This is a stronger oracle than the static oracle \citep{mahdavi2012trading, agrawal2014fast, yu2017online,yuan2018online} that requires the decision variables at different time steps to take the same (static) value. In other words, the OLP regret is computed against a stronger benchmark. This explains why OLP literature considers mainly the random input model and the random permutation model, instead of an adversarial setting.



\section{Problem Formulation}

In this section, we formulate the OLP problem and define the objective. Consider a generic LP problem 
\begin{align}
\label{primalLP}
    \max\ & \sum_{j=1}^n r_j x_j 
    \\ 
    \text{s.t.}\ & \sum_{j=1}^n a_{ij}x_j \le b_i,\ \  i=1,...,m \nonumber \\
    & 0 \le x_{j} \le 1, \ \  j=1,...,n \nonumber
\end{align}
where $r_j \in \R,$ $\bm{a}_{j} = (a_{1j}, ..., a_{mj})^\top \in \R^m$, and $\bm{b} = (b_{1},...,b_{m})^\top \in \R^m.$ Without loss of generality, we assume $b_i>0$ for $i=1,...,m.$ Throughout this paper, we use bold symbols to denote vectors/matrices and normal symbols for scalars.

In the online setting, the parameters of the linear program are revealed in an online fashion and one needs to determine the value of decision variables sequentially. Specifically, at each time $t,$ 
the coefficients $(r_t, \bm{a}_t)$ are revealed, and we need to decide the value of $x_t$ instantly. Different from the offline setting, at time $t$, we do not have the information of the following coefficients to be revealed. Given the history $\mathcal{H}_{t-1} = \{r_j, \bm{a}_j, x_j\}_{j=1}^{t-1}$, the decision of ${x}_t$ can be expressed as a policy function of the history and the coefficients observed in the current time period. That is,
\begin{equation}
    {x}_t = \pi_t(r_t, \bm{a}_t, \mathcal{H}_{t-1}).
    \label{policy}
\end{equation}
The policy function $\pi_t$ can be time-dependent and we denote policy $\bm{\pi} = (\pi_{1},...,\pi_{n}).$ The decision variable ${x}_t$ must conform to the constraints
$$\sum_{j=1}^t a_{ij}{x}_j \le b_i,\ \  i=1,...,m,$$
$$0 \le {x}_t \le 1.$$
The objective is to maximize the objective $\sum_{j=1}^n r_j{x}_j.$ 

We illustrate the problem setting through the following practical example. Consider a market making company receiving both buying and selling orders, and the orders arrive sequentially. At each time $t,$ we observe a new order, and we need to decide whether to accept or reject the order. The order is a buying/selling request for the resources, or it could be a mixed request, e.g., selling the first resource for $1$ unit and buying the second resource for $2$ units with a total order price of \$1. Once our decision is made, the order will leave the system, and it is either fulfilled or rejected. In this example, the term $b_i$ can be interpreted as the total available inventory for the resource $i,$ and the decision variables $x_t$'s can be interpreted as the acceptance and rejection of an order. In particular, we do not allow shorting of resources along the process. Our goal is to maximize the total revenue. 

We assume the LP parameters $(r_j, \bm{a}_j)$ are generated i.i.d. from an unknown distribution $\mathcal{P}.$  We denote the offline optimal solution of linear program (\ref{primalLP}) as $\bm{x}^* = (x_1^*,...,x_n^*)^\top$, and the offline (online) objective value as $R_n^*$ (${R}_n$). Specifically, 
\begin{align*}
R_n^* & \coloneqq \sum_{j=1}^n r_{j} x_j^*\\
{R}_n(\bm{\pi}) & \coloneqq \sum_{j=1}^n r_{j} {x}_j.
\end{align*}
in which online objective value depends on the policy $\bm{\pi}$. The quantity $R_n^*$ assumes the full knowledge of the realization (of the randomness), and it is also known as \textit{hindsight oracle} in the literature of online learning and robust optimization. In this paper, we consider a fixed $m$ and large $n$ regime, and focus on the worst-case gap between the online and offline objective. 
We define the \textit{regret} 
$$ \Delta^\mathcal{P}_n(\bm{\pi}) \coloneqq \E_{\mathcal{P}} \left[R_n^*-{R}_n(\bm{\pi})\right]  $$
and the \textit{worst-case regret} 
$$\Delta_n(\bm{\pi}) \coloneqq \sup_{\mathcal{P} \in \Xi} \Delta^\mathcal{P}_n(\bm{\pi}) = \sup_{\mathcal{P} \in \Xi} \E_{\mathcal{P}} \left[R_n^*-{R}_n(\bm{\pi})\right]$$
where $\Xi$ denotes a family of distributions satisfying some regularity conditions (to be specified later). Throughout this paper, we omit the subscript $\mathcal{P}$ in the expectation notation when there is no ambiguity. The worst-case regret takes the supremum regret over a family of distributions so that it is suitable as a performance guarantee when the distribution $\mathcal{P}$ is unknown. We remark that the offline optimal solution $R_n^*$ can be interpreted as a dynamic oracle that allows the optimal decision variables to take different values at different time steps. This is an important distinction between the OLP problem and the problem of online convex optimization with constraints \citep{mahdavi2012trading, agrawal2014fast, yu2017online,yuan2018online}.

\subsection{Notations}

Throughout this paper, we use the standard big-O notations where $O(\cdot)$ and $\Omega(\cdot)$ represent upper and lower bound, respectively. The notation $\tilde{O}(\cdot)$ further omits the logarithmic factor, such as $\tilde{O}(\sqrt{n})=O(\sqrt{n}\log n)$ and $\tilde{O}(\log{n})=O(\log{n}\log \log n)$. The following list summarizes the notations used in this paper:
\begin{itemize}
    \item $m$: number of constraints; $n$: number of decision variables
    \item $i$: index for constraint; $j,t$: index for the decision variables
    \item $r_j$: the $j$-th coefficient in the objective function
    \item $\Aa_j$: the $j$-th column in the constraint matrix
    \item $\bar{r}, \bar{a}$: upper bound on $|r_j|$'s and $\|\Aa_j\|_2$'s
    \item $\mathcal{P}$: distribution of $(r_{j}, \Aa_{j})$'s; $\mathcal{P}\in\Xi$: a family of distributions to be defined in the next section
    \item $\lambda, \mu$: Parameters that convey the meaning of strong convexity and smoothness; they are related to the condition distribution of $r|\bm{a}$ to be defined in the next section
    \item $\bm{x}^*=(x_1^*,...,x_n^*)^\top$: the offline primal optimal solution
    \item $\bm{x}=(x_1,...,x_n)^\top$: the online solution
    \item $\bm{p}_n^*$: the (random) offline dual optimal solution 
    \item $\bm{p}^*$: the (deterministic) optimal solution of the stochastic program (\ref{asymProblem})
    \item $R_n^*$: the offline optimal objective value
    \item $R_n(\pi)$ (sometimes $R_n$ when the policy is clear): the online return/revenue under policy $\pi$
    \item $\bm{b}=(b_1,...,b_m)^\top$: the right-hand-side, constraint capacity
    \item $\bm{d}$: the average constraint capacity, i.e., $\bm{d}=\bm{b}/n$
    \item $\underline{d}, \bar{d}$: lower and upper bounds for $\bm{d}$
    \item $\Omega_d$: the region for $\bm{d}$, defined by $\bigotimes_{i=1}^m (\underline{d},\bar{d})$
    \item $I_B$ and $I_N$: the index sets for binding and non-binding constraints defined by the stochastic program (\ref{asymProblem}), respectively
    \item $\bm{b}_t$: the remaining constraint capacity at the end of the $t$-th period, $t=1,...,n$
    \item $\bm{d}_t$: the remaining average constraint capacity at the end of the $t$-th period, i.e., $\bm{d}_t=\bm{b}_t/(n-t)$, $t=1,...,n-1$
    \item $\wedge$: minimum operator, $y \wedge z\coloneqq \min\{y,z\}$ for $y,z\in\mathbb{R}$
    \item $\vee$: maximum operator, $y \vee z\coloneqq \max\{y,z\}$ for $y,z\in\mathbb{R}$
    \item $I(\cdot)$: indicator function; $I(\mathcal{E})=1$ when $\mathcal{E}$ is true and $I(\mathcal{E})=0$ otherwise
\end{itemize}

\section{Dual Convergence}

\label{dualconvergence}
Many OLP algorithms rely on solving the dual problem of the linear program (\ref{primalLP}). However, there is still a lack of theoretical understanding of the properties of the dual optimal solutions. In this section, we establish convergence results on the OLP dual solutions and lay foundations for the analyses of the OLP algorithms. 

To begin with, the dual of the linear program (\ref{primalLP}) is
\begin{align}
    \min\ & \sum_{i=1}^m b_ip_i + \sum_{j=1}^n y_j   \label{dualLP} \\
    \text{s.t.}\ & \sum_{i=1}^m a_{ij} p_i +y_j \ge r_j, \ \  j=1,...,n. \nonumber
    \\ 
    & p_i, y_j \ge 0 \text{ for all } i,j. \nonumber
\end{align}
Here the decision variables are $\bm{p} = (p_1,...,p_m)^\top$ and $\bm{y} = (y_1,...,y_n)^\top$.

Let $(\bm{p}^*_n, \bm{y}^*_n)$ be an optimal solution for the dual LP (\ref{dualLP}). From the complementary slackness condition, we know the primal optimal solution satisfies
\begin{equation}
    x^*_{j} = \begin{cases} 1, & \ r_j > \bm{a}_j^{\top} \bm{p}^*_n  \\
0,& \  r_j < \bm{a}_j^{\top} \bm{p}^*_n.
\end{cases}
\label{decisionRule}
\end{equation}
When $r_j=\bm{a}_j^{\top} \bm{p}^*_n$, the optimal solution $x_j^*$ may take non-integer values. This tells us that the primal optimal solution largely depends on the dual optimal solution $\p^*_n$ and thus motivates our study of the optimal dual solutions. An equivalent form of the dual LP can be obtained by plugging the constraints into the objective function (\ref{dualLP}): 
\begin{align}
    \min\ & \sum_{i=1}^m b_ip_i + \sum_{j=1}^n \left(r_j-\sum_{i=1}^m a_{ij}p_i\right)^+   \label{newDual}
    \\ 
   \text{s.t. } & p_i \ge 0, \ \ i=1,...,m. \nonumber
\end{align}
where $(\cdot)^+$ is the positive part function, also known as the ReLu function.  The optimization problem (\ref{newDual}), despite not being a linear program, has a convex objective function. 
It has the advantage of only involving $\bm{p}$ which is closely related to the optimal primal solution. More importantly, the random summands in the second part of the objective function (\ref{newDual}) are independent of each other, and therefore the sum (after a normalization) will converge to a certain deterministic function. To better make this point, let $d_i=b_i/n$ and divide the objective function in (\ref{newDual}) by $n$. Then the optimization problem can be rewritten as 
\begin{align}
    \min\ & f_n(\p) \coloneqq \sum_{i=1}^m d_ip_i + \frac{1}{n}\sum_{j=1}^n \left(r_j-\sum_{i=1}^m a_{ij}p_i\right)^+   \label{new2Dual}
    \\ 
   \text{s.t. } & p_i \ge 0, \ \ i=1,...,m. \nonumber
\end{align}
The second term in the objective function (\ref{new2Dual}) is a summation of $n$ i.i.d. random functions. 

Consider the following stochastic program
\begin{align} \min & \ f(\bm{p}) \coloneqq  \bm{d}^\top \bm{p} + \E\left[(r-\bm{a}^\top \bm{p})^+\right] \label{asymProblem}\\
\text{s.t. \ } & \bm{p} \ge \bm{0}, \nonumber
\end{align}
where the expectation is taken with respect to $(r, \bm{a})$. In the rest of the paper, unless otherwise stated, the expectation is always taken with respect to $(r, \bm{a})$. Evidently,
$$\E f_n(\p) = f(\p)$$
for all $\p.$ This observation casts the dual convergence problem in the form of a stochastic programming problem. The function $f_n(\bm{p})$ in (\ref{new2Dual}) can be viewed as a \textit{sample average approximation} (SAA) (See \citep{kleywegt2002sample, shapiro2009lectures}) of the function $f(\bm{p}).$ Specifically, the dual program associated with a primal linear program with $n$ decision variables is then an $n$-sample approximation of the stochastic program (\ref{asymProblem}). We denote the optimal solutions to the $n$-sample approximation problem (\ref{new2Dual}) and the stochastic program (\ref{asymProblem}) with $\bm{p}^*_n$ and $\bm{p}^*$, respectively. In this section, we provide a finite-sample analysis for the convergence of $\p_n^*$ to $\p^*$. We first introduce the assumptions and then formally establish the convergence.


\subsection{Assumptions and Basics}
The first group of assumptions concerns the boundedness and the linear growth of the constraints. 

\begin{assumption}[Boundedness and Linear Growth Capacity]
We assume 
\begin{itemize}
    \item[(a)] $\left\{(r_j, \bm{a}_{j})\right\}_{j=1}^n$ are generated i.i.d. from distribution $\mathcal{P}$.
    \item[(b)] There exist constants $\bar{r}, \bar{a} >0$ such that $|r_j|\le \bar{r}$ and $\|\bm{a}_j\|_2 \le \bar{a}$ almost surely. 
    \item[(c)] $d_i = b_i/n \in (\underline{d},\bar{d})$ for $\underline{d}, \bar{d}>0$, $i = 1,...,m.$ Denote $\Omega_d = \bigotimes_{i=1}^m (\underline{d}, \bar{d}).$
    \item[(d)] $n>m.$
\end{itemize}
Throughout this paper, $\|\cdot\|_2$ denotes the L$_2$-norm of a vector.
\label{assump1}
\end{assumption}

Assumption \ref{assump1} (a) states that parameters (coefficients in the objective function and columns in the constraint matrix) of linear program (\ref{primalLP}) are generated i.i.d. from an unknown distribution $\mathcal{P}.$ The
vectors $\{(r_j, \bm{a}_{j}), j=1,...,n\}$ are independent of each other, but their components may be dependent. Assumption \ref{assump1} (b) requires the parameters are bounded. The bound parameters $\bar{a}$ and $\bar{r}$ are introduced only for analysis purposes and will not be used for algorithm implementation. Assumption \ref{assump1} (c) requires the right-hand-side of the LP constraints grows linearly with $n$. This guarantees that for the (optimal) solutions, a constant proportion of the $x_j$'s could be $1$. It means the number of orders/requests that can be fulfilled is on the order of $n$ and thus ensures a constant service level (percentage of orders satisfied). On the contrary, if this is not true and all the requests are buying orders ($a_{ij}>0$), the service level may go to zero when the business running period $n$ goes to infinity. The parameter $d_i=b_i/n$ has the interpretation of available constraint/resource per period. Also, we require that the number of decision variables $n$ is larger than the number of constraints $m$. While discussing the dual convergence, the dimension of the dual variable $\p$ is equal to the number of constraints $m$ and the number of primal decision variables $n$ can be viewed as the number of samples used to approximate $\p^*.$ The assumption of $n>m$ restricts our attention to a low-dimensional setting.

Proposition \ref{basicProp} summarizes several basic properties related to the dual LP (\ref{dualLP}) and the stochastic program (\ref{asymProblem}). It states that  both the objective functions in the SAA problem and the stochastic program are convex. In addition, the optimal solutions to these two problems are bounded.

\begin{proposition}
Under Assumption \ref{assump1}, we have the following results on $f_n$ and $f$ (with probability 1).
\begin{itemize}
    \item[(a)] The optimal solution set of problem (\ref{newDual}) is identical to the optimal solution set of problem (\ref{dualLP}).
    \item[(b)] Both $f_n(\p)$ and $f(\p)$ are convex.
    \item[(c)] The optimal solutions $\p_n^*$ and $\p^*$ satisfy
    $$\Dd^\top \p_n^* \le \bar{r},$$
    $$\Dd^\top \p^* \le \bar{r}.$$
\end{itemize}
\label{basicProp}
\end{proposition}

Given the boundedness of the optimal solutions, we define
$$\Omega_{p} \coloneqq \left\{\p\in \mathbb{R}^m: \bm{p} \ge \bm{0}, \bm{e}^\top \p \le \frac{\bar{r}}{\underline{d}}\right\}$$ 
where $\bm{e} \in \mathbb{R}^m$ is an all-one vector.
We know that $\Omega_p$ covers all possible optimal solutions to (\ref{new2Dual}) and (\ref{asymProblem}). 

Next, we introduce the second group of assumptions on the distribution $\mathcal{P}$. Here and hereafter, $I(\cdot)$ denotes the indicator function.

\begin{assumption}[Non-degeneracy]
We assume
\begin{itemize}
    \item[(a)] The second-order moment matrix $\bm{M} \coloneqq \E_{(r,\Aa)\sim \mathcal{P}}[\Aa \Aa^\top]$ is positive-definite. Denote its minimum eigenvalue with $\lambda_{\min}.$
    \item[(b)] There exist constants $\lambda$ and $\mu$ such that if $(r,\bm{a})\sim \mathcal{P}$,
    $$ \lambda |\bm{a}^\top\p - \bm{a}^\top\p^*| \le \left|\prob(r>\bm{a}^\top \bm{p}\vert\bm{a}) - \prob(r>\bm{a}^\top\bm{p}^*\vert\bm{a})\right| \le \mu |\bm{a}^\top\p - \bm{a}^\top\p^*|$$
    holds for any $\bm{p} \in \Omega_{p}$.
    \item[(c)] The optimal solution $\p^*$ to the stochastic optimization problem (\ref{asymProblem}) satisfies $p^*_i = 0$ if and only if $d_i - \E_{(r,\Aa)\sim \mathcal{P}}[a_iI(r>\Aa^\top \p^*)]>0.$
\label{assump2}
\end{itemize}
\end{assumption}

Assumption \ref{assump2}  (a) is mild in that the matrix $\E[\Aa \Aa^\top]$ is positive semi-definite by definition; the positive definiteness holds as long as the constraint matrix $\bm{A}$ of the LP always has a full row rank, which is a typical assumption for solving linear programs. Assumption \ref{assump2}  (b) states that the cumulative distribution function of $r|\Aa$ should not grow too fast or too slowly. Assumption \ref{assump2} (c) imposes a strict complementarity for the stochastic program. Essentially, Assumption \ref{assump2} altogether imposes a non-degeneracy condition for both the primal and dual LPs. It can be viewed as a generalization of the non-degeneracy condition in \citep{jasin2012re, jasin2015performance} and as a stochastic version of the general position condition in \citep{devanur2009adwords, agrawal2014dynamic}.

We remark that all three parts of Assumption \ref{assump2} are crucial for the analyses in the rest of the paper. While Assumption \ref{assump2} (a) and (c) are not necessarily true for all the stochastic programs, a slight perturbation of the distribution $\mathcal{P}$ (for example, through adding small random noises to $r_j$ and $a_{ij}$) would result in them being satisfied. In addition, we provide three examples that satisfy Assumption \ref{assump2} (b). In Example 1, we can simply choose $\lambda = \underline{\alpha}$ and $\mu=\bar{\alpha}$ and then Assumption \ref{assump2} (b) is satisfied. The analyses of Example 2 and Example 3 are postponed to Section \ref{egAssump}.

\begin{example}
Consider a multi-secretary problem where $m=1$ and all the constraint coefficients $a_{1j}=1$ for $j=1,...,n$. The reward $r_j$ is a continuous random variable and its distribution $\mathcal{P}_r$ has a density function $f_{r}(x)$ s.t. $\underline{\alpha} \le f_r(x)\le \bar{\alpha}$ for $x \in [0,1].$ Here $\bar{\alpha}\ge\underline{\alpha}\ge 0$.
\end{example}

\begin{example}
Consider $(r, \Aa)\sim\mathcal{P}$ such that $r=\Aa^\top \p^* + \epsilon$ where $\epsilon$ is a continuous random variable independent with $\bm{a}$ and with bounded support. In addition, the distribution $\mathcal{P}_\epsilon$ of $\epsilon$'s has a density function $f_\epsilon(x)$ such that there exists $\bar{\alpha},\underline{\alpha},c_\epsilon>0$ such that $f_\epsilon(x)\ge \underline{\alpha}$ for $x\in[-c_\epsilon,c_\epsilon]$ and $f_\epsilon(x)\le \bar{\alpha}$ for all $x$. 
\end{example}

\begin{example}
Consider $(r, \Aa)\sim\mathcal{P}$ such that the conditional distribution $r|\Aa$ has a density function $f_{r|\Aa}(x)$ and the density function satisfies $\underline{\alpha} \le f_{r|\Aa}(x)\le \bar{\alpha}$ for $x \in [-\bar{r},\bar{r}]$ with $\underline{\alpha}, \bar{\alpha}>0,$ and $f_{r|\Aa}(x)=0$ for $x \notin [-\bar{r},\bar{r}]$. In addition, there exists $\delta_r>0$ such that $\bm{a}^\top \bm{p}^*\in [-\bar{r}+\delta_r,\bar{r}-\delta_r]$ almost surely.
\end{example}

According to Assumption 2 (c), we define two index sets
$$I_B \coloneqq \left\{i: d_i - \E_{(r,\Aa)\sim \mathcal{P}}[a_iI(r>\Aa^\top \p^*)] = 0\right\},$$
$$I_N \coloneqq \left\{i : d_i - \E_{(r,\Aa)\sim \mathcal{P}}[a_iI(r>\Aa^\top \p^*)] > 0\right\},$$
where the subscripts ``B'' and ``N'' are short for binding and non-binding, respectively. Assumption \ref{assump2} (c) implies $I_B \cap I_N = \varnothing$ and $I_B \cup I_N = \{1,...,m\}.$ Throughout the paper, the binding and non-binding constraints of the OLP problem are always defined according to the stochastic program (\ref{asymProblem}).

Now, we derive more basic properties for the stochastic program (\ref{asymProblem}) based on Assumption \ref{assump2}. First, define a function $h: \R^{m} \times \R^{m+1} \rightarrow \R$,
$$h(\bm{p}, \bm{u}) \coloneqq \sum_{i=1}^m d_ip_i + \left(u_0-\sum_{i=1}^m u_{i}p_i\right)^+$$
and function $\phi: \R^{m} \times  \R^{m+1} \rightarrow \R^m$,
$$\phi(\p, \bm{u}) \coloneqq \frac{\partial h(\p, \bm{u})}{\partial p} = (d_1,...,d_m)^\top - (u_1,...,u_m)^\top \cdot I\left(u_0>\sum_{i=1}^m u_{i}p_i\right)$$
where $\bm{u}=(u_0,u_1,...,u_m)^\top$ and $\bm{p} = (p_1,...,p_m)^\top$. The function $\phi$ is the partial \textit{sub-gradient} of the function $h$ with respect to $\bm{p}$; in particular, $\phi(\bm{p},\bm{u})=\bm{d}$ when $u_0=\sum_{i=1}^m u_ip_i$. We know that
$$f(\bm{p}) = \E_{\bm{u} \sim \mathcal{P}} [h(\p,\bm{u})].$$
and we define 
\begin{equation}
    \nabla f(\p)\coloneqq \E\left[ \phi(\p,\bm{u})\right]
    \label{f_sub_gradient}
\end{equation}
where both expectations are taken with respect to $\bm{u}=(r,\Aa)\sim \mathcal{P}.$ The caveat is that the function $f$ is not necessarily differentiable under our current assumptions. As we will see shortly in Lemma \ref{TaylorLemma0} and Proposition \ref{strConvex}, the definition of $\nabla f(\p)$ constitutes a meaning of sub-gradient for function $f$.

Lemma \ref{TaylorLemma0} represents the difference between $f(\p)$ and $f(\bm{p}^*)$ with the sub-gradient function $\phi$. Note that the identity (\ref{identity0}) holds regardless of the distribution $\mathcal{P}$. Its derivation shares the same idea with the Knight's identity in \citep{knight1998limiting}. Intuitively, the lemma can be viewed as a second-order Taylor expansion for the function $f.$ The first and second term on the right-hand side of the identity can be interpreted as the first- and second- order term in Taylor expansion. They are disguised in this special form due to the non-differentiability of the positive part function at the origin.

\begin{lemma}
For any $\bm{p}\ge \bm{0},$ we have the following identity,
\begin{align}
    f(\p)- f(\p^*) & = \underbrace{ \nabla f(\bm{p}^*)(\p-\p^*)}_\text{First-order} +  \underbrace{\E\left[\int_{\bm{a}^\top \p}^{\bm{a}^\top \p^*}\left(I(r>v) - I(r>\bm{a}^\top\p^*) \right)dv\right]}_\text{Second-order}.
    \label{identity0}
\end{align}
where the expectation is taken with respect to $(r,\bm{a})\sim\mathcal{P}.$
\label{TaylorLemma0}
\end{lemma}

Proposition \ref{strConvex} applies Assumption \ref{assump2} to the identity (\ref{identity0}) and it leads to a local property around $\p^*$ for the function $f(\bm{p})$. Our motivation for the proposition is to provide more intuitions for Assumption \ref{assump2} from a technical perspective. Moreover, the proposition asserts the uniqueness of $\bm{p}^*$ which makes our notion of convergence to $\bm{p}^*$ well-defined. 

\begin{proposition}[Growth and Smoothness of $f(\bm{p})$]
Under Assumption \ref{assump1} and \ref{assump2}, for $\bm{p}\in \Omega_p,$ 
\begin{equation}
    \frac{\lambda\lambda_{\min}}{2} \|\p-\p^*\|_2^2 \le f(\p)- f(\p^*) -\nabla f(\bm{p}^*)(\p-\p^*)
\le\frac{\mu\bar{a}^2}{2} \|\p-\p^*\|_2^2. \label{twoSide}
\end{equation}
Moreover, the optimal solution $\bm{p}^*$ to the stochastic program (\ref{asymProblem}) is unique.
\label{strConvex}
\end{proposition}

The proposition gives a technical interpretation for the distributional conditions in Assumption \ref{assump2}. Essentially, the role of the assumption on the distribution $\mathcal{P}$ is to impose a locally strong convexity and smoothness around $\bm{p}^*$. This is weaker than the classic notion of strong convexity and smoothness for convex functions which requires the inequality (\ref{twoSide}) to hold globally. Assumption \ref{assump2} (b) and Proposition \ref{strConvex} both concern a local property for the optimal solution $\bm{p}^*.$ As we will see in the later chapter, this local property on $f(\bm{p})$ is crucial and sufficient to ensure a fast convergence rate of $\bm{p}_n^*$ and a sharp regret bound for OLP algorithms.

We use the notation $\Xi$ to denote the family of distributions that satisfy Assumption \ref{assump1} and \ref{assump2}. In the rest of the paper except for Section \ref{AHDLA}, all the theoretical results on the dual convergence and the analyses of OLP algorithms are established under Assumption \ref{assump1} and \ref{assump2}. In Section \ref{AHDLA}, we will present a stronger version of Assumption \ref{assump2} and analyze the action-history-dependent algorithm accordingly.

\subsection{Dual Convergence}

Now, we discuss the convergence of $\bm{p}_n^*$ to $\bm{p}^*$. 
First, the SAA function $f_n(\bm{p})$ can be expressed by
$$f_n(\bm{p}) = \frac{1}{n}\sum_{j=1}^n h(\bm{p}, \bm{u}_j)$$
where $\bm{u}_j = (r_j, \bm{a}_{j})$ and the function $h$ is as defined earlier. Lemma \ref{TaylorLemma} is a sample average version of Lemma \ref{TaylorLemma0} and it represents the difference between $f_n(\p)$ and $f_n(\p^*)$ with the sub-gradient function $\phi$. 

\begin{lemma}
For any $\bm{p} \in \mathbb{R}^m,$ we have the following identity,
\begin{align}
    f_n(\p)- f_n(\p^*) & = \underbrace{ \frac{1}{n}\sum_{j=1}^n \phi(\p^*,\bm{u}_j)^\top(\p-\p^*)}_\text{First-order} +  \underbrace{\frac{1}{n} \sum_{j=1}^n \int_{\bm{a}_j^\top \p}^{\bm{a}_j^\top \p^*}\left(I(r_j>v) - I(r_j>\bm{a}_j^\top\p^*) \right)dv}_\text{Second-order}.
    \label{identity}
\end{align}
\label{TaylorLemma}
\end{lemma}
In the following, we will establish the convergence of $\p_n^*$ based on an analysis of the identity (\ref{identity}). The idea is to show that the right-hand-side of (\ref{identity}) concentrates around its expectation as on the right-hand-side of (\ref{identity0}). The following two propositions analyze the first-order and second-order terms respectively. Proposition \ref{pro:gradient} tells that the sample average sub-gradient $\frac{1}{n}\sum_{j=1}^n \phi(\p^*,\bm{u}_j)$ stays close to its expectation $\nabla f(\p)$ \eqref{f_sub_gradient} evaluated at $\bm{p}^*$ with high probability. Proposition \ref{pro:Hessian} states that the second-order term -- the integral on the right hand of (\ref{identity}), is uniformly lower bounded by a strongly convex quadratic function with high probability. Intuitively, the analysis only involves the local property of the function $f_n$ around $\bm{p}^*.$ This explains why we impose only local (but not global) conditions on the function $f$ in Assumption \ref{assump2}. Specifically, Assumption \ref{assump2} (a) and (b) concern the Hessian, while Assumption \ref{assump2} (c) concerns the gradient, with both being evaluated at $\bm{p}^*.$

\begin{proposition} 
\label{pro:gradient}
We have 
$$\prob\left(\Bigg\|\frac{1}{n}\sum_{j=1}^n \phi(\p^*,\bm{u}_j) - \nabla f(\p^*)\Bigg\|_2 \le \epsilon \right) \ge 1-2m \exp\left(-\frac{n\epsilon^2}{2\bar{a}^2m}\right)$$
hold for any $\epsilon>0$, all $n>m$ and $\mathcal{P} \in \Xi$. 
\end{proposition}

Proposition \ref{pro:gradient} is obtained by a direct application of a concentration inequality. Notably, the probability bound on right-hand-side is not dependent on the distribution $\mathcal{P}$ and the inequality holds for any $\epsilon>0.$ From the optimality condition of the stochastic program, we know $\left(\nabla f(\p^*)\right)_i= 0$ for $i \in I_B$ and $\left(\nabla f(\p^*)\right)_i > 0$ for $i \in I_N$. Therefore, the proposition implies that the sample average sub-gradient (first-order term in (\ref{identity})) concentrates around zero for binding dimensions and concentrates around a positive value for non-binding dimensions. As noted earlier, the binding and non-binding dimensions are defined by the stochastic program (\ref{asymProblem}).

\begin{proposition}
\label{pro:Hessian}
We have 
\begin{align}
  & \prob\Bigg(\frac{1}{n} \sum_{j=1}^n \int_{\bm{a}_j^\top \p}^{\bm{a}_j^\top \p^*}\left(I(r_j>v)-I(r_j>\bm{a}_j^\top\p^*)\right)dv \ge -\epsilon^2 -2\epsilon\bar{a}\|\p^*-{\p}\|_2 + \frac{\lambda\lambda_{\min}}{32} \left\|\p^* - {\p}\right\|_2^2 \label{second_order_p}
  \\
  & \hspace{4.5cm} \text{ \ for all \ } \p \in \Omega_p\Bigg)  
  \ge 1- m\exp\left(-\frac{n\lambda_{\min}^2}{4\bar{a}^2}\right) - 2\left(2N\right)^m \cdot \exp\left(-\frac{n\epsilon^2}{2}\right)  \nonumber 
\end{align}
holds for any $\epsilon>0,$ $n>m$ and $\mathcal{P}\in \Xi.$ Here
$$N = \Bfloor{\log_q \left(\frac{\underline{d}\epsilon^2}{\bar{a}\bar{r}\sqrt{m}}\right)}+1, \ \ \ q = \max \left\{\frac{1}{1+\frac{1}{\sqrt{m}}} , \frac{1}{1+\frac{1}{\sqrt{m}}\left(\frac{\lambda\lambda_{\min}}{8\mu \bar{a}^2}\right)^{\frac{1}{3}}} \right\}$$
where $\floor{\cdot}$ is the floor function. 
\end{proposition}

Proposition \ref{pro:Hessian} discusses the second-order term in (\ref{identity}). The inequality (\ref{second_order_p}) tells that the second-order term is uniformly lower bounded by a quadratic function with high probability. To prove the inequality for a fixed $\p$ can be easily done by a concentration argument as in Proposition \ref{pro:gradient}. The challenging part is to show that the inequality (\ref{second_order_p}) holds uniformly for all $\p\in \Omega_p.$ The idea here is to find a collection of sets that covers $\Omega_p$ and then analyze each covering set separately. We utilize the same covering scheme as in \citep{huber1967behavior} which is originally developed to analyze the consistency of non-standard maximum likelihood estimators. The advantage of this covering scheme is that it provides a tighter probability bound than the traditional $\epsilon$-covering scheme. The parameters involved in the proposition are dependent on the parameters in Assumption \ref{assump1} and \ref{assump2}: $\bar{a}, \bar{r}$ and $\underline{d}$ are specified in Assumption \ref{assump1}; $\lambda_{\min},\lambda$ and $\mu$ are specified in Assumption \ref{assump2}. For the newly introduced parameters, $q$ can be viewed as a constant parameter and $N$ is on the order of $\sqrt{m}\log m$. All of these parameters are not dependent on the specific distribution $\mathcal{P}$ and thus the result creates convenience for our later regret analysis of OLP algorithms. Importantly, both probability bounds in Proposition \ref{pro:gradient} and \ref{pro:Hessian} have an exponential term of $n$, and we  utilize this fact to establish the convergence rate of $\p_n^*$ as follows. 

With Proposition \ref{pro:gradient} and \ref{pro:Hessian}, the identity (\ref{identity}) can be written heuristically as
\begin{align}
    f_n(\p) - f_n(\p^*) & \ge \nabla f(\p^*)^\top (\p-\p^*) - \epsilon\|\p^*-{\p}\|_2 -\epsilon^2 -2\epsilon\bar{a}\|\p^*-{\p}\|_2 + \frac{\lambda\lambda_{\min}}{32} \left\|\p^* - {\p}\right\|_2^2  \nonumber \\
    & \ge -\epsilon^2 -(2\bar{a}+1)\epsilon\|\p^*-{\p}\|_2 + \frac{\lambda\lambda_{\min}}{32} \left\|\p^* - {\p}\right\|_2^2  \text{ \ uniformly for all \ $\p\in \Omega_p$} \label{heu}
\end{align}
with high probability, that is, \eqref{heu} is violated with a probability decreasing exponentially in $\epsilon$. Note that the right side in above is a quadratic function of $\p$. From the optimality of $\p_n^*$,
\begin{equation}
    f_n(\p_n^*)\le f_n(\p^*).
    \label{optP}
\end{equation}
By putting together (\ref{heu}) with (\ref{optP}) and then integrating with respect to $\epsilon$, we can obtain the following theorem. 



\begin{theorem}
Under Assumption \ref{assump1} and \ref{assump2}, there exists a constant $C$ such that 
\begin{gather*}
    \E\left[ \|\p^*_n - \p^*\|_2^2 \right] \le \frac{Cm \log m \log\log n}{n}
\end{gather*}
holds for all $n\ge \max\{m,3\}$, $m\ge2,$ and distribution $\mathcal{P} \in \Xi.$ Additionally,
$$\E \left[\|\p^*_n - \p^*\|_2 \right]\le C \sqrt{\frac{m \log m \log\log n}{n}}.$$ For the case of $m=1$, both results still hold without the $\log m$ term on the right-hand-side.
\label{expectation}
\end{theorem}

Theorem \ref{expectation} states the convergence rate in terms of the L$_2$ distance. The second inequality in the theorem can be directly implied from the first one, and here we present both inequalities for later usage. Also, to simplify the notations in the rest of the paper, we adopt the same constant for both inequalities without loss of generality. The theorem formally connects the dual LP and the stochastic program, by characterizing the distance between the optimal solutions of the two problems. As mentioned earlier, it resolves the open question of the convergence of dual optimal solutions in the OLP problem.  The conclusion and the proof of the theorem are all based on finite-sample argument, in parallel with the classic stochastic programming literature where asymptotic results are derived \citep{shapiro1993asymptotic, shapiro2009lectures}. The finite-sample property is crucial in the regret analysis of the OLP algorithms. In the Theorem \ref{representation} of the next section, we will show how the regret of an OLP algorithm can be upper bounded by the L$_2$ approximation error of $\bm{p}^*.$ This explains why we focus on the convergence and the approximation error of $\p_n^*$ instead of the function value $f(\p_n^*).$
As for the convergence rate, recall that Assumption \ref{assump2} (a) and (b) impose a curvature condition around the optimal solution $\p^*$. They are critical to the quadratic form (\ref{second_order_p}) in Proposition \ref{pro:Hessian} and consequently the convergence rate in Theorem \ref{expectation}. An alternative proof of the convergence results in Theorem \ref{expectation} may also be obtained by using the notion of \textit{local Rademacher Complexity}. \cite{bartlett2005local} proposed this local and empirical version of Rademacher complexity and employed the notion to derive the fast rate of estimator convergence. We remark that this possible alternative treatment also requires certain growth and smoothness condition around the optimal solution $\p^*$ as Assumption \ref{assump2} (a) and (b). 



Now, we discuss some implications of Theorem \ref{expectation} on the OLP problem. The theorem tells that the sequence of $\p_n^*$ will converge to $\p^*$, and that when $n$ gets large, $\p^*_n$ stays closer to $\p^*$. From an algorithmic perspective, an OLP algorithm can form an approximate for $\p^*$ based on the observed $t$ inputs at each time $t$. If $t$ is sufficiently large, the approximate dual price $\p_t^*$ should be very close to both $\p^*$ and $\p^*_n.$ Then, the algorithm can use $\p^*_t$ (as an approximate to the optimal $\p_n^*$) to decide the value of the current decision variable. This explains why OLP algorithms in literature always solve a scaled version of the offline LP (based on the observations available at time $t$). However, in the literature, due to the lack of convergence knowledge of the dual optimal solutions, papers devised other approaches to analyze the online objective value. Our convergence result explicitly characterizes the rate of the convergence and thus provides a powerful and natural instrument for the theoretical analysis of the online algorithms.

The dual convergence result also contributes to the literature of approximate algorithms for large-scale LPs.  Specifically, we can perform  one-time learning with the first $t$ inputs and then use $\p_t^*$ as an approximation for $\p_n^*$. In this way, we obtain an approximate algorithm for solving the original LP problem by only accessing the first $t$ columns $\{(r_j,\Aa_j)\}_{j=1}^t.$ The approximate algorithm can be viewed as a constraint sampling procedure \citep{de2004constraint, lakshminarayanan2017linearly} for the dual LP. It also complements the recent work \citep{vu2018random} in which approximate algorithms for large-scale LP are developed under certain statistical assumptions on the coefficients of the LP problem.

\section{Learning Algorithms for OLP}

\label{Algorithms}
\subsection{Dual-based Policy, Constraint Process, and Stopping Time}
\label{thresPolicy}
In this section, we present several online algorithms based on the dual convergence results. We first revisit the definition of online policies and narrow down our scope to a class of \textit{dual-based policies} that rely on the dual solutions. Specifically, at each time $t$, a vector $\p_t$ is computed based on historical data
$$\p_t = h_t(\mathcal{H}_{t-1})$$
where $\mathcal{H}_{t-1} = \{r_j, \bm{a}_j, x_j\}_{j=1}^{t-1}$.
Inspired by the optimality condition of the offline/static LP,  we attempt to set 
$$\tilde{x}_{t} = \begin{cases} 1, & \text{ if } r_{t} > \bm{a}_{t}^\top \bm{p}_t,\\
0, & \text{ if } r_{t} \le \bm{a}_{t}^\top \bm{p}_t.
\end{cases}$$ 
In other words, a threshold is set by the dual price vector $\bm{p}_t$. If the reward $r_t$ is larger than the threshold, we intend to accept the order. Then, we check the constraints satisfaction and assign
$${x}_{t} = \begin{cases} \tilde{x}_{t}, & \text{ if } \sum_{j=1}^{t-1} a_{ij}{x}_j + a_{it}\tilde{x}_t\le b_i,\ \text{ for } i=1,...,m,\\
0, & \text{ otherwise}.
\end{cases}$$ 
We formally define policies with this structure as a dual-based policy. We emphasize that in this dual-based policy class, $\p_t$ is first computed based on history (up to time $t-1$), and then $(r_t,\Aa_t)$ is observed. This creates a natural conditional independence $$(\tilde{x}_t, r_t,\Aa_t) \perp \mathcal{H}_{t-1} \big| \p_t.$$ 
This matches the setting in online convex optimization where at each time $t$, the online player makes her decision before we observe the function $f_t$ (See \citep{hazan2016introduction}). We will frequently resort to this conditional independence in the regret analysis. In this policy class, an online policy $\bm{\pi}$ could be fully specified by the sequence of mappings $h_{t}$'s, i.e., $\bm{\pi} = (h_1,...,h_{n}).$ To facilitate our analysis, we introduce the constraint process $\left\{\bm{b}_{t}\right\}_{t=1}^n$,
$$\bm{b}_{0} = \bm{b} = n\bm{d}$$
$$\bm{b}_{t} = \bm{b}_{t-1} - \Aa_t x_t.$$
In this way, $\bm{b}_{t} = \left(b_{1t},...,b_{mt}\right)^\top $ represents the vector of remaining resources at the end of the $t$-th period. In particular, $\bm{b}_{n} = \left(b_{1n},...,b_{mn}\right)^\top $ represents the remaining resources at the end of the horizon. By the definition of OLP, $\bm{b}_t \ge 0$ for $t=1,...,n.$ Also, the process of $\left\{\bm{b}_{t}\right\}_{t=0}^n$ is pertaining to the policy $\pi$. Based on the constraint process, we define 
$$\tau_s \coloneqq \min \{n\} \cup \left\{t\ge 1: \min_i b_{it} < s\right\}$$
for $s>0.$
In this way, $\tau_s$ denotes the first time that there are less than $s$ units for some type of constraints. Precisely, $\tau_s$ is a stopping time adapted to the process $\left\{\bm{b}_{t}\right\}_{t=1}^n$. Similar to the process $\bm{b}_t$, the stopping time $\tau_{s}$ is also pertaining to the policy $\pi.$ When executing an online policy, we do not close the business at the first time that some constraints are violated. This is because we are considering double-sided problems that include both buying and selling orders. If a certain type of resource is exhausted, we may accept selling orders containing that resource as a way of replenishment. We will see that a careful design of the algorithm will ensure the constraint violation/resource depletion only happens at the very end. So, the decisions afterward will not affect the cumulative revenue significantly. 

In the rest of this section, we first derive a generic upper bound for the regret of OLP algorithms and then present three OLP algorithms. These three algorithms all belong to the dual-based policy class, and their regret analyses all rely on the dual convergence and the generic upper bound. We restrict our attention to large-$n$ and small-$m$ setting, and the regret bounds will be presented with big-O notation which treats $m$ and the parameters in Assumption \ref{assump1} and \ref{assump2} as constants.

\subsection{Upper Bound for OLP Regret}

We first construct an upper bound for the offline optimal objective value. Consider the optimization problem
\begin{align} 
\max_{\p\ge \bm{0}} \ \ \ \ & \E \left[r I(r> \Aa^\top \p)\right] \label{detRelax} \\
\text{s.t. \ } & \E\left[\Aa I(r>\Aa^\top \p)\right] \le \Dd \nonumber
\end{align}
where the expectation is taken with respect to $(r, \bm{a}) \sim \mathcal{P}.$
There are two ways to interpret this optimization problem. On one hand, we can interpret this problem as a ``deterministic'' relaxation of the primal LP (\ref{primalLP}). We substitute both the objective and constraints of (\ref{primalLP}) with an expectation form expressed in dual variable $\p.$ On the other hand, we can view this optimization problem as the primal problem of the stochastic program (\ref{asymProblem}). The consideration of a deterministic form for an online decision making problem has appeared widely in the literature of network revenue management \citep{talluri1998analysis, jasin2013analysis, bumpensanti2018re}, dynamic pricing \citep{besbes2009dynamic, wang2014close, lei2014near, chen2018primal}, and bandits problem \citep{wu2015algorithms}. The idea is that when analyzing the regret of an online algorithm in such problems, the offline optimal value usually does not have a tractable form (such as the primal LP problem (\ref{primalLP})). The deterministic formulation serves as a tractable upper bound for the offline optimal value, and then the gap between the deterministic optimal and the online objective values is an upper bound for the regret of the online algorithm. Different from the literature, we consider the Lagrangian of the deterministic formulation to remove the constraints. Specifically,  define
$$g(\p) \coloneqq \E\left[r I(r> \Aa^\top \p) + \left(\Dd-\Aa I(r>\Aa^\top \p)\right)^\top \p^*\right],$$
where the expectation is taken with respect to $(r, \bm{a})\sim \mathcal{P}$ and $\p^*$ is the optimal solution to the stochastic program (\ref{asymProblem}).
We can view $g(\p)$ as the Lagrangian of the optimization problem (\ref{detRelax}) with a specification of the multiplier by $\p^*.$ Lemma \ref{gP} establishes that the deterministic formulation indeed provides an upper bound for the offline optimal value and that the optimization problems (\ref{asymProblem}) and (\ref{detRelax}) share the same optimal solution.
\begin{lemma}
Under Assumption \ref{assump1} and \ref{assump2}, we have 
$$\E R_n^* \le ng(\p^*)$$
$$g(\p^*) \ge g(\p)$$
for any $\p\ge0.$ Here $\p^*$ is the optimal solution to the stochastic program (\ref{asymProblem}). Additionally,
\begin{equation}
    g(\p^*)- g(\p) \le \mu\bar{a}^2 \|\p^* - \p\|_2^2
    \label{gapG}
\end{equation}
holds for all $\p \in \Omega_p$ and all the distribution $\mathcal{P}\in \Xi.$ 
\label{gP}
\end{lemma}

The presence of constraints makes the regret analysis challenging, because the way the constraints affect the objective value in an online setting is elusive and problem-dependent. The Lagrangian form $g(\p)$ resolves the issue by incorporating the constraints into the objective. Intuitively, it assigns a cost to the constraint consumption and thus unifies the two seemingly conflicting sides -- revenue maximization and constraint satisfaction. The importance of Lemma \ref{gP} is two-fold. First, it provides a deterministic upper bound for the expected offline optimal objective value. The upper bound is not dependent on the realization of a specific OLP instance, so it is more convenient to analyze than the original offline optimal objective value. Second, a dual-based online algorithm employs a dual price $\p_t$ at time $t$ and its instant reward at time $t$ can be approximated by $g(\p_t)$. Then the single-step regret of the algorithm at time $t$ can be upper bounded with (\ref{gapG}).
Theorem \ref{representation} builds upon Lemma \ref{gP} and compares the online objective value $\E R_n(\pi)$ against $ng(\p^*)$. A generic upper bound is developed for dual-based online policies. The upper bound consists of three components: (i) the cumulative ``approximation'' error, (ii) the remaining periods after the constraint is almost exhausted, and (iii) the remaining resources at the end of time period $n$. The first component relates the regret with $\E\left[\|\p_t-\p^*\|_2^2\right]$ studied in the last section. It justifies why the dual convergence (Theorem \ref{expectation}) is studied in the last section. The second component concerns $\tau_{\bar{a}}$, where $\bar{a}$ is defined in Assumption \ref{assump1} and it is the maximal possible constraint consumption per period. The intuition for $\tau_{\bar{a}}$ is that an order $(r_t, \Aa_t)$, if necessary, can always be fulfilled when $t\le \tau_{\bar{a}}.$ Though selling order (with negative $a_{ij}$'s) may still be accepted after $\tau_{\bar{a}}$, from the standpoint of deriving regret upper bound, we simply ignore all the orders that come after $\tau_{\bar{a}}$. The third component considers the constraint leftovers for binding constraints. Intuitively, binding constraints are the bottleneck for producing revenue, so wasting those resources at the end of the horizon will induce a cost. 

\begin{theorem}
Under Assumption \ref{assump1} and \ref{assump2}, there exists a constant $K$ such that the worst-case regret under policy $\bm{\pi}$,
$$\Delta_{n}(\pi) \le K\cdot \E\left[\sum_{t=1}^{\tau_{\bar{a}}} \|\p_t - \p^* \|_2^2 + (n-\tau_{\bar{a}}) + \sum_{i\in I_B} b_{in} \right]$$
holds for all $n>0.$ Here $I_B$ is the set of binding constraints specified by the stochastic program (\ref{asymProblem}), $\p_t$ is specified by the policy $\bm{\pi}$, and $\p^*$ is the optimal solution of the stochastic program (\ref{asymProblem}).
\label{representation}
\end{theorem}

Theorem \ref{representation} provides important insights for the design of an online policy. First, a dual-based policy should learn $\p^*$. Meanwhile, the online policy should have stable control of the resource/constraint consumption. Exhausting the constraints too early may result in a large value of the term $n-\tau_{\bar{a}}$ while the remaining resources at the end of the horizon may also induce regret through the term $\sum_{i\in I_B} b_{in}$. Essentially, both components stem from the fluctuation of the constraint consumption. The ideal case, although not possible due to the randomness, is that the $i$-th constraint is consumed by exactly $d_i$ units in each time period. The following corollary states that the result in Theorem \ref{representation} holds for a general class of stopping times.

\begin{corollary}
For any given $\Bb_t$-adapted stopping time $\tau,$ if $\prob(\tau \le \tau_{\bar{a}})=1$, 
$$\Delta_{n}(\bm{\pi}) \le K \cdot \E\left[\sum_{t=1}^{\tau-1} \|\p_t - \p^* \|_2^2 + (n-\tau) + \sum_{i\in I_B} b_{in}\right]$$
holds for all $n>0$ with the same constant $K$ as in Theorem \ref{representation}. Here $I_B$ is the set of binding constraints, $\p_t$ is specified by the policy $\bm{\pi}$, and $\p^*$ is the optimal solution of the stochastic program (\ref{asymProblem}). 
\label{coro1}
\end{corollary}

\textbf{Remark.} Theorem \ref{representation} and Corollary \ref{coro1} reduce the derivation of regret upper bound to the analysis of approximation errors and the analysis of the constraint process. To highlight the usefulness of this reduction, we briefly discuss the existing techniques developed for analyzing online learning problems with the presence of constraints. The simplest approach is to propose a bi-objective performance measure: one for the objective value and one for the constraint violation. The bi-objective performance measure is adopted in the problem of online convex optimization with constraints \citep{mahdavi2012trading, agrawal2014fast, yu2017online, yuan2018online}. 
In these works, the regret for the objective value and the constraint violation are reported separately. Our results provide a method to convert the bi-objective results into one performance measure. When the upper bound on the constraint coefficients $\bar{a}$ is known, one approach to combine the bi-objective performance measure is to penalize the excessive usage of the constraints, such as the regret analysis in \citep{ferreira2018online}. However, this approach only works for non-adaptive algorithms (like Algorithm \ref{alg:IDLA}) but fails when the constraint process affects the decision (like Algorithm \ref{alg:HDLA} and other re-solving algorithms). In this light, Theorem \ref{representation} and Corollary \ref{coro1} provide a unified treatment for non-adaptive and adaptive algorithms.  Another approach to deal with the constraints is to use the ``shrinkage'' trick. The idea is to perform the online learning as if the constraint is shrunk by a factor of $1-\epsilon$, and then the output online solution will be feasible with high probability for the original problem. This shrinkage trick is used in the OLP literature \citep{agrawal2014dynamic, kesselheim2014primal} and also in the bandits with knapsacks problem \citep{agrawal2014bandits}. The downside of the shrinkage is that it will probably result in an over-conservative decision and may have too many remaining resources at the end of the horizon. The regret analysis closest to our approach is the analysis in \citep{jasin2012re, jasin2015performance, balseiro2019learning}. The regret derivation therein also involves analyzing the stopping time related to the constraint depletion. Our generic upper bound can be viewed as a generalization of their analysis for the case when no prior knowledge on the support of $(r_j, \bm{a}_j)$'s is available. For the network revenue management problem, the support of customer orders $(r_j, \bm{a}_j)$'s is assumed to be finite and known. In this case, the offline optimal solution can be explicitly specified by the optimal quantity of each type of customer orders that should be accepted. Then the performance of an online algorithm can be analyzed by the deviation from the optimal quantity. In a general OLP problem, there is no explicit representation of the offline optimal solution. Theorem \ref{representation} and Corollary \ref{coro1} resolve the issue by identifying the approximation error of $\bm{p}^*$ as a non-parametric generalization of the deviation from the optimal quantity in the network revenue management context. A subsequent work ({\color{blue}Lu et al., 2020}) also applied similar ideas from Theorem \ref{representation} and Corollary \ref{coro1} in their regret analysis for an online allocation problem which has a finite support of $(r_j, \bm{a}_j)$'s.

\subsection{When the Distribution is Known}

We first present an algorithm for the situation when we know the distribution $\mathcal{P}$ that generates the LP coefficients. With the knowledge of $\mathcal{P},$ the stochastic programming problem (\ref{asymProblem}) is well-specified. 
We study the algorithm mainly for benchmark purposes, so we do not discuss the practicability of knowing the distribution $\mathcal{P}$. Moreover, we assume that the stochastic programming problem (\ref{asymProblem}) can be solved exactly. In Algorithm \ref{alg:Distribution}, the optimal solution $\p^*$ is computed before the online procedure and can be viewed as prior knowledge. So, there is ``no need to learn'' the dual price throughout the procedure, and the pre-computed $\p^*$ can be used for thresholding rule. 
 
\begin{algorithm}[ht!]
\caption{No-Need-to-Learn Algorithm}\label{alg:Distribution}
\begin{algorithmic}[1]
\State Input:  $n$, $d_1,...,d_m$, Distribution $\mathcal{P}$
\State Compute the optimal solution of the stochastic programming problem
\begin{align*} \bm{p}^*&  = \argmin \bm{d}^\top \bm{p} + \E_{(r,\bm{a}) \sim \mathcal{P}} \left[(r-\bm{a}^\top \bm{p})^+\right] \\
& \text{s.t. \ }  \bm{p} \ge 0. 
\end{align*}
\For {$t=1,..., n$}
\State If constraints are not violated, choose
$${x}_{t} = \begin{cases} 1, & \text{ if } r_{t} > \bm{a}_{t}^\top \bm{p}^*\\
0, & \text{ if } r_{t} \le \bm{a}_{t}^\top \bm{p}^*
\end{cases}$$
\EndFor
\end{algorithmic}
\end{algorithm}

\begin{theorem}
With the online policy $\bm{\pi}_1$ specified by Algorithm \ref{alg:Distribution},
\begin{gather*}
\Delta_n(\bm{\pi}_1) \le O(\sqrt{n}).
\end{gather*}
\label{Algo3}
\end{theorem}
Theorem \ref{Algo3} tells that the worst-case regret under an online policy with the knowledge of full distribution is $O(\sqrt{n})$. Recall the generic regret upper bound in Theorem \ref{representation}, there is no approximation error for Algorithm \ref{alg:Distribution} because $\p^*$ is employed as the dual price. So the regret of Algorithm \ref{alg:Distribution} stems only from the fluctuation of the constraint process. Intuitively, the fluctuation is on the order of $O(\sqrt{n})$ when it approaches the end of the horizon. Our next algorithm and its according regret analysis show that the contribution of the approximation error of $\p^*$ (when the distribution is unknown) is indeed dominated by the regret caused by the fluctuation of the constraint process.


\subsection{Simplified Dynamic Learning Algorithm}

Algorithm \ref{alg:IDLA} approximates the dual price $\p^*$ in Algorithm \ref{alg:Distribution} by solving an SAA for the stochastic program based on the past observations. 
From the OLP perspective, it can be viewed as a simplified version of the dynamic learning algorithm in \citep{agrawal2014dynamic}. In Algorithm \ref{alg:IDLA}, the dual price vector $\p_t$ is updated only at geometric time intervals and it is computed based on solving the $t$-sample approximation, i.e., minimizing $f_t(\p)$. The key difference between this simplified algorithm and the dynamic learning algorithm in \citep{agrawal2014dynamic} is that we get rid of the shrinkage term $\left(1-\epsilon \sqrt{\frac{n}{t_k}}\right)$ in the constraints. Specifically, the algorithm in that paper considers $\left(1-\epsilon\sqrt{\frac{n}{t_k}}\right)t_kd_i$ on the right-hand side of the constraints in Step 6 of Algorithm \ref{alg:IDLA}. In the random input model, the shrinkage factor results in an over-estimated dual price $\p_t$ and hence will be more conservative in accepting orders. The conservativeness is probably helpful under the random permutation model studied in \citep{agrawal2014dynamic} but may cause the remaining resources to be linear in $n$ under the random input model. 

\begin{algorithm}[ht!]
\caption{Simplified Dynamic Learning Algorithm}\label{alg:IDLA}
\begin{algorithmic}[1]
\State Input: $d_1,...,d_m$ where $d_i=b_i/n$
\State Initialize: Find $\delta \in (1,2]$ and $L >0$ s.t. $\floor{\delta^L} = n.$
\State Let $t_k = \floor{\delta^k }, k = 1,2,...,L-1$ and $t_L=n+1$
\State Set $x_1 = ... = x_{t_1}=0$
\For {$k=1,2,...,L-1$}
\State Specify an optimization problem
\begin{align*}
\max\ & \sum_{j=1}^{t_k} r_j x_j 
    \\ 
    \text{s.t.}\ & \sum_{j=1}^{t_k} a_{ij}x_j \le t_kd_i ,\ \  i=1,...,m \\
    & 0 \le x_{j} \le 1, \ \  j=1,...,{t_k}
\end{align*}
\State Solve its dual problem and obtain the optimal dual variable $\bm{p}_k^*$
\begin{align*}
  \bm{p}_k^* &=  \argmin_{\bm{p}} \sum_{i=1}^m d_ip_i + \frac{1}{t_k}\sum_{j=1}^{t_k} \left(r_j-\sum_{i=1}^m a_{ij}p_i\right)^+   \\
   &\text{s.t.\ \ }  p_i \ge 0, \ \ i=1,...,m. \nonumber
\end{align*}
\For {$t = t_k+1,...,t_{k+1}$}
\State If constraints permit, set
$${x}_{t} = \begin{cases} 1, & \text{ if } r_{t} > \bm{a}_{t}^\top \bm{p}_k^*\\
0, & \text{ if } r_{t} \le \bm{a}_{t}^\top \bm{p}_k^*
\end{cases}$$
\State Otherwise, set $x_t = 0$
\State If $t=n$, stop the whole procedure. 
\EndFor
\EndFor
\end{algorithmic}
\end{algorithm}

\begin{theorem}
\label{Algo1}
With the online policy $\bm{\pi}_2$ specified by Algorithm \ref{alg:IDLA}, 
\begin{gather*}
\Delta_n(\bm{\pi}_2) \le O(\sqrt{n} \log n).
\end{gather*}
\end{theorem}
Theorem \ref{Algo1} tells that the policy incurs a worst-case regret of $O(\sqrt{n} \log n)$. Its proof relies on an analysis of the three components in the regret bound in Theorem \ref{representation}. The summation $\sum_{t=1}^{\tau_{\bar{a}}} \|\p_t-\p^*\|_2^2$ can be analyzed by the dual convergence result. Specifically, the dual price $\bm{p}_{t_k}$ in Algorithm \ref{alg:IDLA} is computed based on a $t_k$-sample approximation to the stochastic program (\ref{asymProblem}), and therefore Theorem \ref{expectation} can be employed to upper bound the distance $\|\p_{t_k}-\p^*\|_2^2$. It reiterates the importance of studying the dual convergence and expressing the approximation error in L$_2$ distance. The two components related to the stopping time and remaining resources are studied based on a careful analysis of the process $\bm{b}_t$. The detailed proof can be found in Section \ref{AsqrtT}. 



\subsection{Action-History-Dependent Learning Algorithm}
\label{AHDLA}

Now, we present our action-history-dependent learning algorithm. In Algorithm \ref{alg:IDLA}, the dual price $\p_t$ is a function of the past inputs $\{(r_j,\Aa_j)\}_{j=1}^{t-1}$ but it does not consider the past actions $(x_1,...,x_{t-1})$. In contrast, Algorithm \ref{alg:HDLA} integrates the past actions into the constraints of the optimization problem of $\p_t$. At the beginning of period $t+1$, the first $t$ inputs $\{(x_j, r_j, \Aa_j)\}_{j=1}^t$ are observed. Algorithm \ref{alg:IDLA} normalizes $b_{i}$ to $\frac{t}{n}b_{i} = td_i$ for the right-hand-side of the LP, while Algorithm \ref{alg:HDLA} normalizes the remaining resource $b_{it}$ for the right-hand-side of the LP (Step 6 of Algorithm \ref{alg:HDLA}). The intuition is that if we happen to consume too much resource in the past periods, the remaining resource $b_{it}$ will shrink, and Algorithm \ref{alg:HDLA} will accordingly push up the dual price and be more inclined to reject an order. On the contrary, if we happen to reject a lot of orders at the beginning and it results in too much remaining resource, the algorithm will lower down the dual price so as to accept more orders in the future. This pendulum-like design in Algorithm \ref{alg:HDLA} incorporates the past actions in computing dual prices indirectly through the remaining resources.

\begin{algorithm}[ht!]
\caption{Action-history-dependent Learning Algorithm}\label{alg:HDLA}
\begin{algorithmic}[1]
\State Input: $n$, $d_1,...,d_m$
\State Initialize the constraint $b_{i0} = nd_i$ for $i=1,...,m$
\State Initialize the dual price $\p_1 = \bm{0}.$
\For {$t=1,...,n$}
\State Observe $(r_t, \Aa_t)$ and set
$${x}_{t} = \begin{cases} 1, & \text{ if } r_{t} > \bm{a}_{t}^\top \bm{p}_t\\
0, & \text{ if } r_{t} \le \bm{a}_{t}^\top \bm{p}_t
\end{cases}$$
if the constraints are not violated
\vspace{0.4cm}

\State Update the constraint vector
$$b_{it} = b_{i,t-1}-a_{it}{x}_{t} \text{ \ for \ } i = 1,...,m$$
\vspace{0.1cm}

\State Specify an optimization problem
\begin{align*}
\max\ & \sum_{j=1}^{t} r_j x_j 
    \\ 
    \text{s.t.}\ & \sum_{j=1}^{t} a_{ij}x_j \le \frac{tb_{it}}{n-t} ,\ \  i=1,...,m \\
    & 0 \le x_{j} \le 1, \ \  j=1,...,{t}
\end{align*}
\State If $t<n,$ solve its dual problem and obtain the dual price $\bm{p}_{t+1}$
\begin{align*}
  \bm{p}_{t+1} &=  \argmin_{\bm{p}} \sum_{i=1}^m \frac{b_{it}p_i}{n-t} + \frac{1}{t}\sum_{j=1}^{t} \left(r_j-\sum_{i=1}^m a_{ij}p_i\right)^+   \\
    & \text{s.t.\ \ }  p_i \ge 0, \ \ i=1,...,m. \nonumber
\end{align*}
\EndFor
\end{algorithmic}
\end{algorithm}

From an algorithmic standpoint, Algorithm \ref{alg:HDLA} implements the re-solving technique in a learning environment, while the idea was implemented in a known-parameter environment by \citep{reiman2008asymptotically, jasin2012re, bumpensanti2018re}. Unlike the work \citep{jasin2015performance} which explicitly estimated the arrival intensity and fed the estimate into a certainty-equivalent problem, the learning part of Algorithm \ref{alg:HDLA} is implicit and integrated into the optimization part. For the optimization problem in Step 6 of Algorithm \ref{alg:HDLA}, the left-hand-side is specified by the history, while the right-hand-side is specified by the real-time constraint capacity. If we define $d_{it} = \frac{b_{it}}{n-t}$ as the average remaining resource capacity, then the optimization problem at time $t$ can be viewed as a $t$-sample approximation to a stochastic program specified by $\bm{d}_t=(d_{1t},...,d_{mt})^\top$. Importantly, the targeted stochastic program at each time period is dynamically changing according to the constraint process, while Algorithm \ref{alg:IDLA} and all the preceding analyses in this paper focus on a static stochastic program specified by the fixed initial $\bm{d}$. To apply the dual convergence result in a changing $\bm{d}$ setting, we need a uniform version of Assumption \ref{assump2}. Specifically, for $\bm{d}'=(d_1',...,d_m')^\top\in \Omega_d$, define $$f_{\bm{d}'}(\bm{p}) \coloneqq \bm{d}'^\top \bm{p} + \E_{(r,\bm{a})\sim\mathcal{P}}\left[(r-\bm{a}^\top\bm{p})^+\right]$$
and denote $\bm{p}^*(\bm{d}')$ denotes its optimal solution. Assumption \ref{assump3} shares the same part (a) with Assumption \ref{assump2} and extends the part (b) and (c) of Assumption \ref{assump2} to a uniform condition for all $\bm{d}'\in \Omega_d$. Accordingly, we update the definition of $\Xi$ to denote the family of distributions that satisfy Assumption \ref{assump1} and \ref{assump3}; for this subsection, we consider the distribution $\mathcal{P}$ within this updated $\Xi.$
\begin{assumption}[Uniform version of Assumption \ref{assump2}]
We assume
\begin{itemize}
    \item[(a)] The second-order moment matrix $\bm{M} \coloneqq \E_{(r,\Aa)\sim \mathcal{P}}[\Aa \Aa^\top]$ is positive-definite. Denote its minimum eigenvalue with $\lambda_{\min}.$
    \item[(b)] There exist constants $\lambda$ and $\mu$ such that if $(r,\bm{a})\sim \mathcal{P}$,
    $$ \lambda |\bm{a}^\top\p - \bm{a}^\top\p^*(\bm{d}')| \le \left|\prob(r>\bm{a}^\top \bm{p}\vert\bm{a}) - \prob(r>\bm{a}^\top\bm{p}^*(\bm{d}')\vert\bm{a})\right| \le \mu |\bm{a}^\top\p - \bm{a}^\top\p^*(\bm{d}')|$$
    holds for any $\bm{p} \in \Omega_{p}$ and $\bm{d}'\in \Omega_d$.
    \item[(c)] The optimal solution $\p^*(\bm{d}')$ satisfies $p^*_i(\bm{d}') = 0$ if and only if $d_i' - \E_{(r,\Aa)\sim \mathcal{P}}[a_iI(r>\Aa^\top \p^*(\bm{d}'))]>0$ for any $\bm{d}'\in \Omega_d$.
\label{assump3}
\end{itemize}
\end{assumption}

Theorem \ref{Algo2} states that Algorithm \ref{alg:HDLA} incurs a worst-case regret of $O(\log{n}\log\log n)$.   Technically, the proof builds upon an analysis of the three components of the generic upper bound in Theorem \ref{representation}. A caveat is that the SAA problem and the underlying stochastic program in Algorithm \ref{alg:HDLA} are dynamically changing over time. To analyze the algorithm, we identify a subset $\mathcal{D}\subset\Omega_d$ around the initial $\bm{d}=\bm{b}/n$ and show that (i) when $\bm{d}_t=\bm{d}'\in \mathcal{D},$ Assumption \ref{assump1} and \ref{assump3} are satisfied, and more importantly, all the stochastic programs specified by $\bm{d}'$ share the same binding and non-binding dimensions; (ii) the process $\bm{d}_t$ exits the region $\mathcal{D}$ only at the very end of the horizon. 
Since the constraint process $\bm{b}_t$ and $\bm{d}_t$ are more complicated, 
we adopt a different approach than the proof of Theorem \ref{Algo1}. The proof of the theorem and more discussions are deferred to Section \ref{AlogT}.

\begin{theorem}
With the online policy $\bm{\pi}_3$ specified by Algorithm \ref{alg:HDLA},
\begin{gather*}
\Delta_n(\bm{\pi}_3) \le O(\log{n}\log\log n).
\end{gather*}
\label{Algo2}
\end{theorem}

\begin{figure}[ht!]
\begin{subfigure}{.5\textwidth}
  \centering
  \includegraphics[width=.98\linewidth]{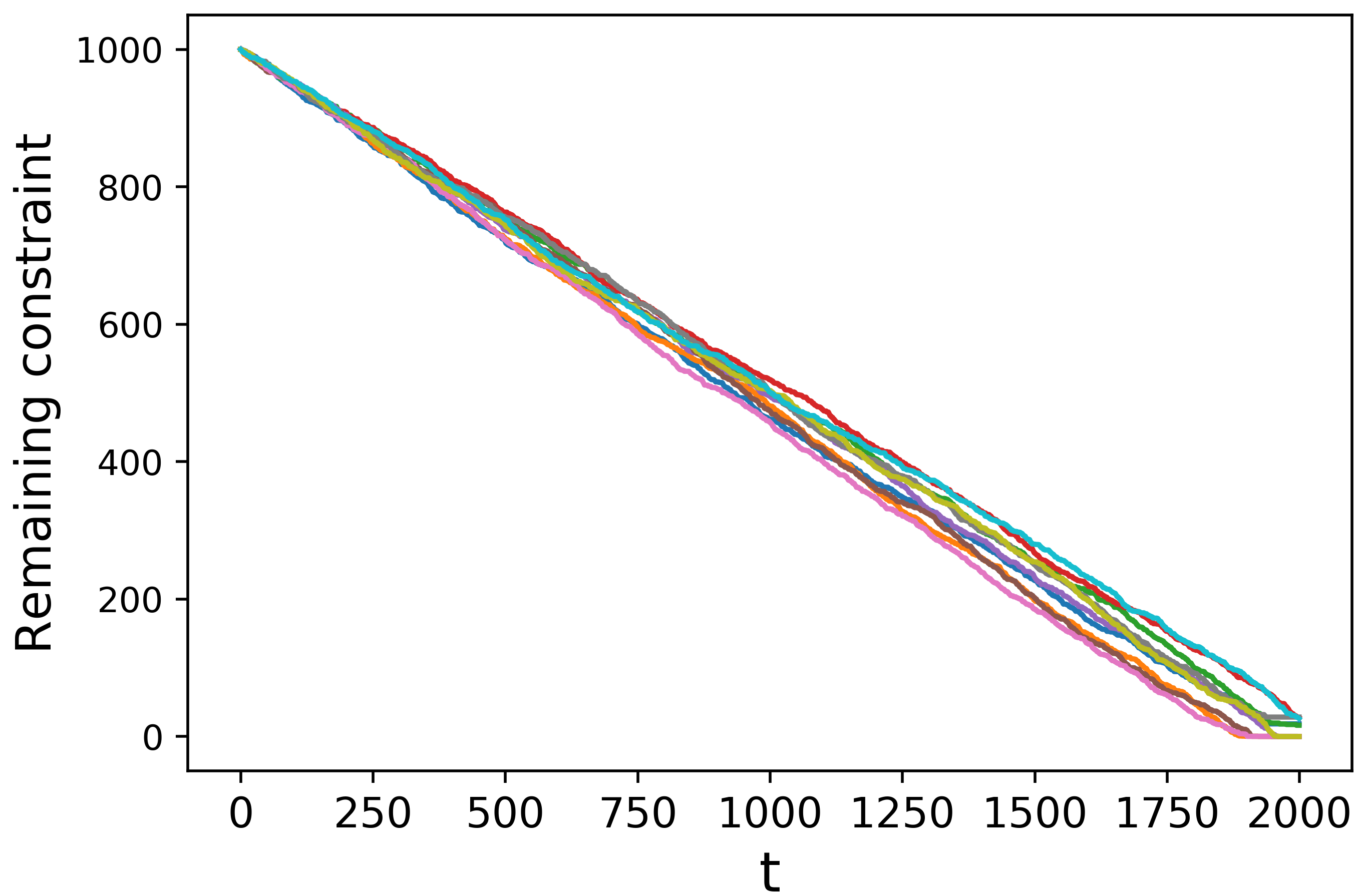}
  \caption{Algorithm 2}
  \label{fig:sub3}
\end{subfigure}%
\begin{subfigure}{.5\textwidth}
  \centering
  \includegraphics[width=.98\linewidth]{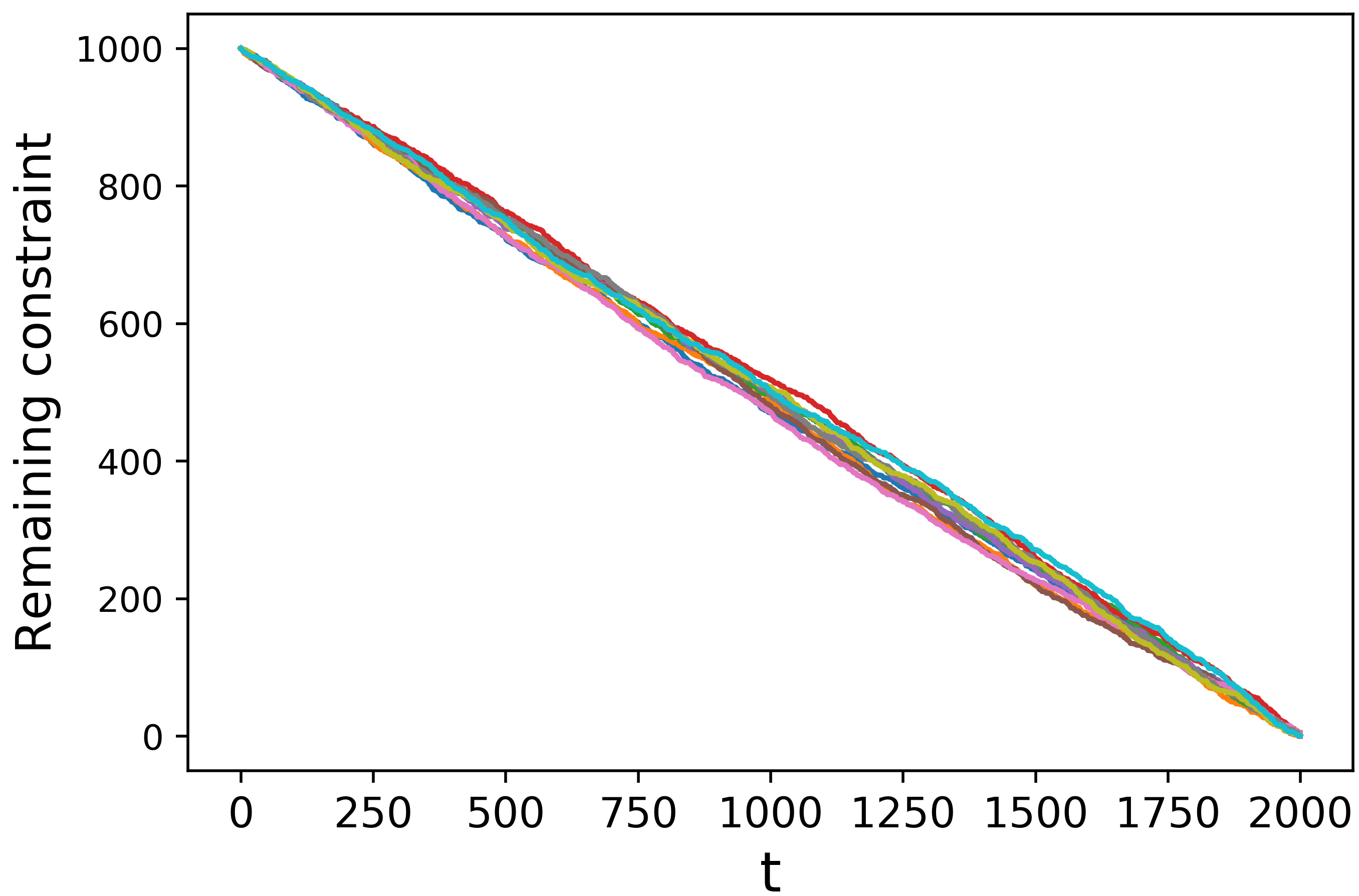}
  \caption{Algorithm 3}
  \label{fig:sub4}
    \end{subfigure}
    \begin{subfigure}{.5\textwidth}
  \centering
  \includegraphics[width=.98\linewidth]{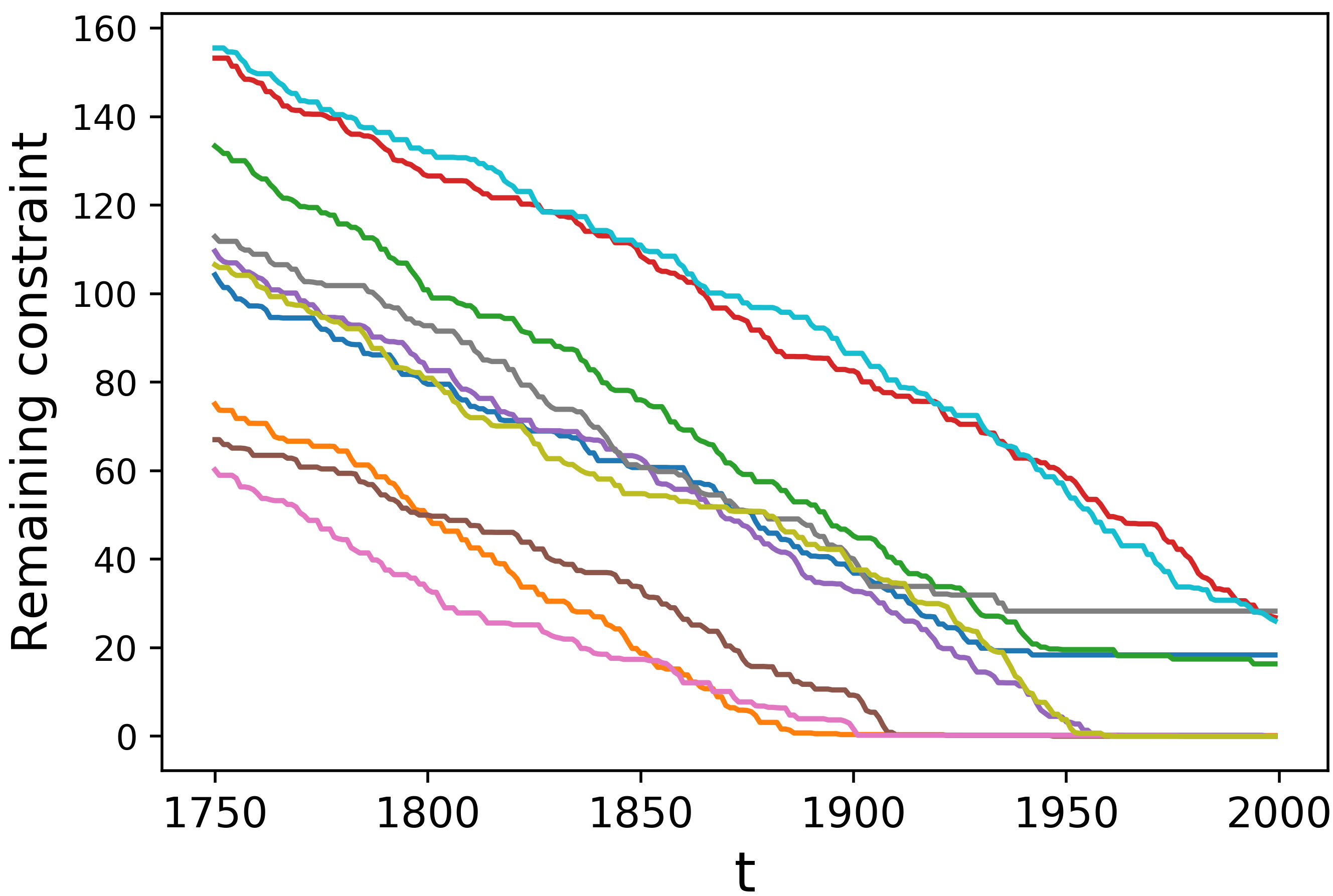}
  \caption{Algorithm 2}
  \label{fig:sub5}
    \end{subfigure}
    \begin{subfigure}{.5\textwidth}
  \centering
  \includegraphics[width=.98\linewidth]{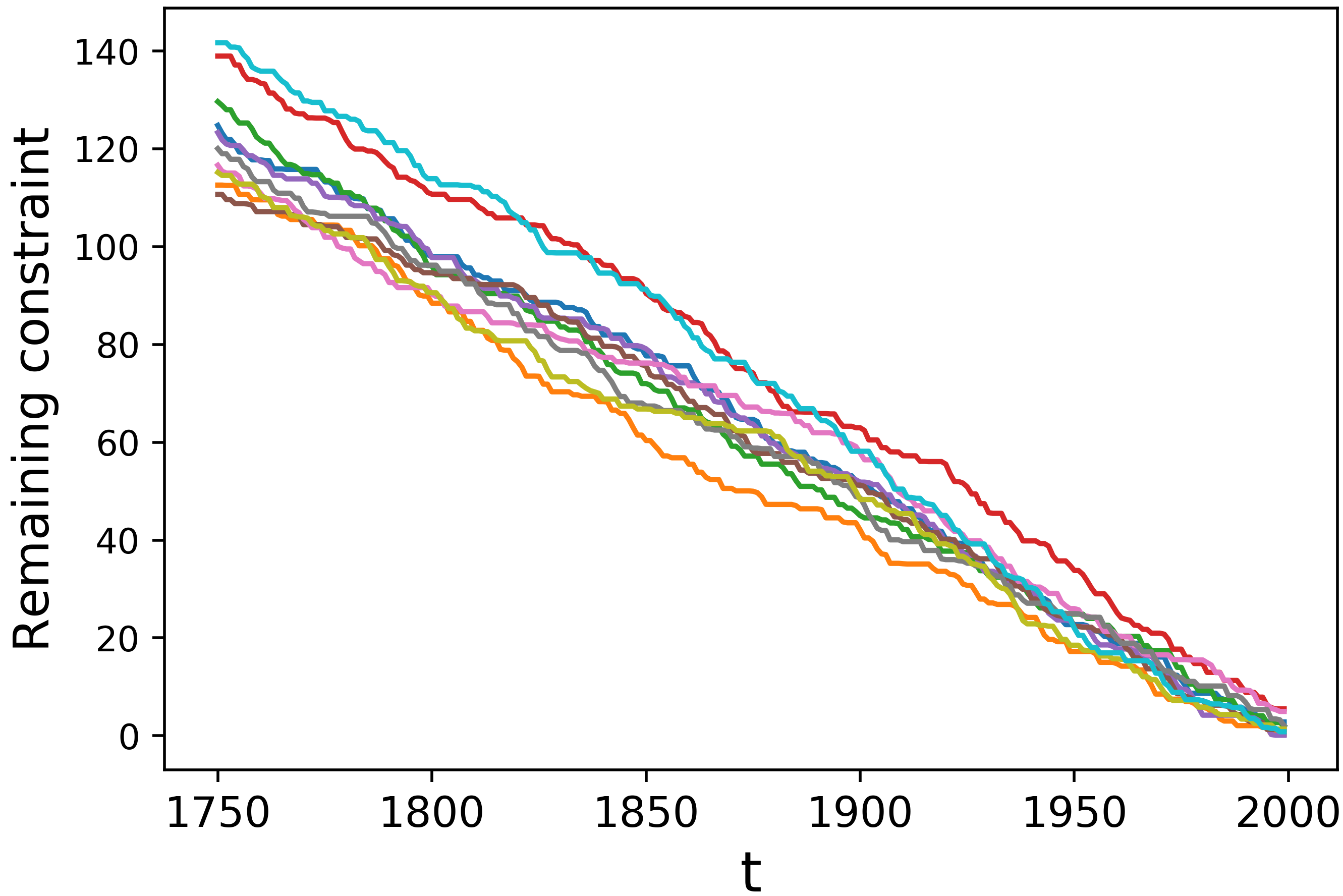}
  \caption{Algorithm 3}
  \label{fig:sub6}
    \end{subfigure}
    \caption{Constraint Consumption: $m=4$ and $n=2000$. Both $r_{j}$ and $a_{ij}$ are generated from $\text{Uniform}[0,1]$. The bottom two plots zoom in to the last 250 periods in the top two plots.}
    \label{fig:RemainingResource}
\end{figure}

Figure \ref{fig:RemainingResource} visualizes the constraint consumption under Algorithm \ref{alg:IDLA} and Algorithm \ref{alg:HDLA}. We run $10$ simulation trials for each algorithm and plot the constraint consumption of a binding constraint from each random trial. The top two figures show that both algorithms seem to perform well in balancing the resource consumption: not to exhaust the constraint too early or to have too much leftover at the end. But if we zoom into the last 250 periods as the bottom two figures, the advantage of Algorithm \ref{alg:HDLA} becomes significant. For Algorithm \ref{alg:IDLA}, some trials exhaust the constraint $O(\sqrt{n})$ periods prior to the end, while some trials have an $O(\sqrt{n})$ remaining at the end. Interestingly, for some curves (like the grey and green ones) in Figure \ref{fig:sub5}, the remaining resource level stops decreasing $O(\sqrt{n})$ periods prior to the end, though the remaining resource level is strictly positive. This is because some other constraint(s) has been exhausted at that time, and from that point on, we can not accept more orders even though there are still remaining resources for the plotted constraint. In comparison, the constraint consumption of Algorithm \ref{alg:HDLA} is much more stable. 

\section{Experiments and Discussions}

\subsection{Numerical Experiments}

We implement the three proposed algorithms on three different models, with model details given in Table \ref{tab:Model}. In the first model (Random Input I), the constraint coefficients $\Aa_j$'s and objective coefficients $r_j$'s are i.i.d. generated as bounded random variables. All $d_i$'s are set to be $0.25$. In the second model (Random Input II), the constraint coefficient $a_{ij}$ is generated from a normal distribution, which violates the boundedness assumption. The assignment of $r_j$ is deterministic conditional of $\bm{a}_j$ and thus violates Assumption \ref{assump2} (b). Both $a_{ij}$ and $r_j$ take negative values with a positive probability. In Random Input II, we set $d_i$'s alternatively to be $0.2$ and $0.3.$  In the third model, we consider a random permutation model, the same as the worst-case example in \citep{agrawal2014dynamic}. The number of decision variables $n$ is a random variable itself in this permutation model, so we specify its expectation to be $100$ and $300$ in the experiment. 

\begin{table}[ht!]
    \centering
    \begin{tabular}{c|c|c}
    \hline 
         Model  & $\Aa_j$ & $r_j$\\  \hline 
         Random Input I & $a_{ij}\sim \text{Uniform[-0.5,1]}$ &  $r_j \perp \Aa_j$ and $r_j \sim \text{Uniform}[0,10]$ \\ \hline
         Random Input II   & $a_{ij}\sim \text{Normal}(0.5,1)$ & $r_j = \sum_{i=1}^{m} a_{ij}$ \\ \hline
         Permutation &   \multicolumn{2}{c}{\citep{agrawal2014dynamic}} \\ \hline 
    \end{tabular}
    \caption{Models used in the experiments}
    \label{tab:Model}
\end{table}

Table \ref{tab:Regret} reports the estimated regrets of the three algorithms under different combinations of $m$ and $n.$ The estimation is based on $200$ simulation trials and in each simulation trial, a problem instance ($a_{ij}$'s and $r_j$'s) is generated from the corresponding model. While solving the stochastic program in Algorithm \ref{alg:Distribution}, we use an SAA scheme with $10^6$ samples. We have the following observations based on the experiment results. First, Table \ref{tab:Regret} shows that Algorithm \ref{alg:HDLA} performs uniformly better than 
Algorithm \ref{alg:Distribution} and Algorithm \ref{alg:IDLA}. Importantly, Algorithm \ref{alg:HDLA} also excels in the models of Random Input II and Permutation where the assumptions for theoretical analysis are violated. In particular, the random permutation problem instance is used in \cite{agrawal2014dynamic} to establish the necessary condition on the constraint capacity and can thus be viewed as one of the most challenging OLP problems. As far as we know, Algorithm \ref{alg:HDLA} is the first algorithm that employs the action-history-dependent mechanism in a generic OLP setting. In this regard, Algorithm \ref{alg:Distribution} and Algorithm \ref{alg:IDLA} stand for typical algorithms in OLP literature \citep{molinaro2013geometry,agrawal2014dynamic, gupta2014experts,devanur2019near}, and the experiments demonstrate the effectiveness of the action-history-dependent design compared with these works. Besides, the advantage of Algorithm \ref{alg:HDLA} becomes smaller as the ratio $m/n$ goes up. This can be explained by the fact that the dual convergence rate is of order $\sqrt{m/n}$ and therefore a dual-based algorithm like Algorithm \ref{alg:HDLA} would be more effective in a large-$n$ and small-$m$ regime.

\begin{table}[ht!]
    \centering
    \small
    \begin{tabular}{c|ccc|ccc|ccc}
    \toprule
    Model  & \multicolumn{3}{c|}{Random Input I} & \multicolumn{3}{c|}{Random Input II} & \multicolumn{3}{c}{Permutation}\\ \midrule
      Algorithm   & A1& A2& A3 &A1& A2& A3&A1& A2& A3 \\ \midrule 
      m = 4, n = 100   &28.17&37.68&\textbf{27.14}&11.75&23.85&\textbf{5.29}&33.17& 37.62&\textbf{7.42}\\
      m = 4, n = 300    &60.17&86.33&\textbf{45.01}&37.59&76.43&\textbf{5.47}& 109.6&51.52&\textbf{11.48}\\
      m = 16, n = 100   &30.21&45.16&\textbf{27.59}&93.34&81.02&\textbf{52.69}&24.96&21.96&\textbf{12.88} \\
      m = 16, n = 300    &60.76&88.91&\textbf{46.30}&184.4&160.4&\textbf{49.13}&140.5&59.56&\textbf{14.87} \\
      m = 64, n = 100    &36.84&40.52&\textbf{34.77}&493.2&461.7&\textbf{414.5}&37.70&20.31&\textbf{15.97} \\
      m = 64, n = 300    &67.78&87.68&\textbf{52.90}&1017.6&881.9&\textbf{611.1}&145.9&47.22&\textbf{18.76} \\
      \bottomrule
    \end{tabular}
    \caption{Regret performance: A1, A2, and A3 stand for Algorithm 1 (No-need-to-learn), Algorithm 2 (Simplified Dynamic Learning), and Algorithm 3 (Action-history-dependent), respectively.}
    \label{tab:Regret}
\end{table}

To illustrate how the regret scales with $n$, we fix $m=4$ and run the experiments for different $n(=25,50,100,250,500,1000,2000)$. The results are presented in Figure \ref{fig:Regret}, where the curves are plotted by connecting the sample points. For the left panel, the curves verify the regret results in Theorem \ref{Algo3}, \ref{Algo1} and \ref{Algo2}. Meanwhile, the right panel looks interesting in that the regrets of Algorithm \ref{alg:Distribution} 
and \ref{alg:IDLA} scale linearly with $n$ while the regret of Algorithm \ref{alg:HDLA} is $O(1)$ (a horizontal line can be fitted). This phenomenon is potentially caused by the deterministic assignment of $r_j$'s in the Random Input II model. 

\begin{figure}[ht!]
\begin{subfigure}{.5\textwidth}
  \centering
  \includegraphics[width=.93\linewidth]{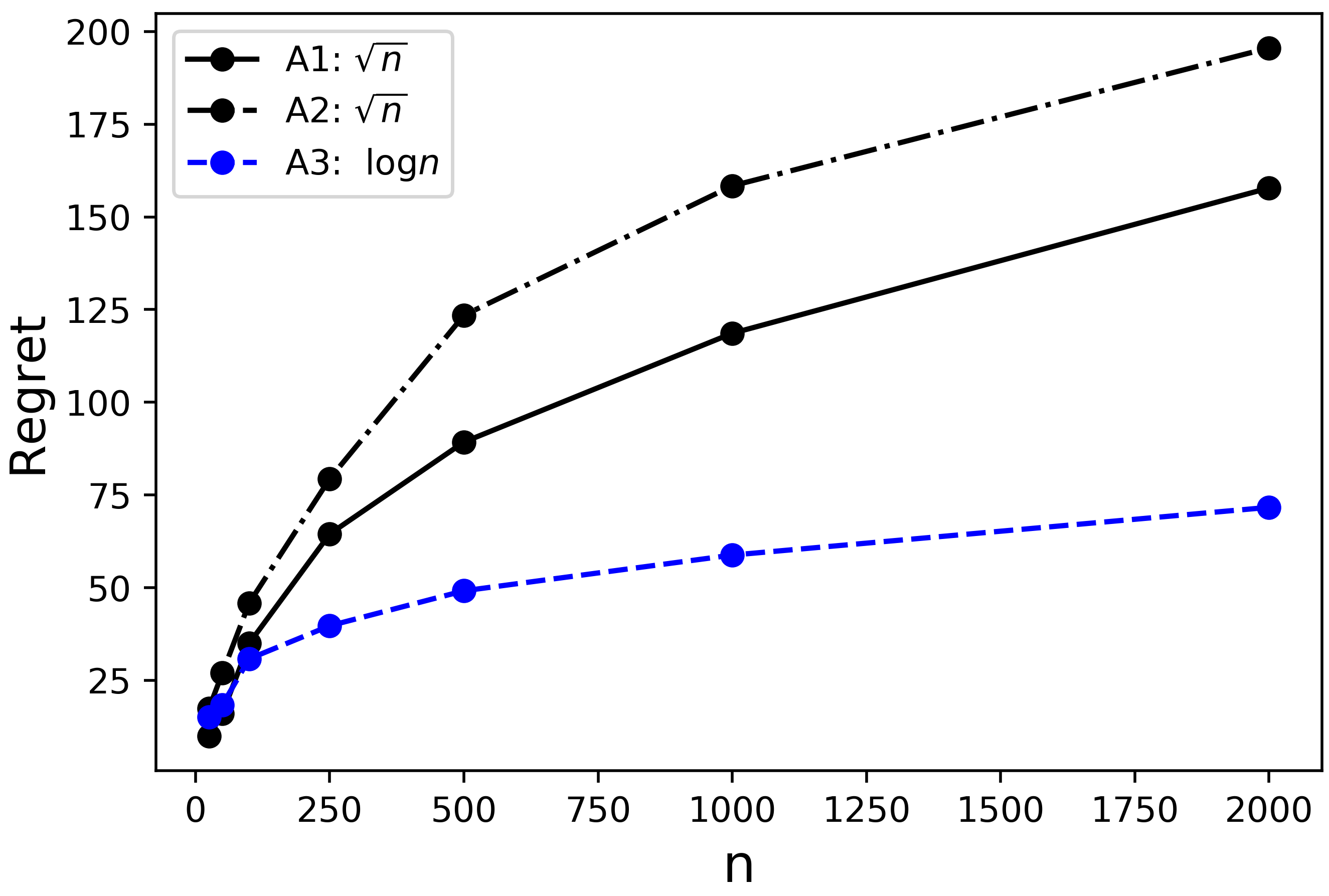}
  \caption{Random Input I}
  \label{fig:sub1}
\end{subfigure}%
\begin{subfigure}{.5\textwidth}
  \centering
  \includegraphics[width=.98\linewidth]{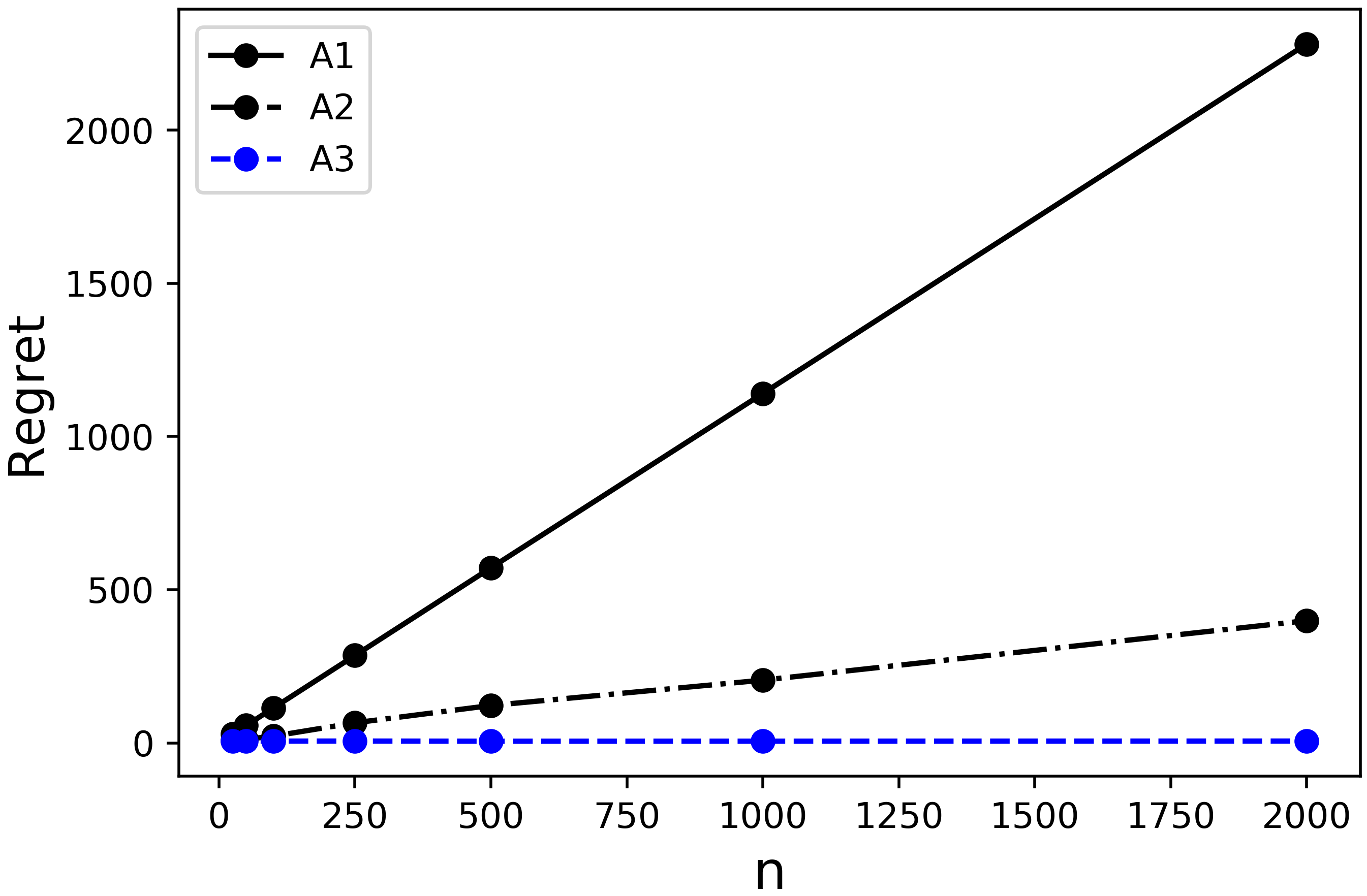}
  \caption{Random Input II}
  \label{fig:sub2}
    \end{subfigure}
    \caption{Regret curves with $m=4$}
    \label{fig:Regret}
\end{figure}

In the numerical experiments (Figure \ref{fig:Regret} (a)), we observe empirically that the regrets for Algorithm \ref{alg:Distribution} and for Algorithm \ref{alg:IDLA} are both of order $\sqrt{n}$. This indicates that the geometric interval is already sufficient for Algorithm \ref{alg:IDLA}. In other words, the performance of Algorithm \ref{alg:IDLA} cannot be further improved by simply increasing learning frequency. Furthermore, it means the consideration of past actions is necessary if we want to reduce the regret to $O(\log n)$. The same observation has been made for the network revenue management problem and thus motivates the study of the re-solving technique.

Table \ref{tab:networkRM} reports numerical experiments based on the network revenue management problem \citep{jasin2015performance}. The network revenue management problem therein can be formulated as an OLP problem by associating each arriving customer with a binary decision variable which denotes the decision of acceptance or rejection. The parameters in our experiment are the same as the numerical experiment in \citep{jasin2015performance}: demand rate, deterministic price, itinerary structure, and flight capacity. We implemented the PAC (Probabilistic Allocation Control) algorithm in \citep{jasin2015performance} and the re-solving algorithm in \citep{jasin2012re}. Two versions of the re-solving algorithm are implemented, one with the knowledge of true parameters (Re-solving (I)) and the other one with an imprecise parameter estimate (Re-solving (II)). Specifically, we perturb the arrival probability estimate by a magnitude of 0.005 for each customer type and re-normalize the perturbed probability. The imprecise estimate aims to capture the scenario that the decision maker has no access to the true parameters and can only calibrate the arrival probability through history data. Different re-solving frequencies and different horizons are tested for all four algorithms and the performance is reported based on an average over 200 simulation trials. Specifically, for the PAC algorithm, we set the initial learning period $t_0=15$ which means accepting all the customers when $t\le t_0.$

\begin{table}[ht!]
    \centering
    \small
    \begin{tabular}{c|c|cc|cc|cc|cc}
    \toprule
    \multicolumn{2}{c|}{} & \multicolumn{2}{c|}{Algorithm 3} & \multicolumn{2}{c|}{PAC} & \multicolumn{2}{c|}{Re-solving (I)} & \multicolumn{2}{c}{Re-solving (II)} \\ \midrule
    \multicolumn{2}{c|}{Re-solving freq.}  & 2 & 10  &2& 10 & 2& 10 & 2 & 10 \\ \midrule 
 \multirow{4}{*}{Horizon}   &  n = 200   & 63.73 & 64.83  &123.49 & 168.32&60.04 & 84.61 &  62.59 &87.34 \\
    &  n = 400 &  82.67 & 92.01  &  141.32& 183.14  & 121.62  & 129.68   &  131.29& 140.52 \\
    &  n = 800 & 114.57  & 111.50 & 179.93    &192.22&  205.67  & 235.31 & 226.55 & 266.97 \\
    &  n = 1600 & 122.81 & 131.76 & 192.18  & 221.03  &  383.37& 402.52 &  477.73 & 486.78 \\
      \bottomrule
    \end{tabular}    
\caption{Regret performance on network revenue management experiment \citep{jasin2015performance}. PAC stands for the Probabilistic Allocation Control algorithm proposed in \citep{jasin2015performance}. Re-solving (I) stands for the re-solving based algorithm proposed in \citep{jasin2012re}. Re-solving (II) stands for the re-solving based algorithm proposed in \citep{jasin2012re} but with an imprecise parameter estimate. Re-solving frequency denotes the frequency with which the dual or primal control is updated for all four algorithms. Horizon denotes the number of periods (total number of customers in network revenue management problem/decision variables in OLP). The unit for the numbers in the table is \$100.}
    \label{tab:networkRM}
\end{table} 

Our action-history-dependent algorithm (Algorithm \ref{alg:HDLA}) has a better empirical performance than the other three algorithms. The result is a bit surprising in that the PAC algorithm can be viewed as a primal version of our action-history-dependent algorithm and it also has the adaptive design through re-solving and re-optimization. In addition, the Re-solving (I) algorithm even utilizes the knowledge of true parameter values. We believe the advantage of our algorithm is only up to a constant factor when $n$ is sufficiently large, and we provide two explanations for the observation. First, there are 41 different itinerary routes and 14 connecting flights ($m=14$) in this experiment. It means that the PAC algorithm needs to estimate 41 parameters, while Algorithm \ref{alg:HDLA} only needs to estimate 14 parameters. Also, the nature of the problem may result in a smaller variance when estimating the optimal dual price. Specifically, the 41 parameters in the PAC algorithms are all probabilities and the average value is approximately 0.02. So the estimation suffers from the efficiency issues in the rare-event simulation \citep{asmussen2007stochastic}. In addition, the two-step procedure of the PAC algorithm feeds the estimation as input for an optimization problem and the optimization procedure may further amplify the estimation error. Second, all three re-solving algorithms are primal-based while our Algorithm \ref{alg:HDLA} is dual-based. Note that the primal-based algorithms output a randomized allocation rule from the optimization procedure, and consequently this randomized rule induces more randomness when deciding the values of the primal variables. The additional randomness may cause more fluctuation of the constraint process. In contrast, our Algorithm \ref{alg:HDLA} is dual-based and can thus be viewed as a smoother version of the three primal-based algorithms.

\subsection{Lower Bound and Open Questions} 

Now, we present a lower bound result for the OLP problem for dual-based policies. \cite{bray2019does} established that the worst-case regret of the multi-secretary problem is $\Omega(\log n)$ even with the knowledge of underlying distribution. We provide an alternative proof for the lower bound for two goals. First, the proof mimics the derivation of the lower bounds in \citep{keskin2014dynamic, besbes2013implications} and the core part is based on van Trees inequality \citep{gill1995applications} -- a Bayesian version of the Cramer-Rao bound. Thus, it extends the previous lower bound analysis from an unconstrained setting to a constrained setting. Second, recall that we develop a generic regret upper bound in Theorem \ref{representation}. The lower bound proof shows that under certain conditions, the generic regret upper bound is rather tight. Intuitively, this indicates that an effective learning of $\bm{p}^*$ and stable control of the constraint process are not only sufficient but also probably necessary to guarantee a sharp regret bound. 

\begin{theorem}
There exist constants $\underline{C}$ and $n_0>0$ such that
$$\Delta_{n}(\bm{\pi}) \ge \underline{C} \log n$$
holds for all $n \ge n_0$ and any dual-based policy $\bm{\pi}.$
\label{lower_bound}
\end{theorem}

Based on the experiments and the theoretical results developed in this paper, we raise several open questions for future study. First, how does the regret depend on $m$? In the experiments of Random Input I, we observe that the regret increases but does not scale up much with $m$. This is apparently not the case for Random Input II where the regret increases significantly as $m$ grows larger (See Section \ref{numeric_m} for more experiments). A possible explanation is that the generation of $r_j$'s in Random Input II causes that $r_j$ scales with $m$ and consequently the offline optimal objective value scales linearly with $mn.$ A natural question then is what are the conditions that render the regret dependent on $m$ and in what way the regret depends on $m$. Second, is it possible to relax the assumptions and extend the theoretical results for more general random input models and permutation models? We observe a good performance of Algorithm \ref{alg:HDLA} when the assumptions are violated and even in the permutation model. For example, we observe an $O(1)$ regret of Algorithm \ref{alg:HDLA} in the experiment under Random Input II (Figure \ref{fig:Regret}).  Since the assumptions are violated, the lower bound does not hold in Random Input II. Also, it is an interesting question to ask if the dual convergence and the regret results still hold under the random permutation model. This question entails a proper definition of the stochastic program (\ref{asymProblem}) in the permutation context. Third, in Assumption \ref{assump1} (c), we require that the constraints scale linearly with $n$. We have not answered the question of whether this linear growth rate is necessary for establishing the dual convergence results. In other words, how can the dual convergence and regret bounds be extended to 
a limited-resource regime? Besides, Algorithm \ref{alg:HDLA} updates the dual price in every period. This raises the question if it is possible to have a less frequent updating/learning scheme but still to achieve the same order of regret. In the network revenue management setting where the distribution is known, \cite{bumpensanti2018re} has shown the effectiveness of an infrequent updating scheme. The analysis therein highly relies on the finiteness of the distribution's support as well as the knowledge of the distribution. We believe it is both interesting and challenging to derive a similar low-regret result with an infrequent updating scheme for the general OLP problem. The last question is on the type of regret bound: all the algorithm bounds we provide in this paper are on regret expectation. An interesting question is how a high-probability regret bound can be derived for the OLP problem (possibly with extra logarithmic terms). In fact, the results of dual convergence in our paper are probabilistic (as in Proposition \ref{pro:gradient} and \ref{pro:Hessian}), which may serve as a good starting point for deriving a high-probability regret bound. 

\section*{Acknowledgment}
\small
\onehalfspacing
We thank Yu Bai, Simai He, Peter W. Glynn, Daniel Russo, Zizhuo Wang, Zeyu Zheng,  and seminar participants at Stanford AFTLab, NYU Stern, Columbia DRO, Chicago Booth, Imperial College Business School and ADSI summer school for helpful discussions and comments. We thank Chunlin Sun, Guanting Chen, and Zuguang Gao for proofreading the proof and Yufeng Zheng for assistance in the simulation experiments.
\normalsize

\onehalfspacing

\bibliographystyle{informs2014} 
\bibliography{sample.bib} 
\newpage 

\appendix
\section*{Appendices}

\renewcommand{\thesubsection}{A\arabic{subsection}}

\section{Proofs for Section 2} 

\subsection{Proof of Proposition \ref{basicProp}}

\begin{proof}
\begin{itemize}
    \item[(a)] The original dual problem \eqref{dualLP} can be recovered by substituting $y_j = \left(r_j-\sum_{i=1}^m a_{ij}p_i\right)^+$ in the objective function (\ref{newDual}). Then the feasible solutions and the objective functions are matched. Therefore, these two problems share the same optimal solution. 
    \item[(b)] We know that each component in the first summation is linear and each component in the second summation is convex. Also, the summation operation preserves convexity (See Chapter 3.2.1 \citep{rockafellar1970convex}). So, both $f_n$ and $f$ are convex functions. 
    \item[(c)] If $\Dd^\top \p > \bar{r},$ then
    \begin{align*}
     f(\p) \ge \Dd^\top \p & >\bar{r}  \ge \E[r] = f(\bm{0}).
    \end{align*}
    Hence $\p$ cannot be the optimal solution. In the same way, we can show the result for $\p_n^*.$
\end{itemize}
\end{proof}

\subsection{Proof of Lemma \ref{TaylorLemma0}}

\begin{proof}
We only need to show the following equality (before taking the expectation),
\begin{align*}
    h(\bm{p}, (r,\bm{a})) - h(\bm{p}^*, (r,\bm{a})) = \phi(\p^*, (r,\bm{a}))^\top(\p-\p^*) +  \int_{\bm{a}^\top \p}^{\bm{a}^\top \p^*}\left(I(r>v) - I(r>\bm{a}^\top\p^*) \right)dv.
\end{align*}
Indeed, 
\begin{align*}
    &h(\bm{p}, (r,\bm{a})) - h(\bm{p}^*, (r,\bm{a})) - \phi(\p^*,(r,\bm{a}))^\top(\p-\p^*) \\
    = &  \bm{d}^\top \bm{p} + (r - \Aa^\top \p)^+-\bm{d}^\top \bm{p}^* - (r - \Aa^\top \p^*)^+ - (\Dd-\Aa\cdot I(r_j > \Aa^\top \p^*))^\top(\p-\p^*) \\
    = & (r - \Aa^\top \p)^+ -(r - \Aa^\top \p^*)^+ +  I(r > \Aa^\top \p^*)\cdot(\Aa^\top \p-\Aa^\top\p^*) \\
    = & \int_{\bm{a}^\top \p}^{\bm{a}^\top \p^*}I(r>v)dv  +  I(r > \Aa^\top \p^*)\cdot(\Aa^\top \p-\Aa^\top\p^*) 
    \\ = & \int_{\bm{a}^\top \p}^{\bm{a}^\top \p^*}\left(I(r>v) - I(r>\bm{a}^\top\p^*) \right)dv.
\end{align*}
By taking expectation with respect to $(r,\bm{a})\sim \mathcal{P}$, we obtain the identity (\ref{identity0}).
\end{proof}

\subsection{Examples for Assumption 2 (b)}

\label{egAssump}

In Example 2, it is easy to see that for $\bm{p}$ such that $|\bm{a}^\top\bm{p}-\bm{a}^\top\bm{p}^*|\le c_\epsilon.$ We have
$$\underline{\alpha}|\bm{a}^\top\p - \bm{a}^\top\p^*| \le \left|\prob(r>\bm{a}^\top \bm{p}\vert\bm{a}) - \prob(r>\bm{a}^\top\bm{p}^*\vert\bm{a})\right| \le \bar{\alpha} |\bm{a}^\top\p - \bm{a}^\top\p^*|.$$
For a general $\bm{p}\in \Omega_p$, we know 
$$|\bm{a}^\top\p - \bm{a}^\top\p^*| \le  \bar{a}\|\p - \p^*\|_2 \le \bar{a}\|\p - \p^*\|_1 \le \frac{\bar{a}\bar{r}}{\underline{d}}$$
where the last inequality comes from Proposition \ref{basicProp} (c). Without loss of generality, we consider $\bm{a}^\top\p<\bm{a}^\top\p^*-c_\epsilon$,  then
$$\left|\prob(r>\bm{a}^\top \bm{p}\vert\bm{a}) - \prob(r>\bm{a}^\top\bm{p}^*\vert\bm{a})\right|\ge \underline{\alpha}c_\epsilon \ge \frac{\underline{\alpha}c_\epsilon\underline{d}}{\bar{a}\bar{r}} |\bm{a}^\top\p - \bm{a}^\top\p^*|.$$
So, we can choose $\mu=\bar{\alpha}$ and $\lambda=\min \left\{\underline{\alpha}, \frac{\underline{\alpha}c_\epsilon\underline{d}}{\bar{a}\bar{r}} \right\}$ and Assumption \ref{assump2} (b) is satisfied. 

The implication for this second example is that for the existence of $\mu$, we only need that the density function has a finite upper bound, and for the existence of $\lambda,$ we can impose a locally lower bound for the density function and extend the condition to a more general support through the above derivation. Then, Example 3 follows the same intuition and analysis as Example 2.

\subsection{Proof of Proposition \ref{strConvex}}

\begin{proof}
To see the result, from Lemma \ref{TaylorLemma0}, we only need to analyze the second-order term
\begin{align}
  &  \E\left[\int_{\bm{a}^\top \p}^{\bm{a}^\top \p^*}\left(I(r>v) - I(r>\bm{a}^\top\p^*) \right)dv\right] \nonumber \\
= & \E\left[\E\left[\int_{\bm{a}^\top \p}^{\bm{a}^\top \p^*}\left(I(r>v) - I(r>\bm{a}^\top\p^*) \right)dv\Bigg\vert a\right] \right]\nonumber\\
= & \E\left[\int_{\bm{a}^\top \p}^{\bm{a}^\top \p^*} \prob(r> v\vert\bm{a}) - \prob(r>\bm{a}^\top\bm{p}^*\vert\bm{a}) dv \right]\nonumber \\
\le & \E\left[\int_{\bm{a}^\top \p}^{\bm{a}^\top \p^*} \mu (\bm{a}^\top\bm{p}^* - v) dv \right] \nonumber \\
=  & \frac{\mu}{2}\E\left[(\bm{a}^\top \bm{p} - \bm{a}^\top \bm{p}^*)^2\right] \\
\le & \frac{\mu\bar{a}^2}{2} \|\p-\p^*\|_2^2\label{smoothness0}
\end{align}
where the first inequality comes from Assumption \ref{assump2} (b) and the second inequality comes from the boundedness of $\bm{a}$ -- Assumption 1 (b). With a similar argument, we can show that 
\begin{align}
  \E\left[\int_{\bm{a}^\top \p}^{\bm{a}^\top \p^*}\left(I(r>v) - I(r>\bm{a}^\top\p^*) \right)dv\right] 
\ge  \frac{\lambda\lambda_{\min}}{2} \|\p-\p^*\|_2^2
\label{smoothness}
\end{align}
where $\lambda_{\min}$ is specified in Assumption \ref{assump2} (a) and $\lambda$ is specified in Assumption \ref{assump2} (b).

To see the uniqueness of $\bm{p}^*$, we first show that 
\begin{align}
    \nabla f(\bm{p}^*) \boldsymbol{\cdot} \bm{p}^* = \bm{0} \label{compliment}
\end{align}
where the operator $\boldsymbol{\cdot}$ denotes the element-wise product. In addition, $\nabla f(\bm{p}^*)\ge \bm{0}$. These results can be shown via a standard argument through the sub-gradient and the optimality condition (See Chapter 23 of \citep{rockafellar1970convex}). For completeness, we provide a self-contained proof here.

To see $\nabla f(\bm{p}^*)\ge \bm{0}$, we note from \eqref{smoothness0} that
$$f(\bm{p})-f(\bm{p}^*)\le \nabla f(\bm{p}^*)(\bm{p}-\bm{p}^*) + \frac{\mu\bar{a}^2}{2} \|\bm{p}-\bm{p}^*\|^2_2.$$
If the $i'$-th entry $(\nabla f(\bm{p}^*))_{i'} <0$ for some $i'$, then we can choose $\bm{p}'=(p'_1,...,p'_m)^\top\ge \bm{0}$ such that ${p}_{i}'={p}^*_i$ for $i\neq i'$ and ${p}_{i'}'={p}^*_{i'} -\frac{(\nabla f(\bm{p}^*))_{i'}}{\mu\bar{a}^2}$. It is easy to verify that 
$$f(\bm{p}')-f(\bm{p}^*)=-\frac{2(\nabla f(\bm{p}^*))_{i'}^2}{\mu\bar{a}^2}<0$$
which contradicts the optimality of $\bm{p}^*.$

Similarly, if $(\nabla f(\bm{p}^*))_{i'} \cdot p^*_{i'}>0,$ we can choose $\bm{p}'=(p'_1,...,p'_m)^\top\ge \bm{0}$ such that ${p}_{i}'={p}^*_i$ for $i\neq i'$ and ${p}_{i'}'=({p}^*_{i'} -\frac{(\nabla f(\bm{p}^*))_{i'}}{\mu\bar{a}^2})\vee 0$. It is easy to verify that $f(\bm{p}')-f(\bm{p}^*)<0$ which also contradicts the optimality of $\bm{p}^*.$

Consequently, the uniqueness follows from
$$f(\bm{p})-f(\bm{p}^*)\ge \nabla f(\bm{p}^*)(\bm{p}-\bm{p}^*) + \frac{\lambda\lambda_{\min}}{2} \|\bm{p}-\bm{p}^*\|^2_2.$$
Specifically, let $\bm{p}'\neq \bm{p}^*$ be another optimal solution to the problem. Then the left-hand-side is zero. But for the right-hand-side, the first term is non-negative from \eqref{compliment} and the fact that $\nabla f(\bm{p}^*)\ge \bm{0}$, and the second term is strictly positive, which leads to a contradiction.
\end{proof}

\subsection{Proof of Lemma \ref{TaylorLemma}}

\begin{proof}
The proof is the same as the proof of Lemma \ref{TaylorLemma0} and it is completed by doing the analysis for each term in the summation of $f_n.$
\end{proof}

\subsection{Proof of Proposition \ref{pro:gradient}}

\label{ApropGrd}
First, we introduce the Hoeffding's inequality for scalar random variables.
\begin{lemma}[Hoeffding's inequality]
Let $X_1, ..., X_n$ be independent random variables such that $X_i$ takes its values in $[u_i, v_i]$ almost surely for all $i\le n.$
Then for every $t>0,$
$$\prob\left(\left\vert \frac{1}{n}\sum_{i=1}^n X_i -\E X_i \right\vert \ge t\right) \le 2\exp\left(-\frac{2n^2t^2}{\sum_{i=1}^{n}(u_i-v_i)^2}\right)$$
\label{HF1}
\end{lemma}
\begin{proof}
We refer to Chapter 2 of the book \citep{boucheron2013concentration}.
\end{proof}

Now, we present the proof of Proposition \ref{pro:gradient}.

\begin{proof}
We use $\phi(\p^*, \uu_j)_i$ to denote the $i$-th coordinate of the gradient vector $\phi(\p^*, \uu_j)$. By the definition of $\phi,$ we know that 
$$\E \phi(p^*, \uu_j)_{i} = \left(\nabla f(\p^*)\right)_i.$$
From the boundedness $\Aa_{j}$'s (Assumption \ref{assump1} (b)), we know that 
$$|\phi(p^*, \uu_j)_{i}| \in [d_i - \bar{a}, d_i+\bar{a}].$$
Then, by applying the Hoeffding's inequality, we obtain
$$\prob\left(\Bigg|\frac{1}{n}\sum_{j=1}^n \phi(\p^*,\bm{u}_j)_i -\left(\nabla f(\p^*)\right)_i\Bigg| \ge \epsilon \right) \le 2 \exp\left(-\frac{n\epsilon^2}{2\bar{a}^2}\right).$$
In fact, 
$$\left\{\Bigg\|\frac{1}{n}\sum_{j=1}^n \phi(\p^*,\bm{u}_j) - \nabla f(\p^*) \Bigg\|_2 \ge \epsilon \right\} \subset  \bigcup_{i=1}^m \left\{\Bigg|\frac{1}{n}\sum_{j=1}^n \phi(\p^*,\bm{u}_j)_i - \left(\nabla f(\p^*)\right)_i\Bigg| \ge \frac{\epsilon}{\sqrt{m}} \right\}.$$
Applying the union bound, 
\begin{align*}
    \prob\left(\Bigg\|\frac{1}{n}\sum_{j=1}^n \phi(\p^*,\bm{u}_j) - \left(\nabla f(\p^*)\right)_i\Bigg\|_2 \ge \epsilon \right) & \le m\prob\left(\Bigg|\frac{1}{n}\sum_{j=1}^n \phi(\p^*,\bm{u}_j)_i - \left(\nabla f(\p^*)\right)_i \Bigg| \ge \frac{\epsilon}{\sqrt{m}} \right) \\
    & \le 2m \exp\left(-\frac{n\epsilon^2}{2\bar{a}^2m}\right). 
\end{align*}
Thus we obtain Proposition \ref{pro:gradient}. 
\end{proof}

\subsection{Proof of Proposition \ref{pro:Hessian}}

\label{ApropHes}
We first introduce a matrix version for the Hoeffding's inequality. 
\begin{lemma}[Matrix Hoeffding's Inequality]
Let $\bm{X}_1, ..., \bm{X}_n \in \mathbb{R}^d$ be i.i.d. random vectors with $\E(\bm{X}_k\bm{X}_k^\top) = \bm{M}$. Also, we assume $\|\bm{X}_{k}\|_2^2 \le B$ almost surely. Let
$$\bm{Z} = \frac{1}{n} \sum_{k=1}^n \bm{X}_k \bm{X}_k^\top.$$
Then 
$$\prob\left(\|\bm{M} - \bm{Z}\|_S \ge t\right) \le d\cdot \exp\left(\frac{-t^2}{Bn}\right)$$
for all $t>0.$
\label{HF2}
\end{lemma}
\begin{proof}
We refer to the Corollary 4.2 (Matrix Hoeffding Inequality) of \citep{mackey2014matrix}. The proof of this lemma simply reduces the matrix in Corollary 4.2 to a vector setting. 
\end{proof}

Now, we prove Proposition \ref{pro:Hessian}.

\begin{proof}[Proof of Proposition \ref{pro:Hessian}]

We complete the proof in three steps: 
\begin{itemize}
    \item[Step 1.] We show that the quantity $\bm{M}_n = \frac{1}{n} \sum_{j=1}^n \Aa_j \Aa_j^\top$ concentrates around its mean 
$\bm{M} = \E \left[\Aa_j \Aa_j^\top \right]$ with high probability. Intuitively, this removes the randomness on $\Aa_j$'s. In the later part of the proof, we will see that the matrix $\bm{M}$ works approximately as the Hessian matrix of the function $f_n(\p).$ Assumption \ref{assump2} (b) states that the minimum eigenvalue of the Hessian matrix is $\lambda_{\min}.$
\item [Step 2.] To establish the uniform result, we introduce a finite set of representative points $\p_{kl}$'s (the indices $k$ and $l$ to be specified later) for the set $\Omega_p.$ As mentioned in the main body of the paper, for each $\p \in \Omega_p$, the function value $$\sum_{j=1}^n \int_{\bm{a}_j^\top \p}^{\bm{a}_j^\top \p^*}\left(I(r_j>v)-I(r_j>\bm{a}_j^\top\p^*)\right)dv$$ is a random variable dependent on $(r_j,\Aa_j)$'s. Though we can apply the concentration inequality to analyze the function value for each specific $\p$, the argument does not go through for all (uncountably many) points in $\Omega_p$ as in the proposition. The representative points serve as an intermediary between the point-wise argument and the uniform argument. The idea is:
\begin{itemize}
    \item[Part (i).] We first show the quantity $\sum_{j=1}^n \int_{\bm{a}_j^\top \p_{kl}}^{\bm{a}_j^\top \p^*}\left(I(r_j>v)-I(r_j>\bm{a}_j^\top\p^*)\right)dv$ concentrates around its mean for all representative points $\p_{kl}$'s 
    \item[Part (ii).] We then establish that for each $\p\in \Omega_p$, there is a  $\p_{kl}$ near $\p$ such that the difference 
\begin{align*}
   & \underbrace{\sum_{j=1}^n \int_{\bm{a}_j^\top \p_{kl}}^{\bm{a}_j^\top \p^*}\left(I(r_j>v)-I(r_j>\bm{a}_j^\top\p^*)\right)dv}_{\text{what has been analyzed by the above Part (i)}}- \underbrace{\sum_{j=1}^n \int_{\bm{a}_j^\top \p}^{\bm{a}_j^\top \p^*}\left(I(r_j>v)-I(r_j>\bm{a}_j^\top\p^*)\right)dv}_{\text{the goal of the proposition}}\\
    = &\sum_{j=1}^n \int_{\bm{a}_j^\top \p_{kl}}^{\bm{a}_j^\top \p}\left(I(r_j>v)-I(r_j>\bm{a}_j^\top\p^*)\right)dv
\end{align*}
is small with high probability. 
\end{itemize}
In this way, we make the argument through for all $\p \in \Omega.$

\item[Step 3.] We combine the two parts in Step 2 and then put the result together with Step 1.
\end{itemize}

For \textbf{Step 1}, consider 
$$\bm{M}_n = \frac{1}{n} \sum_{j=1}^n \Aa_j \Aa_j^\top$$
$$\bm{M} = \E \left[\Aa_j \Aa_j^\top \right]$$
where the expectation is taken with respect to $(r_j, \bm{a}_j)\sim \mathcal{P}.$ We know from Assumption \ref{assump2} (a) that the minimum eigenvalue of $\bm{M}$ is $\lambda_{\min}$. Also,
$$\lambda_{\min} - \lambda_{\min}(\bm{M}_n) \le \lambda_{\max}(\bm{M} -\bm{M}_n) \le \|\bm{M}-\bm{M}_n\|_S$$
where $\lambda_{\min}(\cdot)$ and $\lambda_{\max}(\cdot)$ refer to the smallest and largest eigenvalue of a matrix, respectively. Denote event
$$\mathcal{E}_0 = \left\{\lambda_{\min}(\bm{M}_n) \le \frac{\lambda_{\min}}{2}\right\}.$$ Applying Lemma \ref{HF2},
\begin{align}
    \prob(\mathcal{E}_0) = \prob\left(\lambda_{\min}(\bm{M}_n) \le \frac{\lambda_{\min}}{2}\right) & \le  \prob\left(\|\bm{M}-\bm{M}_n\|_S \ge \frac{\lambda_{\min}}{2} \right) \nonumber \\
    & \le m\cdot \exp\left(\frac{-n\lambda_{\min}^2}{4\bar{a}^2}\right). \label{p3prob1}
\end{align}
where $\bar{a}$ is the upper bound on $\Aa_j$ from Assumption \ref{assump1} (b). So, we complete \textbf{Step 1} by showing that the random matrix $\bm{M}_n$ has a minimum eigenvalue larger than $\frac{\lambda_{\min}}{2}$ with high probability.

For \textbf{Step 2}, we first present how we select the representative set of points $\bm{p}_{kl}$'s. From Proposition \ref{basicProp} (c), we know that the optimal solution $\p_n^*$ and $\p^*$ is bounded. Define set $\bar{\Omega} =\left\{\p \in \mathbb{R}^m \Big| \|\p-\p^*\|_{\infty} \le \frac{\bar{r}}{\underline{d}}\right\}$ and $\p_n^*\in \Omega_p \subset \bar{\Omega}$ almost surely. We only need to show the results for the larger set $\bar{\Omega}.$ The region $\bar{\Omega}$ can be split into a union of disjoint sets $$\bar{\Omega} = \bigcup_{k=1}^{N} \bigcup_{l=1}^{l_k} \Omega_{kl}.$$
The splitting scheme is inspired by \citep{huber1967behavior}. 
Specifically, these sets are divided layer by layer. The set $\bar{\Omega}_{k} =\left\{\p \in \mathbb{R}^m \Big| \|\p-\p^*\|_{\infty} \le q^k\frac{\bar{r}}{\underline{d}}\right\}$ for $k=0,...,N.$ Here $N$ and $q\in (0,1)$ will be determined later. The $k$-th layer $\bar{\Omega}_{k-1}\setminus \bar{\Omega}_{k}$ is further divided into disjoint cubes $\{\Omega_{kl}\}_{l=1}^{l_k}$ with edges of length $(1-q)q^{k-1}\frac{\bar{r}}{\underline{d}}$ for $k=1,...,N-1$ and $l=1,...,l_k$. The center cube is simply $\bar{\Omega}_{N} = \Omega_{N1}$ with edge of length $q^N\frac{\bar{r}}{\underline{d}}$ and $l_N=1.$ Also, the length $q$ is adjusted in a way that the splitting scheme cut the region $\bar{\Omega}$ into an integer number of cubes. Figure \ref{fig:omega} gives a visualization of the splitting scheme. In total, there are no more than $(2N)^m$ cubes (See Lemma 3 in \citep{huber1967behavior}). In the following, we will complete the two parts in Step 2 for the cubes. First, we analyze the outer cubes $\Omega_{kl},$ for $k=1,...,N-1$ and $l=1,...,l_k$ and then we treat the center cube $\Omega_{N1}$ 
separately.

\begin{figure}[ht!]
    \centering
    \includegraphics[width=0.6\textwidth]{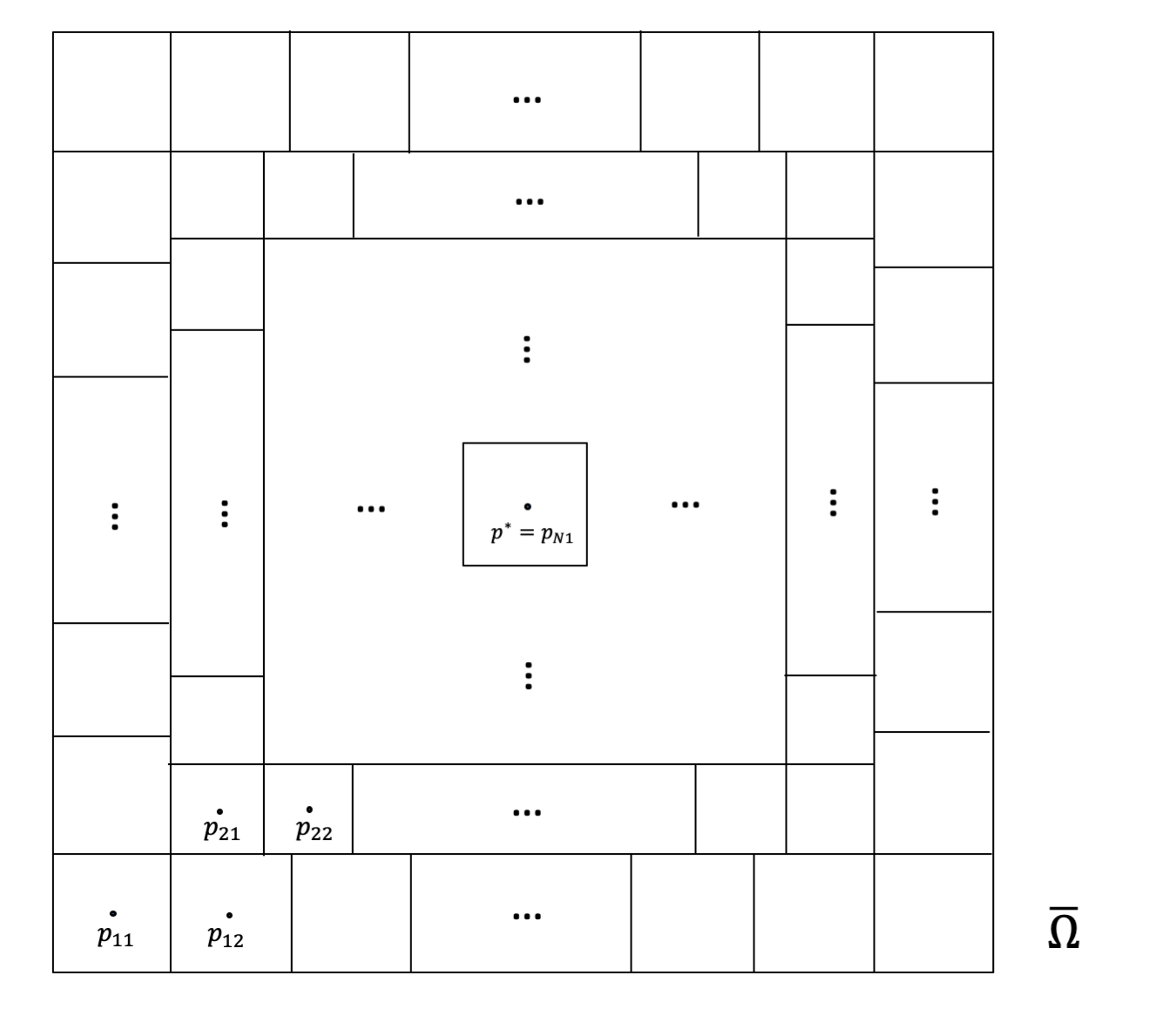}
    \caption{Visualization of the splitting scheme}
    \label{fig:omega}
\end{figure}

Let $\p_{kl}$ be the center of the cube $\Omega_{kl}$, $\underline{\p}_{kl}$ and $\bar{\p}_{kl}$ be the points in $\Omega_{kl}$  that are closest and furthest from $\p^*$, respectively. That is,
$$\underline{\p}_{kl} = \argmin_{\p \in \Omega_{kl}} \|\p-\p^*\|_2,$$
$$\bar{\p}_{kl} = \argmax_{\p \in \Omega_{kl}} \|\p-\p^*\|_2.$$

Now, we derive the \textbf{Part (i) of Step 2.}
To obtain an upper bound for the usage of Hoeffding's Inequality, we have
$$ \left\vert\int_{\bm{a}_j^\top \p_{kl}}^{\bm{a}_j^\top \p^*}\left(I(r_j>v)-I(r_j>\bm{a}_j^\top\p^*)\right)dv\right\vert \le \left|\bm{a}_j^\top \p_{kl}- \bm{a}_j^\top \p^*\right|\le \bar{a}\|\p^*-\bar{\p}_{kl}\|_2$$
where the right-hand-side is a deterministic quantity that does not depend on $(r_j, \Aa_j)$'s. We define the following event (for the cube $\Omega_{kl}$ and the point $\bm{p}_{kl}$ indexed by $k,l$) that the integral deviates from its mean, 
\begin{align*}
    \mathcal{E}_{kl,1} & = \Bigg\{\frac{1}{n} \sum_{j=1}^n \int_{\bm{a}_j^\top \p_{kl}}^{\bm{a}_j^\top \p^*}\left(I(r_j>v)-I(r_j>\bm{a}_j^\top\p^*)\right)dv \le \\ & \ \  \ \ \ \frac{1}{n} \sum_{j=1}^n \E\left[ \int_{\bm{a}_j^\top \p_{kl}}^{\bm{a}_j^\top \p^*}\left(I(r_j>v)-I(r_j>\bm{a}_j^\top\p^*)\right)dv\Bigg \vert \Aa_{1},...,\Aa_{n} \right] - \epsilon \bar{a}\|\p^*-\bar{\p}_{kl}\|_2  \Bigg\}.
\end{align*} 
Note that $r_1,...,r_n$ are independent conditional on $\Aa_1,...,\Aa_n.$ By applying Hoeffding's inequality, we know,
\begin{equation}
    \prob(\mathcal{E}_{kl,1}\vert \Aa_1,...,\Aa_n)\le \exp\left(-\frac{n\epsilon^2}{2}\right)
    \label{p3hf1}
\end{equation}
for $\Aa_1,...,\Aa_n,$ $k=1,...,N-1$ and $l=1,...,l_k.$ Consequently, we know $\prob\left(\mathcal{E}_{kl,1}\right) \le \exp \left(-\frac{n\epsilon^2}{2}\right)$ by integrating with respect to $(\Aa_1,...,\Aa_n).$ This completes Part (i) of Step 2, establishing that the random integral from $\bm{a}_j^\top \p_{kl}$ to $\bm{a}_j^\top \p^{*}$ concentrates around its (conditional) mean with high probability.

For \textbf{Part (ii) of Step 2}, define 
$$\Gamma_{kl}(r_j,\Aa_j) = \max_{\p \in \Omega_{kl}} \int_{\Aa_j^\top\p}^{\Aa_j^\top\p_{kl}} I(r_j>v) - I(r_j>\Aa_j^\top \p^*)dv.$$
We know that
\begin{align*}
    \E [\Gamma_{kl}(r_j,\Aa_j)|\Aa_1,...,\Aa_n] & = \E\left[\max_{\p \in \Omega_{kl}} \int_{\Aa_j^\top\p}^{\Aa_j^\top\p_{kl}} I(r_j>v) - I(r_j>\Aa_j^\top \p^*)dv\Bigg\vert \Aa_1,...,\Aa_n\right] \\
    & \le \E \left[\max_{\p \in \Omega_{kl}} \int_{\Aa_j^\top\p}^{\Aa_j^\top\p_{kl}}  I(v<r_j\le\Aa_j^\top \p^*)dv\Bigg\vert \Aa_1,...,\Aa_n\right] \\
    & \le \bar{a} \max_{\p \in \Omega_{kl}} \|\p-\p_{kl}\|_2 \cdot \E \left[\max_{\p \in \Omega_{kl}} I(\Aa_j^\top\p <r_j\le\Aa_j^\top \p^*)dv \Bigg\vert \Aa_1,...,\Aa_n\right]
    \\
    & \le \bar{a} \max_{\p \in \Omega_{kl}} \|\p-\p_{kl}\|_2 \cdot \mu \bar{a} \|\p^*-\bar{\p}_{kl}\|_2 \\
    & = \mu \bar{a}^2 \|\p^*-\bar{\p}_{kl}\|_2 \max_{\p \in \Omega_{kl}} \|\p-\p_{kl}\|_2. 
\end{align*}
Here the second line comes from the definition of the indicator function. The third line singles out the two limits of the integral. The fourth line comes from Assumption \ref{assump2} (b) and the definition of $\bar{\bm{p}}_{kl}$ (the furthest point in $\Omega_{kl}$ from $\p^*$.) Specifically, the parameter $\bar{a}$ comes from Assumption \ref{assump1} (b) and $\mu$ comes from Assumption \ref{assump2} (b).
Also,
$$|\Gamma_{kl}(r_j,\Aa_j)| \le \max_{\p \in \Omega_{kl}}\left|\Aa_j^\top\p -\Aa_j^\top\p _{kl}\right| \le \bar{a} \max_{\p \in \Omega_{kl}} \|\p-\p_{kl}\|_2$$
holds for all $k=1,...,N-1$ and $l=1,...,l_k.$ Let
$$\mathcal{E}_{kl,2} = \left\{\Bigg \vert\frac{1}{n}\sum_{j=1}^n \Gamma_{kl}(r_j,\Aa_j) - \frac{1}{n}\sum_{j=1}^n \E[\Gamma_{kl}(r_j,\Aa_j)\vert \Aa_1,...,\Aa_n] \Bigg \vert \ge 2\epsilon \bar{a} \max_{\p \in \Omega_{kl}} \|\p-\p_{kl}\|_2\right\}.$$
By applying Hoeffding's Inequality with the independence of $\Gamma_{kl}(r_j,\Aa_j)$'s conditional on $\Aa_1,...,\Aa_n$,
\begin{equation}\prob\left(\mathcal{E}_{kl,2}\vert \Aa_1,...,\Aa_n\right) \le \exp \left(\frac{-n\epsilon^2}{2}\right)
\label{p3hf2}
\end{equation}
for all $\Aa_1,...,\Aa_n,$ $k=1,...,N-1$ and $l=1,...,l_k.$ Consequently, $\prob\left(\mathcal{E}_{kl,2}\right) \le \exp \left(\frac{-n\epsilon^2}{2}\right).$

Next, we handle \textbf{Step 3} and combine the previous parts together. First, we analyze the conditional expectation (\ref{p3hf1}) on the right hand side of $\mathcal{E}_{kl,1}$. Conditional on event $\mathcal{E}_0,$
\begin{align} 
 & \E \left[\frac{1}{n} \sum_{j=1}^n \int_{\bm{a}_j^\top \p_{kl}}^{\bm{a}_j^\top \p^*}\left(I(r_j>v)-I(r_j>\bm{a}_j^\top\p^*)\right)dv \Bigg| \Aa_1,..., \Aa_n, \mathcal{E}_0 \right]\nonumber\\
 =& \ \frac{1}{n} \sum_{j=1}^n \int_{\bm{a}_j^\top \p_{kl}}^{\bm{a}_j^\top \p^*} \prob(r_j>v|\bm{a}_j)- \prob(r_j>\bm{a}_j^\top\p^*|\bm{a}_j) dv \nonumber\\ 
\ge & \ \frac{\lambda}{2n} \sum_{j=1}^n (\bm{a}_j^\top \p^* - \bm{a}_j^\top \p_{kl})^2\nonumber \\
 =&  \ \frac{\lambda}{2}(\p^* - \p_{kl})^\top \left(\frac{1}{n} \sum_{j=1}^n \bm{a}_j \bm{a}_j^\top\right) (\p^* - \p_{kl})\nonumber
 \\
 \ge &\ \frac{\lambda\lambda_{\min}}{4} \left\|\p^* - \p_{kl}\right\|_2^2 \label{dist0}
\end{align}
where the third line comes from Assumption \ref{assump2} (b) and the last line comes from the definition of $\mathcal{E}_0$ (earlier in (\ref{p3prob1})).
Based on the splitting scheme, we find the relationship between $\p_{kl}$, $\bar{\p}_{kl}$, $\p^*,$ and an arbitrary $\p\in\Omega_{kl}.$ Specifically, since the cubes on the $k$-th layer shrink the whole region with a factor of $q^k,$
$$\max_{\p \in \Omega_{kl}} \|\p-\p_{kl}\|_2 = \sqrt{m}(1-q)q^{k-1}\frac{\bar{r}}{\underline{d}},$$ 
$$\left\|\p^* - \p_{kl}\right\|_2 \ge q^{k}\frac{\bar{r}}{\underline{d}},$$
for all $k=1,...,N-1$ and $l=1,...,l_k.$
As a result,
\begin{align}
    \|\p^*-\bar{\p}_{kl}\|_2 & \le \left\|\p^* - \p_{kl}\right\|_2 + \max_{\p \in \Omega_{kl}} \|\p-\p_{kl}\|_2 \nonumber \\
    & \le \left(1 + \frac{\sqrt{m}(1-q)}{q}\right)  \left\|\p^* - \p_{kl}\right\|_2   \label{dist1}
\end{align}
and
\begin{align}
    \max_{\p \in \Omega_{kl}} \|\p-\p_{kl}\|_2 \le \frac{\sqrt{m}(1-q)}{q} \left\|\p^* - \p_{kl}\right\|_2 \le \frac{\sqrt{m}(1-q)}{q} \left\|\p^* - \bar{\p}_{kl}\right\|_2.
    \label{dist2}
\end{align}
With (\ref{dist0}), (\ref{dist1}) and (\ref{dist2}), we can compare the difference (when the event $\mathcal{E}_0$ happens) between the conditional expectations appearing in $\mathcal{E}_{kl,1}$ and $\mathcal{E}_{kl,2},$
\begin{align*}
&\E\left[\frac{1}{n} \sum_{j=1}^n \int_{\bm{a}_j^\top \p_{kl}}^{\bm{a}_j^\top \p^*}\left(I(r_j>v)-I(r_j>\bm{a}_j^\top\p^*)\right)dv\Bigg\vert \Aa_1,...,\Aa_n\right] - \E\left[\frac{1}{n} \sum_{j=1}^n \Gamma_{kl}(r_j,a_j) \Bigg\vert \Aa_1,...,\Aa_n\right] \\
  \ge & \frac{\lambda\lambda_{\min}}{4} \left\|\p^* - \p_{kl}\right\|_2^2 - \mu \bar{a}^2 \|\p^*-\bar{\p}_{kl}\|_2 \max_{\p \in \Omega_{kl}} \|\p-\p_{kl}\|_2  \\
   \ge & \frac{\lambda\lambda_{\min}}{4}\left(\frac{1}{1 + \frac{\sqrt{m}(1-q)}{q}}\right)^2 \|\p^*-\bar{\p}_{kl}\|_2^2  - \mu\bar{a}^2\frac{\sqrt{m}(1-q)}{q} \|\p^*-\bar{\p}_{kl}\|_2^2, \text{\ \ \ (Applying (\ref{dist2}))}
\end{align*}
where the last line represents the difference in a quadratic form of $\|\p^*-\bar{\p}_{kl}\|.$
By choosing 
$$q = \max \left\{\frac{1}{1+\frac{1}{\sqrt{m}}} , \frac{1}{1+\frac{1}{\sqrt{m}}\left(\frac{\lambda\lambda_{\min}}{8\mu \bar{a}^2}\right)^{\frac{1}{3}}} \right\},$$
we have 
\begin{equation}
    \E\left[\frac{1}{n} \sum_{j=1}^n \int_{\bm{a}_j^\top \p_{kl}}^{\bm{a}_j^\top \p^*}\left(I(r_j>v)-I(r_j>\bm{a}_j^\top\p^*)\right)dv\Bigg\vert \Aa_1,...,\Aa_n\right] - \E\left[\frac{1}{n} \sum_{j=1}^n \Gamma_{kl}(r_j,a_j) \Bigg\vert \Aa_1,...,\Aa_n\right] \ge  \frac{\lambda\lambda_{\min}}{32} \left\|\p^* - \bar{\p}_{kl}\right\|_2^2.
    \label{p3hf3}
\end{equation}
Intuitively, the choice of $q$ ensures that the distance between the cube center $\bm{p}_{kl}$ and the furthest point from $\bm{p}_{kl}$ in the cube $\Omega_{kl}$ is dominated by the distance between $\bm{p}_{kl}$ and $\bm{p}^*.$ In this way, the integral from $\Aa_j^\top\bm{p}^*$ to $\Aa_j^\top\bm{p}_{kl}$ in the above will dominate the integral from $\Aa_j^\top\bm{p}_{kl}$ to $\Aa_j^\top\bm{p}$ for any $\bm{p} \in \Omega_{kl}.$ And the former integral  approximately takes a quadratic form based on the concentration argument.

Now, we are ready to return to the main objective and  analyze $\sum_{j=1}^n \int_{\bm{a}_j^\top \p}^{\bm{a}_j^\top \p^*}\left(I(r_j>v)-I(r_j>\bm{a}_j^\top\p^*)\right)dv$. In \textbf{Part 2} of the proof, we decompose it into two parts and derive concentration results for the two parts respectively. The above inequality (\ref{p3hf3}) puts together the expectation terms in the two concentration results. By connecting the dots, we know that
on event $\mathcal{E}_0\cap \mathcal{E}_{kl,1}^c \cap \mathcal{E}_{kl,2}^c$, 
\begin{align}
 & \frac{1}{n} \sum_{j=1}^n \int_{\bm{a}_j^\top \p}^{\bm{a}_j^\top \p^*}\left(I(r_j>v)-I(r_j>\bm{a}_j^\top\p^*)\right)dv \nonumber \\
 \stackrel{(a)}{=}& \frac{1}{n} \sum_{j=1}^n \int_{\bm{a}_j^\top \p_{kl}}^{\bm{a}_j^\top \p^*}\left(I(r_j>v)-I(r_j>\bm{a}_j^\top\p^*)\right)dv + \frac{1}{n} \sum_{j=1}^n \int_{\bm{a}_j^\top \p }^{\bm{a}_j^\top \p_{kl}}\left(I(r_j>v)-I(r_j>\bm{a}_j^\top\p^*)\right)dv\nonumber \\
 \stackrel{(b)}{\ge} & \frac{1}{n} \sum_{j=1}^n \int_{\bm{a}_j^\top \p_{kl}}^{\bm{a}_j^\top \p^*}\left(I(r_j>v)-I(r_j>\bm{a}_j^\top\p^*)\right)dv - \frac{1}{n} \sum_{j=1}^n \Gamma_{kl}(r_j,a_j) \nonumber\\
 \stackrel{(c)}{=} & \frac{1}{n} \sum_{j=1}^n \int_{\bm{a}_j^\top \p_{kl}}^{\bm{a}_j^\top \p^*}\left(I(r_j>v)-I(r_j>\bm{a}_j^\top\p^*)\right)dv - \E\left[\frac{1}{n} \sum_{j=1}^n \int_{\bm{a}_j^\top \p_{kl}}^{\bm{a}_j^\top \p^*}\left(I(r_j>v)-I(r_j>\bm{a}_j^\top\p^*)\right)dv\right]\nonumber \\
 & \ \ - \frac{1}{n} \sum_{j=1}^n \Gamma_{kl}(r_j,a_j) + \E\left[\frac{1}{n} \sum_{j=1}^n \Gamma_{kl}(r_j,a_j) \right] \nonumber\\
 &\ \ + \E\left[\frac{1}{n} \sum_{j=1}^n \int_{\bm{a}_j^\top \p_{kl}}^{\bm{a}_j^\top \p^*}\left(I(r_j>v)-I(r_j>\bm{a}_j^\top\p^*)\right)dv\right] - \E\left[\frac{1}{n} \sum_{j=1}^n \Gamma_{kl}(r_j,a_j) \right] \nonumber\\
 \stackrel{(d)}{\ge} & -2\epsilon\bar{a}\|\p^*-\bar{\p}_{kl}\|_2 + \frac{\lambda\lambda_{\min}}{32} \left\|\p^* - \bar{\p}_{kl}\right\|_2^2\nonumber \\
 \stackrel{(e)}{\ge} & -2\epsilon\bar{a}\|\p^*-{\p}\|_2 + \frac{\lambda\lambda_{\min}}{32} \left\|\p^* - {\p}\right\|_2^2 \label{quad1}
\end{align}
holds for any $\p\in \Omega_{kl}$, $k=1,...,N-1$ and $l=1,...,l_k.$ Here (a) decomposes the integral in two parts. (b) is from the definition of $\Gamma_{kl}.$ (d) comes from applying (\ref{p3hf1}), (\ref{p3hf2}) and (\ref{p3hf3}) to the three lines in (c). (e) comes from the definition of $\bar{\p}_{kl}$ and the property of a quadratic function. To interpret this result (\ref{quad1}), it provides a quadratic lower bound (with high probability) for the integral of interest and the quadratic lower bound has a dominating second-order part plus a linear part with small coefficient $\epsilon$.

Note the above results hold for all the cubes with $k\le N-1$. We need some special treatment for the center cube $\Omega_{N1}.$ Specifically, its center is $\p^*,$ and
\begin{align*}
    \frac{1}{n} \sum_{j=1}^n \int_{\bm{a}_j^\top \p}^{\bm{a}_j^\top \p^*}\left(I(r_j>v)-I(r_j>\bm{a}_j^\top\p^*)\right)dv \ge - \bar{a} \|\p^*-\bar{\p}_{N1}\|_2 = - \bar{a} \sqrt{m}q^N\frac{\bar{r}}{\underline{d}}
\end{align*} 
where the last equality comes from the splitting scheme.

Intuitively, this center cube is not significant because we can always choose $N$ -- the number of layers (of cubes) so that the effect of the center cube is small enough and it is no greater than $\epsilon^2.$
To achieve this, we choose 
$$N = \Bfloor{\log_q \left(\frac{\underline{d}\epsilon^2}{\bar{a}\bar{r}\sqrt{m}}\right)}+1$$
so that
\begin{align}
    \frac{1}{n} \sum_{j=1}^n \int_{\bm{a}_j^\top \p}^{\bm{a}_j^\top \p^*}\left(I(r_j>v)-I(r_j>\bm{a}_j^\top\p^*)\right)dv &\ge -\epsilon^2 \label{quad2}
\end{align} 
for all $\p \in \Omega_{N1}.$

Therefore, we obtain from (\ref{quad1}) and (\ref{quad2}) that, on the event $\cap_{k=1}^N \cap_{l=1}^{l_k}(\mathcal{E}_{kl,1}^c \cap \mathcal{E}_{kl,2}^c) \cap \mathcal{E}_0,$
$$ \frac{1}{n} \sum_{j=1}^n \int_{\bm{a}_j^\top \p}^{\bm{a}_j^\top \p^*}\left(I(r_j>v)-I(r_j>\bm{a}_j^\top\p^*)\right)dv \ge -\epsilon^2 -2\epsilon\bar{a}\|\p^*-{\p}\|_2 + \frac{\lambda\lambda_{\min}}{32} \left\|\p^* - {\p}\right\|_2^2$$
for all $\p\in \bar{\Omega}.$

We complete the proof by computing the probability,
\begin{align*}
 1-\prob\left(\bigcap_{k=1}^N \bigcap_{l=1}^{l_k}\left(\mathcal{E}_{kl,1}^c \bigcap \mathcal{E}_{kl,2}^c\right) \bigcap \mathcal{E}_0\right) & = \prob \left(\bigcup_{k=1}^N \bigcup_{l=1}^{l_k}\left(\mathcal{E}_{kl,1} \bigcup \mathcal{E}_{kl,2}\right) \bigcup \mathcal{E}_0^c\right) \\
 & \le \prob(\mathcal{E}_0^c)+\sum_{k=1}^N \sum_{l=1}^{l_k} \left(\prob(\mathcal{E}_{kl,1}) + \prob(\mathcal{E}_{kl,2})\right)
 \\ & \le m\exp\left(-\frac{n\lambda_{\min}^2}{4\bar{a}^2}\right)+ 2\exp\left(-\frac{n\epsilon^2}{2}\right)\cdot \left(2N\right)^m.  
\end{align*}

\end{proof}

\subsection{Proof of Theorem \ref{expectation}}
\label{AproofDC}

\begin{proof}
Let event
$$\mathcal{E}_1 = \left\{\left\|\frac{1}{n}\sum_{j=1}^n \phi(\p^*,\bm{u}_j) - \nabla f(\p^*)\right\|_2 \le \epsilon \right\}.$$
From Proposition \ref{pro:gradient},
$$\prob\left(\mathcal{E}_1^c \right) \le  2m \exp\left(-\frac{n\epsilon^2}{2\bar{a}^2m}\right).$$
Let event
$$\mathcal{E}_2 = \left\{\frac{1}{n} \sum_{j=1}^n \int_{\bm{a}_j^\top \p}^{\bm{a}_j^\top \p^*}\left(I(r_j>v)-I(r_j>\bm{a}_j^\top\p^*)\right)dv \ge -\epsilon^2 -2\epsilon\bar{a}\|\p^*-{\p}\|_2 + \frac{\lambda\lambda_{\min}}{32} \left\|\p^* - {\p}\right\|_2^2\right\}.$$
From Proposition \ref{pro:Hessian},
$$\prob\left(\mathcal{E}_2^c \right) \le m\exp\left(-\frac{n\lambda_{\min}}{4\bar{a}^2}\right)+ 2\left(2N\right)^m\cdot\exp\left(-\frac{n\epsilon^2}{2}\right) $$
where $N$ is defined in Proposition \ref{pro:Hessian}.

On the event $\mathcal{E}_1 \cap \mathcal{E}_2,$ the following inequality holds for all $\p\in \Omega_{p}$
\begin{equation}
    f_{n}(\p) - f_{n}(\p^*) \ge -\epsilon^2 -\epsilon(2\bar{a}+1)\|\p^*-\p\|_2 +\frac{\lambda\lambda_{\min}}{32} \left\|\p^* - {\p}\right\|_2^2.
    \label{quadLower}
\end{equation} 
This comes from a combination of Lemma \ref{TaylorLemma} with Assumption \ref{assump2} (b) and (c). Specifically, we plug in the event $\mathcal{E}_2$ for the second-order term in Lemma \ref{TaylorLemma} and we apply Assumption \ref{assump2} (b) and (c) with the event $\mathcal{E}_1$ for the first-order term in Lemma \ref{TaylorLemma}. 

By the definition of $\p_n^*$ (that minimizes $f_n$), (\ref{quadLower}) leads to
$$ -\epsilon^2 -\epsilon(2\bar{a}+1)\|\p^*-\p_n^*\|_2 +\frac{\lambda\lambda_{\min}}{32} \left\|\p^* - {\p}_n^*\right\|_2^2 \le f_{n}(\p_n^*) - f_{n}(\p^*)\le 0$$
and this implies,
$$\|\p_n^*-\p^*\|_2 \le \kappa \epsilon$$
with
$$\kappa = \frac{2\bar{a}+1+\sqrt{(2\bar{a}+1)^2+\frac{\lambda\lambda_{\min}}{8}}}{\lambda\lambda_{\min}/16}.$$ 
The intuition for the above result is that on the event $\mathcal{E}_1 \cap \mathcal{E}_2$, the difference $f_{n}(\p_n^*) - f_{n}(\p^*)$ is lower bounded by a convex quadratic function and thus the optimality condition entails $\p^*$ should stay close to $\p_n^*.$
Thus we can view the event $\mathcal{E}_1 \cap \mathcal{E}_2$ as a ``good'' event that ensures $\p^*$ close to $\p^*_n$. The probability of the ``good'' event
\begin{align*}
\prob(\mathcal{E}_1 \cap \mathcal{E}_2) & \ge 1 - \prob(\mathcal{E}_1^c) - \prob(\mathcal{E}_2^c) \\
& \ge 1- 2m \exp\left(-\frac{n\epsilon^2}{2\bar{a}^2m}\right) - m\exp\left(-\frac{n\lambda_{\min}}{4\bar{a}^2}\right) - 2\left(2N\right)^m\cdot\exp\left(-\frac{n\epsilon^2}{2}\right).
\end{align*}
We emphasize that this probability bound holds for all $\epsilon>0.$ So, if we let $\epsilon'=\epsilon^2,$ we have
\begin{align*}
\prob\left(\frac{\|\p_n^*-\p^*\|_2^2}{\kappa^2}>\epsilon'\right) & = \prob\left(\frac{\|\p_n^*-\p^*\|_2^2}{\kappa^2}>\epsilon^2\right)  \\
& \le 2m \exp\left(-\frac{n\epsilon^2}{2\bar{a}^2m}\right) + m\exp\left(-\frac{n\lambda_{\min}}{4\bar{a}^2}\right) + 2\left(2N\right)^m\cdot\exp\left(-\frac{n\epsilon^2}{2}\right) \\
& = 2m \exp\left(-\frac{n\epsilon'}{2\bar{a}^2m}\right) + m\exp\left(-\frac{n\lambda_{\min}}{4\bar{a}^2}\right)+ 2\left(2N\right)^m\cdot\exp\left(-\frac{n\epsilon'}{2}\right).
\end{align*}
With this probability bound, we can upper bound the L$_2$ distance between $\p_n^*$ and $\p^*$ by integration. Specifically, given that $\|\p_n^*-\p^*\|_2 \le\frac{\bar{r}}{\underline{d}},$
\begin{align}
    \frac{1}{\kappa^2}\E \|\p_n^*-\p^*\|_2^2 & =\int_{0}^{\frac{\bar{r}^2}{\underline{d}^2}} \prob\left(\frac{\|\p_n^*-\p^*\|_2^2}{\kappa^2}>\epsilon'\right) \mathrm{d}\epsilon' \nonumber \\
    & \le \int_{0}^{\frac{\bar{r}^2}{\underline{d}^2}} \left(2m \exp\left(-\frac{n\epsilon'}{2\bar{a}^2m}\right) + m\exp\left(-\frac{n\lambda_{\min}}{4\bar{a}^2}\right) + 2\left(2N\right)^m\cdot\exp\left(-\frac{n\epsilon'}{2}\right)\right) \wedge 1 \mathrm{d}\epsilon' \label{integral}
\end{align}
where $y\wedge z = \min\{y,z\}$ for $y,z\in\mathbb{R}.$ Next, we analyze the integral (\ref{integral}) term by term.

First, with $\epsilon' =  \frac{m\log m}{n}\cdot \varepsilon,$
\begin{align}
    & \int_{0}^{\frac{\bar{r}^2}{\underline{d}^2}} \left(2m \exp\left(-\frac{n\epsilon'}{2\bar{a}^2m}\right) \right) \wedge {1}\mathrm{d}\epsilon' \nonumber \\ \le &  \frac{m\log m}{n} \int_{0}^{\infty}\left( 2m \exp\left(-\frac{\varepsilon\log m}{2\bar{a}^2}\right)\right) \wedge {1} \mathrm{d}\varepsilon  \nonumber \\
    \le &\frac{m\log m}{n} \int_{0}^{\infty} 2\left( \exp\left(\log m-\frac{\varepsilon\log m}{2\bar{a}^2}\right) \right)\wedge {1} \mathrm{d}\varepsilon \le c \sqrt{\frac{m\log m}{n}}. 
    \label{probs1}
\end{align}
where $c$ is dependent only on $\bar{a}.$ The inequality on the last line is referred to Lemma \ref{intg1} in the following subsection.

Second, 
\begin{align}
   \int_{0}^{\frac{\bar{r}^2}{\underline{d}^2}} m\exp\left(-\frac{n\lambda_{\min}}{4\bar{a}^2}\right) \mathrm{d}\epsilon' =  \frac{m\bar{r}^2}{\underline{d}^2}\exp\left(-\frac{n\lambda_{\min}}{4\bar{a}^2}\right) \le c' \frac{m}{n}
   \label{probs2}
\end{align}
where $c'$ is dependent only on $\bar{a}$, $\bar{r}$, $\underline{d}$, and $\lambda_{\min}.$

Third, we can show that there exists constant $c_0$ such that 
$$2N \le c_0\sqrt{m}\log\left(\frac{\sqrt{m}}{\epsilon'}\right)$$
from the definition of $N$ in Proposition \ref{pro:Hessian}.
Hence,
\begin{align*}
    \int_{0}^{\frac{\bar{r}^2}{\underline{d}^2}}1 \wedge \left( 2\exp\left(-\frac{n\epsilon'}{2}\right)\cdot \left(2N\right)^m \right) \mathrm{d}\epsilon' &\le \int_{0}^{\infty} 1 \wedge \left( 2\exp\left(-\frac{n\epsilon'}{2}\right)\cdot \left(c_0\sqrt{m}\log\left(\frac{\sqrt{m}}{\epsilon'}\right)\right)^m \right) \mathrm{d}\epsilon'  \\
    & = \int_{0}^{\infty} 1 \wedge\left( 2\exp\left(-\frac{n\epsilon'}{2}+m\log\left(c_0\sqrt{m}\log\left(\frac{\sqrt{m}}{\epsilon'}\right)\right)\right)  \right)\mathrm{d}\epsilon'  \\
   & \le  \int_{0}^{\infty}1 \wedge  \left( 2\exp\left(-\frac{n\epsilon'}{2}+m\log\left(c_0\sqrt{m}\log\left(\frac{\sqrt{m}}{\epsilon'}\right)\right)\right) \right)\mathrm{d}\epsilon'
\end{align*}
Let $\epsilon'= \frac{m\log m \log\log n}{n}\cdot \varepsilon$ and use $\varepsilon$ to replace $\epsilon'$ in above. We have,
\begin{align}
    & \int_{0}^{\frac{\bar{r}^2}{\underline{d}^2}} 1 \wedge \left( 2\exp\left(-\frac{n\epsilon'}{2}\right)\cdot \left(2N\right)^m \right) \mathrm{d}\epsilon' \nonumber \\
\le     & {\frac{m\log m \log\log n}{n}} \int_{0}^{\infty}1 \wedge  \left( 2\exp\left(-\frac{\varepsilon}{2} m
    \log m 
    \log \log n+m\log\left(c_0\sqrt{m}\log\left(\frac{n}{\varepsilon}\right)\right)\right) \right)\mathrm{d}\varepsilon \nonumber \\
    \le & {\frac{c'' m\log m \log\log n}{n}},\label{probs3}
\end{align}
where $c''$ depends only on $c_0$. The inequality in the last line is deferred to Lemma \ref{intg2} in the following subsection.

Combining (\ref{integral}), (\ref{probs1}), (\ref{probs2}) and (\ref{probs3}), we conclude that
$$\E\|\p_n^* - \p^*\|_2^2\le \kappa^2(c+c'+c'')\frac{m\log m \log \log n}{n}$$
holds for all $n>m$ and $\mathcal{P} \in \Xi.$
From the concavity of the square root function, we know
$$\E \|\p_n^* - \p^*\|_2 \le \kappa\sqrt{c+c'+c''}\cdot\sqrt{\frac{m \log m \log \log n}{n}}.$$
Thus we complete the proof.

\end{proof}

\subsubsection{Two inequalities used in the proof of Theorem \ref{expectation}}

We introduce two inequalities that will be used in the proof of Theorem \ref{expectation}. The proof for these two inequalities are based on basic calculus. We use $\wedge$ to denote the minimum operator, i.e., $y\wedge z=\min\{y,z\}.$

\begin{lemma}
The inequality
$$\int_{0}^{\infty} \left( \exp\left(\log m-x\log m\right) \right)\wedge {1} \mathrm{d}x \le 2$$
holds for all $m \ge 2$.
\label{intg1}
\end{lemma}

\begin{proof}
We have
\begin{align*}
      &\int_{0}^{\infty} \left( \exp\left(\log m-x\log m\right) \right)\wedge {1} \mathrm{d}x\\
=     &\int_{0}^{1} \left(\exp\left(\log m-x\log m\right) \right)\wedge {1} \mathrm{d}x + \int_{1}^{\infty} \left(\exp\left(\log m-x\log m\right) \right)\wedge {1} \mathrm{d}x\\
\le & \int_{0}^{1} 1 \mathrm{d}x + \int_{1}^{\infty} \exp\left(\log m-x\log m\right) \mathrm{d}x \text{ \ \ (Splitting the integral in two parts)}\\
\le & 1 + \int_{1}^{\infty} \exp\left(\log 2 \cdot (1- x)\right) \mathrm{d}x \le 1+\frac{1}{\log 2},
\end{align*}
where in the last line the $\log m$ term is replaced with its lower bound $\log 2$.
\end{proof}

\begin{lemma}
The inequality
$$\int_{0}^{\infty}1 \wedge  \left( \exp\left(-xm
\log m 
\log \log n+m\log\left(\sqrt{m}\log\left(\frac{n}{x}\right)\right)\right) \right)\mathrm{d}x  \le 2$$
holds for all $n\ge \max\{m,3\}$ and $m\ge2.$
\label{intg2}
\end{lemma}

\begin{proof}
We have 
\begin{align*}
& \int_{0}^{\infty}1 \wedge  \left( \exp\left(-xm\log m \log \log n+m\log\left(\sqrt{m}\log\left(\frac{n}{x}\right)\right)\right) \right)\mathrm{d}x \\
\le & 1 + \int_{1}^{\infty} \exp\left(-xm\log m \log \log n+m\log\left(\sqrt{m}\log\left(\frac{n}{x}\right)\right)\right)\mathrm{d}x \text{ \ \ (Splitting the integral in two parts)}\\
\le & 1 + \int_{1}^{\infty} \exp\left(-xm\log m \log \log n+m\log\left(\sqrt{m}\log n\right)\right)\mathrm{d}x \text{ \ \ (Using the fact that $x>1$)} \\
\le & 1 + \int_{1}^{\infty} \exp\left(-x m\log m \log \log n+m\log m \log \log n\right)\mathrm{d}x \\
\le & 1 + \frac{1}{m\log m \log \log n}\le 3
\end{align*}
where the last line follows the same argument as the last inequality in Lemma \ref{intg1}.
\end{proof}

\renewcommand{\thesubsection}{B\arabic{subsection}}

\section{Proof for Generic Regret Upper Bound}

\subsection{Proof of Lemma \ref{gP}}

\begin{proof}
First, we show $g(\p^*)$ provides an upper bound for $\E R_n^*.$
\begin{align}
  \E R_n^* & = \E \left[\sum_{j=1}^n r_j x_j^*\right] \nonumber \\
 & = \E \left[n\Dd^\top \p_n^* + \sum_{j=1}^n \left(r_j - \Aa_j^\top \p_n^* \right)^+\right] \text{ \ (From the strong duality)} \nonumber \\  
 & \le \E \left[n\Dd^\top \p^* + \sum_{j=1}^n \left(r_j - \Aa_j^\top \p^* \right)^+\right]  \text{\ (From the optimality of $\p_n^*$)}\nonumber  \\
 & = ng(\p^*)\label{par1}
\end{align}
where the expectation is taken with respect to $(r_j, \Aa_j)$'s. 

Then, by taking the difference between $g(\p^*)$ and $g(\p)$, 
\begin{align*}
    g(\p^*)-g(\p) & = \E\left[r I(r> \Aa^\top \p^*) + \left(\Dd-\Aa I(r>\Aa^\top \p^*)\right)^\top \p^*\right]-\E\left[r I(r> \Aa^\top \p) + \left(\Dd-\Aa I(r>\Aa^\top \p)\right)^\top \p^*\right]\\
    &= \E\left[\left(r-\Aa^\top \p^*\right)\left(I(r> \Aa^\top \p^*)-I(r> \Aa^\top \p)\right)\right]\\
    &= \E\left[\left(\Aa^\top \p^*-r\right)I(\Aa^\top \p^*\ge r> \Aa^\top \p)\right] + \E\left[\left(r-\Aa^\top \p^*\right)I(\Aa^\top \p^*<r\le \Aa^\top \p)\right] \ge 0
\end{align*}
where the expectation is taken with respect to $(r,\bm{a}).$ The last line is true because when the indicator functions are positive, the terms go before the indicators must be non-negative accordingly. This proves the maximum of $g(\p)$ is achieved at $\p^*$. Furthermore, with a more careful analysis, we have 
\begin{align*}
    g(\p^*)-g(\p) 
    &= \E\left[\left(\Aa^\top \p^*-r\right)I(\Aa^\top \p^*\ge r> \Aa^\top \p)\right] + \E\left[\left(r-\Aa^\top \p^*\right)I(\Aa^\top \p^*<r\le \Aa^\top \p)\right] \\
    & \le \E\left[\left(\Aa^\top \p^*-\Aa^\top \p\right)I(\Aa^\top \p^*\ge r> \Aa^\top \p)\right] + \E\left[\left(\Aa^\top \p-\Aa^\top \p^*\right)I(\Aa^\top \p^*<r\le \Aa^\top \p)\right] \\
    & = \E\left[\left(\Aa^\top \p^*-\Aa^\top \p\right)\left(\prob(r>\Aa^\top \p^*|\bm{a})-\prob(r>\Aa^\top \p|\bm{a})\right)I(\Aa^\top \p^*>\Aa^\top \p)\right]\\
    & \ \ \ \ \ + \E\left[\left(\Aa^\top \p-\Aa^\top \p^*\right)\left(\prob(r>\Aa^\top \p|\bm{a})-\prob(r>\Aa^\top \p^*|\bm{a})\right)I(\Aa^\top \p^*<\Aa^\top \p)\right]\\
    & \le \mu \E\left[\left(\Aa^\top \p^*-\Aa^\top \p\right)^2\right] \\
    & \le \mu \bar{a}^2 \|\p^*-\p\|_2^2
\end{align*}
where the expectation is taken with respect to $(r,\bm{a})\sim \mathcal{P}.$ Here the second line is because when $\Aa^\top \p^* \ge r> \Aa^\top \p$ is true, we have $\Aa^\top \p^*-r \le \Aa^\top \p^*-\Aa^\top \p$; the same for the second part of this line. The third line comes from taking conditional expectation with respect to $r$. The fourth line applies Assumption \ref{assump2} (b) and the last line applies the upper bound on $\bm{a}$ in Assumption \ref{assump1} (b).
\end{proof}

\subsection{Proof of Theorem \ref{representation}}

\begin{proof}
For any dual-based online policy $\pi,$ its expected revenue
\begin{align}
   \E R_n(\bm{\pi}) & = \E \left[\sum_{t=1}^n r_t x_t\right] \nonumber \\ 
   & = \E \left[\sum_{t=1}^n r_t x_t + \bm{b}_n^\top \p^* \right]  - \E\left[\bm{b}_n^\top \p^*\right] \nonumber \\
   & = \E \left[\sum_{t=1}^n r_t x_t + \left(n\Dd - \sum_{t=1}^n \Aa_t x_t\right)^\top \p^* \right]  - \E\left[\bm{b}_n^\top \p^*\right] \text{ \ \ (By the definition of $\bm{b}_n$)}\nonumber \\
   & = \E \left[\sum_{t=1}^n \left(  r_t x_t + \Dd^\top \p^* - \Aa_t^\top \p^* x_t\right) \right]  - \E \left[\bm{b}_n^\top \p^* \right] \nonumber  \\
   & = \E \left[\sum_{t=1}^{\tau_{\bar{a}}} \left(  r_t x_t + \Dd^\top \p^* - \Aa_t^\top \p^* x_t\right) \right] + \E \left[\sum_{t=\tau_{\bar{a}}+1}^{n}\left(  r_t x_t + \Dd^\top \p^* - \Aa_t^\top \p^* x_t\right) \right]  - \E \left[\bm{b}_n^\top \p^* \right] \label{decomp1}
\end{align}
where the expectation is taken with respect to $(r_t,\Aa_t)$'s. 

We analyze the first two terms in (\ref{decomp1}) separately. For the first term in (\ref{decomp1}), it could be represented by the Lagrangian function $g(\cdot)$ as in Lemma \ref{gP},
\begin{align}
    \E \left[\sum_{t=1}^{\tau_{\bar{a}}} \left(  r_t x_t + \Dd^\top \p^* - \Aa_j^\top \p^* x_t\right) \right] & =\E \left[\sum_{t=1}^{n} \left(r_t x_t + \Dd^\top \p^* - \Aa_t^\top \p^* x_t\right)I(\tau_{\bar{a}}\ge t) \right] \nonumber  \\
    & \stackrel{(a)}{=}\sum_{t=1}^{n} \E \left[\left(r_t x_t + \Dd^\top \p^* - \Aa_t^\top \p^* x_t\right)I(\tau_{\bar{a}}\ge t) \right] \nonumber \\
     & \stackrel{(b)}{=}\sum_{t=1}^{n}\E \left[ \E \left[\left(r_t x_t + \Dd^\top \p^* - \Aa_t^\top \p^* x_j\right)I(\tau_{\bar{a}}\ge t) | \bm{b}_{t-1}, \mathcal{H}_{t-1}\right]\right]\nonumber  \\
     & \stackrel{(c)}{=} \sum_{t=1}^{n}\E \left[g(\p_{t})I(\tau_{\bar{a}}\ge t) \right] \nonumber \\
     & \stackrel{(d)}{=} \E \left[\sum_{t=1}^{n}g(\p_{t})I(\tau_{\bar{a}}\ge t) \right]  = \E \left[\sum_{t=1}^{\tau_{\bar{a}}}g(\p_{t}) \right] \label{decomp2}
\end{align}
where $\p_t$'s are the dual price vectors specified by the policy $\pi$ and the expectation is taken with respect to $(r_t, \Aa_t)$'s. Here (a) and (d) come from the exchange of summation and expectation. (b) comes from nesting a conditional expectation. (c) is from two facts: first, on the event $\tau_{\bar{a}}\ge t$, the remaining inventory $\bm{b}_{t-1}$ (at the end of time period $t-1$) is enough to satisfy the $t$-th order; second, the dual-based policy is adopting the price vector $\p_t$ in deciding the value of $x_t,$ i.e., $x_t = I(r_t>\Aa_t^\top \p_t).$

For the second term in (\ref{decomp1}), we show that it is lower bounded by $-\E[n-\tau_{\bar{a}}]$ up to some constant. We know that $\|\p^*\| \le \frac{\bar{r}}{\underline{d}}$ from Proposition \ref{basicProp} and $\|\Aa_t\|_2 \le \bar{a}$ from Assumption \ref{assump1}. Combining these two facts,  
\begin{align}
\E \left[\sum_{t=\tau_{\bar{a}}+1}^{n}\left(  r_t x_t + \Dd^\top \p^* - \Aa_t^\top \p^* x_t\right)\right] & \ge \E \left[\sum_{t=\tau_{\bar{a}}+1}^{n}\left( r_t x_t - \Aa_t^\top \p^* x_t\right)\right]  \nonumber \\
&\ge -\E[n-\tau_{\bar{a}}]\cdot \left(\bar{r} + \frac{\bar{r}\bar{a}}{\underline{d}}\right). \label{decomp3}
\end{align}
where the first lines comes from the fact that $\Dd^\top \p^*\ge 0.$

Plugging (\ref{decomp2}) and (\ref{decomp3}) into (\ref{decomp1}), we obtain 
\begin{equation}
    \E R_{n}(\bm{\pi}) \ge \E \left[\sum_{j=1}^{\tau_{\bar{a}}}g(\p_{j}) \right] -\E[n-\tau_{\bar{a}}]\cdot \left(\bar{r} + \frac{\bar{r}\bar{a}}{\underline{d}}\right) - \E\left[\frac{\bar{r}}{\underline{d}}\cdot\sum_{i\in I_B} b_{in}\right]. \label{par2}
\end{equation}

To obtain an upper bound on the regret, we simply take the difference between $\E R_n^*$ and (\ref{par2}), and then apply Lemma \ref{gP} for an upper bound on $\E R_n^*$,
\begin{align*}
    \E R_n^* - \E R_n(\bm{\pi}) \le \E \left[\sum_{j=1}^{\tau_{\bar{a}}}\mu \bar{a}^2\|\p_{j}-\p^*\|_2^2 \right]+\E[n-\tau_{\bar{a}}]\cdot \left(\bar{r} + \frac{\bar{r}\bar{a}}{\underline{d}}\right) + \E\left[\frac{\bar{r}}{\underline{d}}\cdot\sum_{i\in I_B} b_{in}\right]
\end{align*}
holds for all $n>0$ and distribution $\mathcal{P} \in \Xi.$ By choosing
$$K = \max\left\{\mu \bar{a}^2,\bar{r} + \frac{\bar{r}\bar{a}}{\underline{d}},\frac{\bar{r}}{\underline{d}} \right\},$$
we finish the proof.
\end{proof}

\subsection{Proof of Corollary \ref{coro1}}

\begin{proof}
From the proof of Theorem \ref{representation}, the role that the stopping time $\tau_{\bar{a}}$ plays is to guarantee the orders coming before $\tau_{\bar{a}}$ can always be satisfied. When $\prob(\tau \le \tau_{\bar{a}})=1$, the stopping time $\tau$ has the same property. Therefore, the derivations in the proof of Theorem \ref{representation} still hold for $\tau$.
\end{proof}

\renewcommand{\thesubsection}{C\arabic{subsection}}
\section{Regret Analyses for OLP Algorithms}

\subsection{Proof of Theorem \ref{Algo3}}
\label{Ap*}
\begin{proof}


The proof of Theorem \ref{Algo3} builds upon the generic regret upper bound in Theorem \ref{representation}. 

For \textbf{the first part in the generic upper bound} in Theorem \ref{representation}, since we apply $\p^*$ as the decision rule, i.e., $\p_t=\p^*$,
\begin{equation}
    \E\left[\sum_{t=1}^n \|\p_t-\p^*\|_2^2\right] = 0. \label{part1p*}
\end{equation}
So, we only need to focus on $\E\left[n-\tau_{\bar{a}}\right]$ and $\E\left[\sum_{i\in I_B} b_{in}\right].$
Define
$$\tau^i_{\bar{a}} = \min \{n\} \cup \left\{t\ge 1:\sum_{j=1}^t a_{ij}I(r_j>\Aa_j^\top\p^*)> nd_i-\bar{a}\right\}.$$ 
From the optimality condition on $\bm{p}^*$, we know that the expected constraint consumption of Algorithm \ref{alg:Distribution} under the dual price $\bm{p}^*$ has an upper bound 
\begin{align*}
    \E\left[ \sum_{j=1}^t a_{ij}I(r_j>\Aa_j^\top\p^*)\right]\le td_i
\end{align*}
for $i=1,...,m$ and $t=1,...,n$. That is, $\nabla f(\bm{p}^*)\ge 0,$ and this can be derived from the proof of Proposition \ref{strConvex}. In addition, the variance of the constraint consumption has a trivial upper bound due to the independence across different time periods,
\begin{align*}
    \text{Var}\left[ \sum_{j=1}^t a_{ij}I(r_j>\Aa_j^\top\p^*)\right] \le \bar{a}^2t.
\end{align*}
for $i=1,...,m$ and $t=1,...,n$. In the following, we use these two upper bounds to derive upper bounds for the second and third part of the generic upper bound in Theorem \ref{representation}.

For \textbf{the second part of the generic upper bound},
\begin{align}
    \E[n-\tau_{\bar{a}}^i]  \le & \sum_{t=1}^n \prob(\tau_{\bar{a}}^i\le t) \nonumber\\
    = &\sum_{t=1}^n \prob\left(\sum_{j=1}^t a_{ij}I(r_j>\Aa_j^\top\p^*) \ge nd_i-\bar{a} \right) \nonumber\\
    \le &\sum_{t=1}^n \prob\left(\sum_{j=1}^t a_{ij}I(r_j>\Aa_j^\top\p^*)  - \E\left[ \sum_{j=1}^t a_{ij}I(r_j>\Aa_j^\top\p^*)\right] \ge (n-t)d_i-\bar{a} \right) \nonumber\\
    \le & n-n_0+\sum_{t=1}^{n_0}\left(\frac{\bar{a}^2t}{((n-t)d_i-\bar{a})^2}\right) \wedge 1  \nonumber \text{ \ (Applying the Chebyshev's inequality)} \\
    \le & \left(2+\frac{\bar{a}}{d_i}\right)\sqrt{n} \nonumber
\end{align}
where $n_0=\floor{n-\frac{\bar{a}}{d_i}}$. We refer the last line to Lemma \ref{lemmaTheoremAlgo3} in the following subsection. Therefore, 
\begin{align}
    \E[n-\tau_{\bar{a}}] & = \E[\max_i \{n-\tau_{\bar{a}}^{i}\}] \nonumber \\
    & \le \sum_{i=1}^m \E[n-\tau_{\bar{a}}^{i}] \nonumber \\
    & \le \left(2+\frac{\bar{a}}{\underline{d}}\right)m\sqrt{n}
    \label{part2p*}
\end{align} 
where the first line comes from the definition of $\tau_{\bar{a}}$ and $\tau_{\bar{a}}^i$ and the second line comes from a replacement of the maximum of $m$ (non-negative) random variables with an upper bound of their summation. Next, for \textbf{the third part of the generic upper bound},
\begin{align}
\E[b_{in}] & = \E\left[\left(nd_i - \sum_{j=1}^n a_{ij}I(r_j>\Aa_j^\top\p^*)\right)^+ \right]\le\E\left[ \Bigg|nd_i - \sum_{j=1}^n a_{ij}I(r_j>\Aa_j^\top\p^*)\Bigg|\right]\nonumber \\
& = \E\left[\Bigg|\sum_{j=1}^n \left(d_i-a_{ij}I(r_j>\Aa_j^\top\p^*)\right)\Bigg|\right]\nonumber \\
&\le \sqrt{\E\left[\Bigg|\sum_{j=1}^n \left(d_i-a_{ij}I(r_j>\Aa_j^\top\p^*)\right)\Bigg|^2\right]}\nonumber \\
& =\sqrt{\text{Var}\left[\sum_{j=1}^n a_{ij}I(r_j>\Aa_j^\top\p^*)\right]} \nonumber \\
&\le \bar{a}\sqrt{n}
    \label{part3p*}
\end{align}
holds for all $i \in I_B$ and $n>0$. The second line comes from absorbing the term $nd_i$ into the summation, the fourth line comes from the fact that $d_i=\E[a_{ij}I(r_j>\Aa_j^\top \bm{p}^*)]$ for a binding constraint $i$, and the last line comes from the independence between different time periods.
Combining (\ref{part1p*}), (\ref{part2p*}), and (\ref{part3p*}) with Theorem \ref{representation}, we complete the proof.
\end{proof}

\subsubsection{Inequality in the Proof of Theorem \ref{Algo3}}

\begin{lemma}
\label{lemmaTheoremAlgo3}
The following inequality holds
\begin{align*}
n-n_0+\sum_{t=1}^{n_0}\left(\frac{\bar{a}^2t}{((n-t)d_i-\bar{a})^2}\right) \wedge 1 \le \left(2+\frac{\bar{a}}{d_i}\right)\sqrt{n} \nonumber
\end{align*}
where $n_0=\floor{n-\frac{\bar{a}}{d_i}}$.
\end{lemma}

\begin{proof}
 We have
 \begin{align*}
& n-n_0+\sum_{t=1}^{n_0}\left(\frac{\bar{a}^2t}{((n-t)d_i-\bar{a})^2}\right) \wedge 1 \\
= & n-n_0+\sum_{t=1}^{n_0}\left(\frac{t}{\left((n-t)\frac{d_i}{\bar{a}}-1\right)^2}\right) \wedge 1 \\
\le & \sqrt{n}+1+\sum_{t=1}^{\floor{n-\sqrt{n}}}\left(\frac{t}{\left((n-t)\frac{d_i}{\bar{a}}-1\right)^2}\right) \\
\le &\sqrt{n}+1+\sum_{t=1}^{\floor{n-\sqrt{n}}}\left(\frac{2t}{\left((n-t)\frac{d_i}{\bar{a}}\right)^2}\right) \\
\le & \sqrt{n}+1 + \int_{0}^{n-\sqrt{n}} \frac{2x}{\left((n-x)\frac{d_i}{\bar{a}}\right)^2} \mathrm{d}x \\
\le & \sqrt{n}+1 + (n-\sqrt{n}) \cdot \int_{0}^{n-\sqrt{n}} \frac{1}{\left((n-x)\frac{d_i}{\bar{a}}\right)^2} \mathrm{d}x \\
= & \sqrt{n}+1 + (n-\sqrt{n}) \frac{\bar{a}}{d_i\sqrt{n}} \le \left(2+\frac{\bar{a}}{d_i}\right)\sqrt{n}.
\end{align*}
\end{proof}

\subsection{Proof of Theorem \ref{Algo1}}

\label{AsqrtT}

\begin{proof}
We prove the theorem by using the results in Theorem \ref{representation} and Corollary \ref{coro1}. The proof is similar but more complicated than Theorem \ref{Algo3}, because the dual price is computed based on SAA in Algorithm \ref{alg:IDLA}. Specifically, we analyze the three parts in the generic regret upper bound separately. 

First, define $$\tau^i_{\bar{a}} = \min \{n\} \cup \left\{t\ge 1:\sum_{j=1}^t a_{ij}I(r_j>\Aa_j^\top\p_j)> nd_i-\bar{a}\right\}$$ where $\p_j$'s are specified by Algorithm \ref{alg:IDLA}. Here the stopping time $\tau^i_{\bar{a}}$ is associated with the constraint process under policy $\pi_1$. In this way, 
$$\tau_{\bar{a}} = \min_{i}\tau^i_{\bar{a}}.$$
From the dual convergence result in Theorem \ref{expectation}, we know that there exists a constant $C,$ such that
$$\E \|\p_{t_k} - \p^*\|_2^2 \le  \frac{Cm}{t_k}\log\log t_k$$
holds for all $k\ge1$ and distribution $\mathcal{P}\in \Xi.$ Here $\p_k$ is the dual price used in Algorithm \ref{alg:IDLA} and $t_k$ appears on the right-hand-side instead of $k$ because Algorithm \ref{alg:IDLA} updates the dual price only in periods $t_k$'s.

First, for \textbf{the first part of the generic regret bound}, 
\begin{align}
    \E\left[\sum_{j=1}^{\tau_{\bar{a}}} \|\p_j-\p^*\|_2^2 \right] & \le \sum_{j=1}^n \E\|\p_j-\p^*\|_2^2 \nonumber \\
    & \le \sum_{k=1}^{L-1} \sum_{t=t_k+1}^{t_{k+1}} \E \|\p_{t_k} - \p^* \|_2^2 \nonumber \\
    & \le \sum_{k=1}^{L-1} (t_{k+1} - t_k)\cdot \frac{Cm}{{t_k}} \log\log t_k \nonumber \\
    & \le \sum_{k=1}^{L-1} \frac{Cm(\delta^{k+1}-\delta^{k} +1)}{\delta^{k}} \log\log n \text{ \ \  (Plugging in $t_k$'s value)}\nonumber \\
    & =  2Cm(\delta-1)L\log\log n \le 3Cm\log n\log \log n, \label{part1sqrt}
\end{align}
where the last line comes from the fact that $\delta\in (1,2]$ and $n = \floor{\delta^L}.$ As we will see shortly, the contribution of this first part is daunted by the later two parts.

Next, we analyze \textbf{the second part in the generic regret bound} -- the stopping time $\tau_{\bar{a}}$. Specifically, we consider the constraint process $b_{it}.$ From the definition of $\tau_{\bar{a}}^{i},$
\begin{equation}
\left\{\tau_{\bar{a}}^{i} \le t\right\} = \left\{\sum_{j=1}^{t'} a_{ij}I(r_j>\Aa_j^\top\p_j)\ge nd_i-\bar{a} \text{ for some } 1\le t'\le t\right\}   \label{tau_t}
\end{equation}
where $\p_j$'s are specified by Algorithm \ref{alg:IDLA}. To obtain an upper bound of $\E[n-\tau_{\bar{a}}^i]$, we only need to analyze the probability of the event on the right hand side. Notice that from the optimality condition of the stochastic programming problem,
$$\sum_{j=1}^t \E\left[a_{ij}I(r_j>\Aa_j^\top\p^*)\right] \le td_i$$
where the expectation is taken with respect to $(r_j,\Aa_{j})\sim \mathcal{P}.$ The equality holds for the binding constraints (where the binding and non-binding constraints are defined according to the stochastic program \eqref{asymProblem}). It tells that if we apply the dual price $\bm{p}^*$, then the constraints will not be exhausted until the last step. The idea is to use the fact that $\bm{p}_j$'s are close to $\bm{p}^*$ to show that if we apply the dual price $\bm{p}_j$'s, the constraints will also not be exhausted until the very end of the horizon.

Define function
$$g_0(\p) = \E\left[a_{ij}I(r_j>{\Aa}_j^\top\p)\right].$$
We know 
\begin{align}
    |g_0(\p) - g_0(\p^*)| & = |\E\left[a_{ij}I(r_j>{\Aa}_j^\top\p)\right]- \E\left[a_{ij}I(r_j>{\Aa}_j^\top\p^*)\right] |\nonumber\\
    & = |\E\left[a_{ij}\prob(r_j>{\Aa}_j^\top\p|{\Aa}_j)-a_{ij}\prob(r_j>{\Aa}_j^\top\p^*|{\Aa}_j)\right] |\nonumber\\
    & \le \E\left[|a_{ij}\prob(r_j>{\Aa}_j^\top\p|{\Aa}_j)-a_{ij}\prob(r_j>{\Aa}_j^\top\p^*|{\Aa}_j)|\right] \nonumber\\
    & \le \bar{a}  \E\left[|\prob(r_j>{\Aa}_j^\top\p|{\Aa}_j)-\prob(r_j>{\Aa}_j^\top\p^*|{\Aa}_j)|\right]  \nonumber \\ 
    & \le \bar{a}^2 \mu \|\p -\p^*\|_2 \label{g0P}
\end{align}
for any $\p\in \Omega_p$ and distribution $\mathcal{P} \in \Xi.$ The parameter $\mu$ in the last line comes from Assumption \ref{assump2} (b). The above inequality states that the difference in terms of constraint consumption is upper bounded by the difference between dual prices (up to a constant factor).

The expectation of $i$-th constraint consumption up to time $t$ under Algorithm \ref{alg:IDLA},
\begin{align}
\E\left[ \sum_{j=1}^t a_{ij}I(r_j>\Aa_j^\top\p_j)\right]  &=\sum_{j=1}^t \E\left[a_{ij}I(r_j>{\Aa}_j^\top\p_j) \right]\nonumber \\
& \le   \sum_{j=1}^t \left(\E\left[a_{ij}I(r_j>{\Aa}_j^\top\p_j)\right]-\E \left[a_{ij}I(r_j>{\Aa}_j^\top\p^*)\right]\right) + td_i \nonumber\\
& \le \bar{a}^2\mu \sum_{j=1}^t  \E\|\p_j-\p^*\|_2 + td_i \nonumber\\
& = \bar{a}^2\mu \sum_{k=1}^L \sum_{j=t_k+1}^{t_{k+1}} \E \|\p_j - \p^* \|_2I(j\le t) + td_i \nonumber\\
& \le \bar{a}^2\mu \sum_{k=1}^L \sum_{j=t_k+1}^{t_{k+1}} \frac{C\sqrt{m}}{\sqrt{t_k}} \sqrt{\log\log t_k} I(j\le t)+ td_i \nonumber \\
& \le 5C\bar{a}^2\mu\sqrt{m}\sqrt{t} \sqrt{\log \log t} + td_i, \label{expAlg1}
\end{align}
for $i=1,...,m$ and $t=1,...,n.$ The constant $C$ is the coefficient of the dual convergence from Theorem \ref{expectation}. Here the second line comes from plugging in the feasibility of $\bm{p}^*$ for the stochastic program. The third line comes from the analysis  of the function $g_0(\bm{p}).$  The fifth line comes from the dual convergence result. The detailed derivation of the last line is defered to Lemma \ref{lemmasqrtT} in the following subsection.

The variance of $i$-th constraint consumption up to time $t$ has the following decomposition,
\begin{align}
    \text{Var}\left[ \sum_{j=1}^t a_{ij}I(r_j>\Aa_j^\top\p_j)\right] & = \E \left[ \sum_{j=1}^t a_{ij}I(r_j>\Aa_j^\top\p_j) - \sum_{j=1}^t \E\left[ a_{ij}I(r_j>\Aa_j^\top\p_j)|\p_j\right]\right]^2  \nonumber \\ 
    & \ \ \ +  \text{Var} \left[\sum_{j=1}^t \E[ a_{ij}I(r_j>\Aa_j^\top\p_j)|\p_j]\right]. \label{var1}
\end{align}
This is because, for two random variables $X_1$ and $X_2,$ we have
\begin{align*}
    \text{Var}[X_1] &= \E[X_1-\E X_1]^2 \\
    & = \E\left[X_1-\E [X_1|X_2] +\E [X_1|X_2] -\E X_1 \right]^2 \\
    & = \E\left[X_1 - \E[X_1|X_2]\right]^2 + \text{Var}[\E[X_1|X_2]].
\end{align*}
Let $Z_j \coloneqq a_{ij}I(r_j>\Aa_j^\top\p_j) - \E\left[ a_{ij}I(r_j>\Aa_j^\top\p_j)|\p_j\right].$ It is easy to see that $Z_j$'s is a martingale difference sequence adapted to $\{\mathcal{H}_t\}_{t=1}^n.$ Recall that $\mathcal{H}_t$ is defined as the $\sigma$-algebra generated by $\{(r_j, \Aa_j)\}_{j=1}^{t}.$ Specifically, we have 
$$\E[Z_j] < \infty\text{\ \ and \ \ } \E[Z_j|\mathcal{H}_{j-1}]=0.$$
Then, for the first term in (\ref{var1}),
\begin{align}
   & \E \left[ \sum_{j=1}^t a_{ij}I(r_j>\Aa_j^\top\p_j) - \sum_{j=1}^t \E\left[ a_{ij}I(r_j>\Aa_j^\top\p_j)|\p_j\right]\right]^2 \nonumber \\
     = &\sum_{j=1}^t \E\left[a_{ij}I(r_j>\Aa_j^\top\p_j)-\E\left[ a_{ij}I(r_j>\Aa_j^\top\p_j)|\p_j\right]\right]^2 \le \bar{a}^2t.  \label{mart1}
\end{align}
For the second term in (\ref{var1}),
\begin{align}
    \text{Var} \left[\sum_{j=1}^t \E[ a_{ij}I(r_j>\Aa_j^\top\p_j)|\p_j]\right] &\le  \E \left[\sum_{j=1}^t \E[ a_{ij}I(r_j>\Aa_j^\top\p_j)|\p_j] - \sum_{j=1}^t \E[ a_{ij}I(r_j>\Aa_j^\top\p^*)]\right]^2 \nonumber \\
    & \le \E \left[\bar{a}^2\mu\sum_{j=1}^t \|\p_j-\p^*\|_2 \right]^2 \nonumber \\
    & \le  C\bar{a}^4\mu^2 mt\log t\log \log t \label{mart2}
\end{align}
where the last line is referred to Lemma \ref{sqrtSum} in the following subsection and the constant $C$ is the coefficient of the dual convergence from Theorem \ref{expectation}.
Putting together (\ref{mart1}) and (\ref{mart2}),
\begin{equation}
    \text{Var} \left[\sum_{j=1}^t a_{ij}I(r_j>\Aa_j^\top\p_j)\right] \le \bar{a}^2 t + C\bar{a}^4\mu^2 mt\log t\log \log t,
    \label{varAlg1}
\end{equation}
for $i=1,...,m$ and $t=1,...,n.$

To summarize, (\ref{expAlg1}) and (\ref{varAlg1}) provide upper bounds for the expectation and variance of the constraint consumption under Algorithm \ref{alg:IDLA}. With these two bounds, we can proceed to analyze the right-hand-side of (\ref{tau_t}),
\begin{align}
&\mathbb{P}\left(\sum_{j=1}^{t'} a_{ij}I(r_j>\Aa_j^\top\p_j)\ge nd_i-\bar{a} \text{ for some } 1\le t'\le t\right)  \nonumber \\
& = \prob\Bigg(\sum_{j=1}^{t'} a_{ij}I(r_j>\Aa_j^\top\p_j) - \E\left[\sum_{j=1}^{t'} a_{ij}I(r_j>\Aa_j^\top\p_j)\right] \ge nd_i-\bar{a} - \E\left[\sum_{j=1}^{t'} a_{ij}I(r_j>\Aa_j^\top\p_j) \right]   \nonumber \\ & \ \ \text{for some \ } 1\le t' \le t \Bigg) \nonumber\\
& \le \prob\Bigg(\sum_{j=1}^{t'} a_{ij}I(r_j>\Aa_j^\top\p_j) - \E\left[\sum_{j=1}^{t'} a_{ij}I(r_j>\Aa_j^\top\p_j)\right] \ge (n-t)d_i-\bar{a} - 5C\bar{a}^2\mu\sqrt{m}\sqrt{t} \sqrt{\log \log t}   \nonumber \\ & \ \ \text{for some \ } 1\le t' \le t \Bigg) \label{MartingaleRepre}\end{align}
where the last line comes from plugging in the upper bound (\ref{expAlg1}) of the expected constraint consumption.

We can view the process
$$M_{t} = \sum_{j=1}^{t} a_{ij}I(r_j>\Aa_j^\top\p_j) - \E\left[\sum_{j=1}^{t} a_{ij}I(r_j>\Aa_j^\top\p_j)\right]$$
as a martingale adapted to the filtration $\mathcal{H}_{t}$ generated by $\{(r_{j},\Aa_j)\}_{j=1}^{t}$. Applying Doob's martingale inequality, when $(n-t)d_i-\bar{a} - 5C\bar{a}^2\mu\sqrt{m}\sqrt{t} \sqrt{\log \log t} >0,$ we have,
\begin{align*}
    & \mathbb{P}\left(\sum_{j=1}^{t'} a_{ij}I(r_j>\Aa_j^\top\p_j)\ge nd_i-\bar{a} \text{ for some } 1\le t'\le t\right) \\
    & \le \prob\left(M_{t'} \ge (n-t)d_i-\bar{a} - 5C\bar{a}^2\mu\sqrt{m}\sqrt{t} \sqrt{\log \log t} \text{ for some } 1\le t'\le t\right ) \\ & \le \frac{ \text{Var} \left[\sum_{j=1}^t a_{ij}I(r_j>\Aa_j^\top\p_j)\right]}{\left((n-t)d_i-\bar{a} - 5C\bar{a}^2\mu\sqrt{m}\sqrt{t} \sqrt{\log \log t}\right)^2} \\
    & \le \frac{\bar{a}^2 t + C\bar{a}^4\mu^2 mt\log t\log \log t}{\left((n-t)d_i-\bar{a} - 5C\bar{a}^2\mu\sqrt{m}\sqrt{t} \sqrt{\log \log t}\right)^2}.
\end{align*}
where the second line is (\ref{MartingaleRepre}), the third line applies Doob's Martingale inequality \citep{revuz2013continuous} and the fourth line plugs in the variance upper bound (\ref{varAlg1}).

Then, we complete the analysis of \textbf{the second part of the generic upper bound}, 
\begin{align}
    \E[n-\tau_{\bar{a}}^i]  \le & \sum_{t=1}^n \prob(\tau_{\bar{a}}^i\le t) \nonumber\\
    = &\sum_{t=1}^n  \mathbb{P}\left(\sum_{j=1}^{t'} a_{ij}I(r_j>\Aa_j^\top\p_j)\ge nd_i-\bar{a} \text{ for some } 1\le t'\le t\right) \nonumber\\
    \le & n-n_0+\sum_{t=1}^{n_0}\left(\frac{\bar{a}^2 t + C\bar{a}^4\mu^2m t\log t\log \log t}{\left((n-t)d_i-\bar{a} - 5C\bar{a}^2\mu\sqrt{m}\sqrt{t} \sqrt{\log \log t}\right)^2}\right) \wedge 1 \nonumber \\
    \le & C'\sqrt{m}\sqrt{n} \log n  \nonumber
\end{align}
for some constant $C'$ dependent on $\bar{a}$, $\underline{d}$, $\bar{d}$, $\mu$ and $C.$ Here $n_0$ is the largest index $t$ such that $(n-t)d_i-\bar{a} - 5C\bar{a}^2\mu\sqrt{m}\sqrt{t} \sqrt{\log \log t} >0,$. We refer the derivation of the last line to Lemma \ref{lemmaTheoremAlgo2} in the following subsection. 
Then,
\begin{align}
    \E[n-\tau_{\bar{a}}]  &=   \E[\max_i \{n-\tau_{\bar{a}}^{i}\}] \nonumber \\
    & \le \sum_{i=1}^m \E[n-\tau_{\bar{a}}^{i}] \nonumber \\
    & \le C' m^{\frac{3}{2}} \sqrt{n} \log n. \label{part2sqrt}
\end{align}
where the first two lines follow the same argument as \eqref{part2p*} in Theorem \ref{Algo3}. The last thing is \textbf{the third part of the generic regret bound} -- $\E\left[b_{in}\right]$ for $i \in I_B.$ Indeed,
$$b_{in} = \left(nd_i - \sum_{j=1}^n a_{ij}I(r_j>\Aa_j^\top\p_j)\right)^+ \le \Bigg|nd_i - \sum_{j=1}^n a_{ij}I(r_j>\Aa_j^\top\p_j)\Bigg|.$$
Therefore,
\begin{align}
    \E\left[b_{in}\right] & \le \E\left[\Bigg|nd_i - \sum_{j=1}^n a_{ij}I(r_j>\Aa_j^\top\p_j)\Bigg|\right]\nonumber  \\ 
    & \le \sqrt{\E\left[\Bigg|nd_i - \sum_{j=1}^n a_{ij}I(r_j>\Aa_j^\top\p_j)\Bigg|^2\right]} \nonumber \\
    & = \sqrt{\left(\E\left[nd_i - \sum_{j=1}^n a_{ij}I(r_j>\Aa_j^\top\p_j)\right]\right)^2 + \text{Var}\left[nd_i - \sum_{j=1}^n a_{ij}I(r_j>\Aa_j^\top\p_j)\right]} \nonumber \\
    & \le \sqrt{25C^2\bar{a}^4\mu^2m^2n\log\log n + \bar{a}^2 n + C\bar{a}^4\mu^2 mn\log n\log \log n} \nonumber \\
    & \le 6C\bar{a}^2\mu m\sqrt{n} \sqrt{\log\log n} \label{part3sqrt}
\end{align}
where the second line comes from Cauchy-Schwartz inequality and the fourth line comes from plugging in the upper bounds on expectation and variance of constraint consumption, namely, (\ref{expAlg1}) and (\ref{varAlg1}). 
Combining the three inequalities (\ref{part1sqrt}), (\ref{part2sqrt}), and (\ref{part3sqrt}) which correspond to the three parts in the upper bound from Theorem \ref{representation}, we complete the proof.
\end{proof}

\subsubsection{Inequalities in the Proof of Theorem \ref{Algo1}}

\begin{lemma}
The following inequality holds for $3\le t<n,$
 \begin{align*}
 \sum_{k=1}^L \sum_{j=t_k+1}^{t_{k+1}} \frac{1}{\sqrt{t_k}} \sqrt{\log\log t_k} I(j\le t) \le  5\sqrt{t}\sqrt{\log\log t}.
\end{align*}
\label{lemmasqrtT}
\end{lemma}

\begin{proof}
Let $L'$ be the smallest integer such that $\delta^{L'}>t.$ By its definition, $\delta^{L'-1}\le t<\delta^{L'}.$
\begin{align*}
 & \sum_{k=1}^L \sum_{j=t_k+1}^{t_{k+1}} \frac{1}{\sqrt{t_k}} \sqrt{\log\log t_k} I(j\le t) \\
 \le & \sum_{k=1}^{L'} \sum_{j=t_k+1}^{t_{k+1}} \frac{1}{\sqrt{t_k}} \sqrt{\log\log t}  \\
 =& \sum_{k=1}^{L'}  \frac{t_{k+1}-t_k}{\sqrt{t_k}} \sqrt{\log\log t} \\
 =& \sum_{k=1}^{L'}  \frac{\delta-1}{\delta^{k/2}} \sqrt{\log\log t} \text{ \ \ \ (Plugging in the value of $t_k$'s)}\\
  =& \sum_{k=1}^{L'}  (\delta-1)\delta^{k/2} \sqrt{\log\log t} \\
  =&  \delta^{1/2}(\delta-1)\frac{\delta^{L'/2}-1}{\delta^{1/2}-1} \sqrt{\log\log t}  \text{ \ \ (Using the fact that $\delta^{L'-1}\le t$)}\\
  \le &  \delta(\delta^{1/2}+1)\sqrt{t}\sqrt{\log\log t} \\ \le& 5\sqrt{t}\sqrt{\log\log t}
\end{align*}
where the last line comes from the choice of $\delta\in(1,2].$
\end{proof}

\begin{lemma}
The following inequality holds for $3\le t\le n,$
$$\E \left[\sum_{j=1}^t \|\p_j-\p^*\|_2 \right]^2\le Cmt \log t \log \log t$$
where $C$ is the coefficient of dual convergence in Theorem \ref{expectation}.
\label{sqrtSum}
\end{lemma}
We have 
\begin{align*}
  \E \left[\sum_{j=1}^t \|\p_j-\p^*\|_2 \right]^2  \le& t \E \left[\sum_{j=1}^t \|\p_j-\p^*\|_2^2 \right] \\
   &\text{(Use the dual convergence result)} \\
  \le& t \sum_{k=1}^L \sum_{j=t_k+1}^{t_{k+1}} \frac{Cm}{t_k} \log\log t_k I(j\le t)\\
  &\text{(Use the same method as Lemma \ref{lemmasqrtT})} \\
  \le& t\sum_{k=1}^{L'} \sum_{j=t_k+1}^{t_{k+1}} \frac{Cm}{t_k} \log\log t_k \\
 = &t \sum_{k=1}^{L'} \frac{Cm(t_{k+1}-t_k)}{t_k} \log\log t_k  \\
  = &t\sum_{k=1}^{L'} Cm(\delta-1)\log\log t_k \\
  \le & Cmt \log t \log \log t
\end{align*}
where $L'$ is the smallest integer such that $\delta^{L'}>t.$ By this definition, $\delta^{L'-1}\le t<\delta^{L'}.$

\begin{lemma}
\label{lemmaTheoremAlgo2}
The following inequality holds
\begin{align*}
 n-n_0+\sum_{t=1}^{n_0}\left(\frac{mt \log t\log\log t}{\left(n-t - \sqrt{m}\sqrt{t}\sqrt{\log \log t}\right)^2}\right) \wedge 1 
    \le &6\sqrt{mn\log n\log \log n}
\end{align*} 
where $n_0$ is the largest index $t$ such that $n-t - \sqrt{m}\sqrt{t}\sqrt{\log \log t}>0$.
\end{lemma}

\begin{proof}
First, we choose $n_1 = n-3\sqrt{mn\log n\log \log n}.$ It is easy to show that $n_0<n_1.$ Without loss of generality, we assume $n_1>0.$
\begin{align*}
   &  n-n_0+\sum_{t=1}^{n_0}\left(\frac{mt \log t\log\log t}{\left(n-t - \sqrt{m}\sqrt{t}\sqrt{\log \log t}\right)^2}\right)  \wedge 1 \\
   \le &  n-n_1+\sum_{t=1}^{n_1}\frac{mt \log t\log\log t}{\left(n-t - \sqrt{m}\sqrt{t}\sqrt{\log \log t}\right)^2} \\
    \le &  3\sqrt{mn\log n\log \log n} + mn\log n\log\log n\sum_{t=1}^{n_1}\frac{1}{\left(n-t - \sqrt{m}\sqrt{t}\sqrt{\log \log t}\right)^2} \\
    \le  & 3\sqrt{mn\log n\log \log n} + mn\log n\log\log n \frac{3}{\sqrt{mn\log n\log \log n}} 
    \\ 
    \le& 6\sqrt{mn\log n\log \log n}.
\end{align*}
\end{proof}

\subsection{Proof of Theorem \ref{Algo2}}

\label{AlogT}

\begin{proof}

As the proof of Theorem \ref{Algo3} and Theorem \ref{Algo1}, we utilize again the generic upper bound in Theorem \ref{representation}. We first define the stochastic process of the constraint consumption. 
Define $\bm{d}_t = (d_{1,t},...,d_{m,t})^\top$ where $$d_{it} \coloneqq \frac{b_{it}}{n-t}$$ as the remaining resource per period after the end of $t$-th period, for $i=1,...,m$ and $t=1,...,n-1$. The reason that we analyze the process $\bm{d}_{t}$ instead of the original process $\bm{b}_t$ is that a more careful analysis is required in the proof. Specifically, we define $\bm{d}_0=\bm{d}$.

Throughout this proof, we will reserve the vector  $\bm{d}$ to denote the initial average resource and use $\tilde{\bm{d}}$ and $\bm{d}'$ to denote a general vector in $\Omega_d.$ Assumption \ref{assump3} states uniform conditions on $\tilde{\bm{d}}\in \Omega_d.$ In Lemma \ref{d-p} and Lemma \ref{delta_d}, we first build on Assumption \ref{assump3} and establish the existence of $\delta_d>0$ such that  $\mathcal{D}\coloneqq\bigotimes_{i=1}^m [d_i-\delta_d,d_i+\delta_d]$ and for all $\tilde{\bm{d}} \in \mathcal{D}$, the stochastic program $f_{\tilde{\bm{d}}}(\bm{p})$ specified with $\tilde{\bm{d}}$ shares the same binding and non-binding sets ($I_B$ and $I_N$) with the original $\bm{d}$. In comparison, Assumption \ref{assump3} defines a set $\Omega_d$, and here $\mathcal{D}\subset \Omega_d$ is a set of $\bm{d}'$ that not only satisfies Assumption \ref{assump3} but also shares the same binding and non-binding dimensions with the initial $\bm{d}.$ 

The property of sharing binding/non-binding dimensions will create great convenience in this proof. Intuitively, the existence of $\delta_d$ is due to the continuity of $$f_{\bm{d}'}(\bm{p}) \coloneqq \bm{d}'^\top \bm{p} + \E_{(r,\bm{a})\sim\mathcal{P}}\left[(r-\bm{a}^\top\bm{p})^+\right]$$ with respect to $\bm{d}'$ together with Assumption \ref{assump3} (b).  We emphasize that the constant $\delta_d$ is pertaining to the stochastic program and it is not dependent on $n$. The statement and the proof of Lemma \ref{d-p} and Lemma \ref{delta_d} are deferred to the following subsection.

Now, we define a stopping time $\tau$ according to $\mathcal{D}.$ As we will see, the definition of $\mathcal{D}$ guarantees that for $t<\tau$, the binding and non-binding dimensions will not switch and this creates great convenience for the analysis.

Define 
$$\tau = \min \ \left\{n-\ceil{\frac{\bar{a}}{\underline{d}}}\right\} \cup \{ t \ge 0: \bm{d}_t \notin \mathcal{D}\},$$
where $\ceil{\cdot}$ represents the ceiling function.
Intuitively, $\tau$ is the first time that the binding/nonbinding structure of the problem may be changed, i.e., some binding constraints become non-binding or some non-binding constraints become binding. 
To put it in another way, it means for the $i$-th constraint, the average remaining resource level $d_{it}$ deviates from the initial $d_i$ by a constant. This is more strict than the stopping time $\tau_{\bar{a}}$ defined earlier in Theorem \ref{representation}. Comparatively, $\tau_{\bar{a}}$ denotes the time that a certain type of constraint is almost exhausted while $\tau$ here characterizes the time that a certain type of constraint deviates certain amount
from its original level $d_i.$ Since $\underline{d}$ is a lower bound for all $\tilde{\bm{d}}\in\Omega_d,$ the definition ensures that $\tau\le \tau_{\bar{a}}.$

We first work on the \textbf{second part of the generic upper bound} and derive an upper bound for $\E[n-\tau].$ Define 
$$d_{it}'= \begin{cases} d_{it}, & \text{ if } t\le \tau \\
d_{i,t-1}', &  \text{ if } t > \tau
\end{cases}$$
for $i=1,...,m$ and $t=1,...,n$. The process $\{d_{it}'\}_{t=1}^n$ can be interpreted as an auxiliary process associated with $\{d_{it}\}_{t=1}^n$ and it freezes the value of $d_{it}$ from the time $\tau.$ The motivation for defining this auxiliary process is that the original process $d_{it}$ may behave irregularly after time $\tau$ (the binding and non-binding dimensions switch), but the new process $d'_{it}$ remains the same value after time $\tau$. In such a way, we separate the effect of irregularity across different constraints and single out one constraint for analysis. Since our objective here is to analyze $\E[\tau]$ and the two processes take the same value before time $\tau$, so it makes no difference to study the more ``regular'' process $d'_{it}$.

As in the proof of Theorem \ref{Algo1}, we define a stopping time for each constraint,
$$\tau_i = \min\ \left\{n-\ceil{\frac{\bar{a}}{\underline{d}}}\right\} \cup \left\{t\ge1: d_{it}' \notin [d_i-\delta_d, d_i+\delta_d] \right\}.$$
It is easy to see that $\tau = \min_{i} \tau_i,$ so we only need to study $\tau_i$ and $\E[n-\tau_i].$ 

Define $\tilde{\p}^*_{t+1}$ be the optimal solution to the following optimization problem
\begin{align} 
\min & \ f_{\bm{d}_t}(\bm{p}) \coloneqq \bm{d}_t^\top \bm{p} + \E\left[(r-\bm{a}^\top \bm{p})^+\right] \label{t-step} \\
\text{s.t. \ } & \bm{p} \ge \bm{0},\nonumber
\end{align}
where $\Dd_t = (d_{1,t},...,d_{m,t})^\top.$ The problem (\ref{t-step}) is different from the original stochastic program (\ref{asymProblem}) in terms of $\Dd.$ The specification of $\bm{d}_t$ makes the stochastic program (\ref{t-step}) correspond to the SAA problem solved at time $t$ in Algorithm \ref{alg:HDLA}. Assumption \ref{assump3} ensures that the dual convergence result in Theorem \ref{expectation} extends to $\p_{t+1}$ (the dual price used in Algorithm \ref{alg:HDLA}) and $\tilde{\p}^*_{t+1}.$

Now, we analyze the dynamics of $d_{it}'$. With the execution of Algorithm \ref{alg:HDLA}, we have
$$\bm{b}_{t+1} = \bm{b}_{t} - \bm{a}_{t+1}I(r_{t+1}>\bm{a}_{t+1}^\top \p_{t+1})$$
for $t=0,1,...,n-1.$
Normalizing both sides,
\begin{align*}
    d_{i,t+1}' &= d_{it}' I(\tau < t) + \frac{(n-t)d_{it} - a_{i,t+1}I(r_{t+1}>\Aa^\top_{t+1} \p_{t+1})}{n-t-1} I(\tau \ge t) \\
    &= d_{it}' I(\tau < t) + d_{it}I(\tau>t) +  \frac{d_{it} - a_{i,t+1}I(r_{t+1}>\Aa^\top_{t+1} \p_{t+1})}{n-t-1} I(\tau \ge t),
\end{align*}
for $t=0,...,n-1$ and $i=1,...,m.$ 

Since $\tau$ is defined by the deviation from $d_i,$ we can take off $d_i$ on both sides,
\begin{align*}
    d_{i,t+1}'-d_i & = (d_{it}'-d_i) I(\tau < t) + (d_{it}-d_i)I(\tau\ge t) +  \frac{d_{it} - a_{i,t+1}I(r_{t+1}>\Aa^\top_{t+1} \p_{t+1})}{n-t-1} I(\tau \ge t) \\
& = (d_{it}'-d_i) I(\tau < t) + (d_{it}-d_i)I(\tau\ge t) +  \frac{d_{it} - a_{i,t+1}I(r_{t+1}>\Aa^\top_{t+1} \tilde{\p}^*_{t+1})}{n-t-1} I(\tau \ge t) \\ 
& \ \ \ + \frac{a_{i,t+1}\left(I(r_{t+1}>\Aa^\top_{t+1} \tilde{\p}^*_{t+1})- I(r_{t+1}>\Aa^\top_{t+1} {\p}_{t+1})\right)}{n-t-1} I(\tau \ge t).
\end{align*}
Taking square for both sides and take expectation,
\begin{align*}
\E \left[\left(d_{i,t+1}'-d_i\right)^2\right] &= \E \left[(d_{it}'-d_i)^2 I(\tau < t)\right] + \E \left[(d_{it}-d_i)^2 I(\tau \ge t)\right] \\
   & \ \ \ + \E\left[\frac{\left(d_{it} - a_{i,t+1}I(r_{t+1}>\Aa^\top_{t+1} \tilde{\p}^*_{t+1})\right)^2}{(n-t-1)^2} I(\tau \ge t)\right]
   \\ & \ \ \ + \E\left[\frac{\left(a_{i,t+1}I(r_{t+1}>\Aa^\top_{t+1} \tilde{\p}^*_{t+1}) - a_{i,t+1}I(r_{t+1}>\Aa^\top_{t+1} \p_{t+1})\right)^2}{(n-t-1)^2} I(\tau \ge t)\right]\\
   & \ \ \ + 2\E\left[\frac{\left(a_{i,t+1}I(r_{t+1}>\Aa^\top_{t+1} \tilde{\p}^*_{t+1}) - a_{i,t+1}I(r_{t+1}>\Aa^\top_{t+1} \p_{t+1})\right)(d_{it}-d_i)}{n-t-1} I(\tau \ge t)\right]
\end{align*}
\vspace{-0.7cm}
\begin{equation}
    + 2\E\left[\frac{\left(a_{i,t+1}I(r_{t+1}>\Aa^\top_{t+1} \tilde{\p}^*_{t+1}) - a_{i,t+1}I(r_{t+1}>\Aa^\top_{t+1} \p_{t+1})\right)\left(d_{it} - a_{i,t+1}I(r_{t+1}>\Aa^\top_{t+1} \tilde{\p}^*_{t+1})\right)}{(n-t-1)^2} I(\tau \ge t)\right].\label{d-d}
\end{equation}
where the expectation is taken with respect to the online process, i.e., $(r_t, \Aa_t)$'s. Here the cross terms that contain both $I(\tau<t)$ and $I(\tau \ge t)$ will cancel out because these two events are exclusive to each other.

We analyze (\ref{d-d}) separately for binding and non-binding dimensions. We emphasize that binding and non-binding dimensions are defined according to the original (initial) $\bm{d}$ but extend to all $\tilde{\bm{d}}\in\mathcal{D}.$ For \textbf{binding dimensions} ($i \in I_B$), the following cross term disappears
\begin{align*}
 & 2\E\left[\frac{\left(d_{it} - a_{i,t+1}I(r_{t+1}>\Aa^\top_{t+1} \tilde{\p}^*_{t+1})\right)(d_{it}-d_i)}{n-t-1} I(\tau \ge t)\right] \\
  = & 2\E \left [ \E\left[\frac{\left(d_{it} - a_{i,t+1}I(r_{t+1}>\Aa^\top_{t+1} \tilde{\p}^*_{t+1})\right)(d_{it}-d_i)}{n-t-1} I(\tau \ge t)\Bigg | d_{1,t},...,d_{m,t}\right]\right] \\
  = &  2\E \left [ \E\left[\frac{d_{it} - a_{i,t+1}I(r_{t+1}>\Aa^\top_{t+1} \tilde{\p}^*_{t+1})}{n-t-1} \Bigg | d_{1,t},...,d_{m,t}\right](d_{it}-d_i)I(\tau \ge t)\right] \\
  =& 0.
\end{align*}
The last line is because the definition of $\tau$ ensures that the binding and non-binding dimensions will not switch before time $\tau$ and therefore, for the binding dimensions, we always have 
$$d_{it} = \E_{(r,\Aa)\sim\mathcal{P}}[a_iI(r>\Aa^\top \tilde{\p}^*_{t+1})|\bm{d}_t]$$
due to the optimality condition of the stochastic program (\ref{t-step}) and the definition of $\tilde{\p}^*_{t+1}$. Intuitively, it means the stochastic program will output a solution $\tilde{\p}^*_{t+1}$ so that the future average constraint consumption under $\tilde{\p}^*_{t+1}$ is always equal to the current average constraint level $d_{it}$ for binding dimensions. 

Back to (\ref{d-d}), we analyze the right hand side term by term. 
\begin{itemize}
    \item[1)] Combine two terms, $$\E \left[(d_{it}'-d_i)^2 I(\tau < t)\right] + \E \left[(d_{it}-d_i)^2 I(\tau \ge t)\right] =\E \left[(d_{it}'-d_i)^2\right]^2.$$
    \item[2)] Utilize the upper bound on $d_i$'s and $\Aa_j$'s, $$\E\left[\frac{\left(d_{it} - a_{i,t+1}I(r_{t+1}>\Aa^\top_{t+1} \tilde{\p}^*_{t+1})\right)^2}{(n-t-1)^2} I(\tau \ge t)\right] \le \frac{(\bar{d}+\bar{a})^2}{(n-t-1)^2}.$$
    \item[3)] The following cross term characterizes the difference between the constraint consumption under $\bm{p}_{t+1}$ (computed by SAA in Algorithm \ref{alg:HDLA}) and $\tilde{\p}^*_{t+1}$ (the optimal solution of the stochastic program (\ref{t-step})).
    \begin{align*}
    & 2\E\left[\frac{\left(a_{i,t+1}I(r_{t+1}>\Aa^\top_{t+1} \tilde{\p}^*_{t+1}) - a_{i,t+1}I(r_{t+1}>\Aa^\top_{t+1} \p_{t+1})\right)(d_{it}-d_i)}{n-t-1} I(\tau \ge t)\right] \\ 
    \le & \frac{2}{n-t-1} \bar{a}^2\mu\|\tilde{\p}^*_{t+1}-\p_{t+1}\|_2 \sqrt{\E (d_{it}'-d_i)^2} 
    \\
    \le & \frac{2\bar{a}^2\mu C\sqrt{m}\sqrt{\log \log t}}{(n-t-1)\sqrt{t}} \cdot \sqrt{\E \left[(d_{it}'-d_i)^2\right]} \text{ \ (Applying the dual convergence result)}
\end{align*}
where the second line applies the Cauchy–Schwartz inequality and the constant $C$ is the coefficient of the dual convergence.
\item[4)] The following term also characterizes the difference in constraint consumption as the last one, but since the denominator is larger than the last one, we directly apply the trivial upper bound for the numerator. $$\E\left[\frac{\left(a_{i,t+1}I(r_{t+1}>\Aa^\top_{t+1} \tilde{\p}^*_{t+1}) - a_{i,t+1}I(r_{t+1}>\Aa^\top_{t+1} \p_{t+1})\right)^2}{(n-t-1)^2} I(\tau \ge t)\right] \le \frac{\bar{a}^2}{(n-t-1)^2}.$$
\item[5)] With the same reason as before, the denominator is large. So, we can apply the trivial upper bound for the numerator as follows. \begin{align*}
    & 2\E\left[\frac{\left(a_{i,t+1}I(r_{t+1}>\Aa^\top_{t+1} \tilde{\p}^*_{t+1}) - a_{i,t+1}I(r_{t+1}>\Aa^\top_{t+1} \p_{t+1})\right)\left(d_{it} - a_{i,t+1}I(r_{t+1}>\Aa^\top_{t+1} \tilde{\p}^*_{t+1})\right)}{(n-t-1)^2} I(\tau \ge t)\right] \\ & \le \frac{2\bar{a}(\bar{a}+\bar{d})}{(n-t-1)^2}.
\end{align*}
\end{itemize}

Combining the above five components (upper bounds) into (\ref{d-d}), we obtain
$$\E \left[\left(d_{i,t+1}'-d_i\right)^2\right] \le \E \left[\left(d_{it}'-d_i\right)^2\right] + \frac{4\bar{a}^2+\bar{d}^2+4\bar{a}\bar{d}}{(n-t-1)^2} + \frac{2\bar{a}^2\mu C\sqrt{m}\sqrt{\log \log t}}{(n-t-1)\sqrt{t}} \cdot \sqrt{\E \left[\left(d_{it}'-d_i\right)^2\right]}$$
for $i=1,...,m$ and $t=0,1,...,n-1.$ Also, the above inequality holds for all $n>0$ and distribution $\mathcal{P}\in \Xi.$
From Lemma \ref{sum_Z} in the following subsection, there exists a constant $\alpha$ such that
\begin{equation}
    \sum_{t=1}^n \E \left[\left(d_{it}'-d_i\right)^2\right] \le \alpha m\log n \log \log n \label{key-d}
\end{equation}
holds for binding dimensions $i\in I_B$, $n>0$ and distribution $\mathcal{P}\in \Xi.$ 
Then, we can analyze the stopping time associated with the $i$-th constraint (binding),
\begin{align}
    \E[n-\tau_i] \le &  \sum_{t=1}^n \prob(\tau_i \le t) \nonumber \\
    \le &  1+\frac{\bar{a}}{\underline{d}}+\sum_{t=1}^n \prob\left(|d_{it}' -d_i|\le \delta_d\right) \nonumber \\
    \le &  1+\frac{\bar{a}}{\underline{d}}+\sum_{t=1}^n \frac{\E\left[(d_{it}' -d_i)^2\right]}{\delta_d^2}  \text{ \ (Applying the Chebyshev's Inequality)}\nonumber  \\
    \le &1+\frac{\bar{a}}{\underline{d}}+ \frac{\alpha}{\delta_d^2}  m\log n \log \log n. \label{bindingStopTime}
\end{align}
Here the second line (with the extra term $\frac{\bar{a}}{\underline{d}}$) comes from the definition of $\tau_i$'s, and the third line applies Chebyshev's Inequality for the probabilities in the second line.  

Now, we analyze the process $d_{it}$ and the stopping time $\tau_i$ for \textbf{the non-binding dimensions}, i.e. $i \in I_N$. In fact, the non-binding dimensions are easier to analyze. This is by the definition of $\delta_d$ and the binding/non-binding dimensions, the average resource consumption (under the optimal solution to (\ref{t-step}) for any $\tilde{\bm{d}}\in \mathcal{D}$) will not exceed $d_i-\delta_d$ for non-binding dimensions. Thus there is a safety region of $n\delta_d$ that avoids the non-binding constraints to be exhausted too early. Also, Theorem \ref{representation} tells that there is no need to worry about the left-overs of non-binding dimensions at the end of the horizon. Specifically, for a non-binding dimension,  we apply the same argument as (\ref{MartingaleRepre}) and (\ref{part2sqrt}) in the proof of Theorem \ref{Algo1}. For the expectation of the resource consumption of non-binding constraints, we adopt a similar derivation of (\ref{expAlg1}),
\begin{align}
\E\left[ \sum_{j=1}^t a_{ij}I(r_j>\Aa_j^\top\p_j)\right]  &=\sum_{j=1}^t \E\left[a_{ij}I(r_j>{\Aa}_j^\top\p_j) \right]\nonumber \\
& \le   \sum_{j=1}^t \left(\E\left[a_{ij}I(r_j>{\Aa}_j^\top\p_j)\right]-\E \left[a_{ij}I(r_j>{\Aa}_j^\top\tilde{\p}^*_{j})\right]\right) + t(d_i-\delta_d) \nonumber\\
& \le \bar{a}^2\mu \sum_{j=1}^t  \E\|\p_j-\tilde{\p}^*_{j}\|_2 + t(d_i-\delta_d) \nonumber\\
& \le \bar{a}^2\mu \sum_{j=1}^t  \frac{C\sqrt{m}}{\sqrt{j}} \sqrt{\log\log j} + t(d_i-\delta_d) \nonumber \\
& \le  5C\bar{a}^2\mu\sqrt{m}\sqrt{t} \sqrt{\log \log t} + t(d_i-\delta_d). \label{expAlg2}
\end{align}
For the variance of the resource consumption of non-binding constraints, with the same analysis as in (\ref{varAlg1}),
\begin{equation}
    \text{Var} \left[\sum_{j=1}^t a_{ij}I(r_j>\Aa_j^\top\p_j)\right] \le \bar{a}^2 t + C^2 \bar{a}^4\mu^2 mt\log\log t.
    \label{varAlg2}
\end{equation}
Putting together (\ref{expAlg2}) and (\ref{varAlg2}),
\begin{align}
    \E[n-\tau_i] \le  &  \sum_{t=1}^n \prob(\tau_i \le t) \nonumber\\
    = &\sum_{t=1}^n  \mathbb{P}\left(\sum_{j=1}^{t'} a_{ij}I(r_j>\Aa_j^\top\p_j)\ge t(d_i-\delta_d) + n\delta_d\text{ for some } 1\le t'\le t\right) \nonumber\\
    \le & \sum_{t=1}^n\left(\frac{\bar{a}^2t+C^2\bar{a}^4\mu^2 mt\log\log t}{\left(n\delta_d - C\bar{a}^2\mu\sqrt{m}\sqrt{t}\sqrt{\log \log t}\right)^2}\right) \wedge 1 \nonumber \\
    \le & \frac{\alpha'}{\delta_d^2}m \log n \log\log n. \label{nonbindingStopTime}
\end{align}
There exists a constant $\alpha'$ dependent on $\bar{a}, \mu$ and $C$ such that the above inequality holds for all non-binding constraints $i\in I_N$. Here the last line comes from a similar derivation as Lemma \ref{lemmaTheoremAlgo2}. Technically, the bound here is tighter than the previous bound (\ref{part2sqrt}) because we utilize the knowledge that the non-binding constraint will not be consumed more than $d_i-\delta_d$ under $\tilde{\bm{p}}_t^*$ and thus it creates a gap of $n\delta_d$ in the second line above. This extra term (on the order of $n$) reduces the bound from $O(\sqrt{n}\log \log n)$ in (\ref{part2sqrt}) to $O(\log n \log \log n)$ here in (\ref{nonbindingStopTime}).

We complete the analysis of the \textbf{second part in the generic regret upper bound} by combining (\ref{bindingStopTime}) and (\ref{nonbindingStopTime}),
\begin{align}
 \E[n-\tau] = & \E[\max_{i} \{n-\tau_i\}]  \nonumber \\ 
 \le &  \sum_{i=1}^m \E[n-\tau_i]  \nonumber \\ 
 \le & \frac{\max\{\alpha,\alpha'\}}{\delta_d^2}  m^2\log n \log \log n + m\left(\frac{\bar{a}}{\underline{d}}+1\right). 
\label{part2Log}
\end{align}

Now, we analyze the \textbf{first part in the generic regret upper bound}. First, note that
\begin{align}
    & \E \left[\sum_{t=1}^\tau \|\p_t-\p^*\|_2^2\right] \nonumber \\
    \le & \E \left[\sum_{t=1}^\tau \|\p_t - \tilde{\p}_t^*\|_2^2 + \sum_{t=1}^\tau \|\tilde{\p}_t^*-\p^*\|_2^2\right]. \label{pSum0}
\end{align}
For the first summation in (\ref{pSum0}), we can apply the dual convergence result in Theorem \ref{expectation} which is always valid when $t\le \tau,$
\begin{align}
    \E \left[\sum_{t=1}^\tau \|\p_t - \tilde{\p}_t^*\|_2^2\right] & \le \E \left[\sum_{t=1}^n \|\p_t - \tilde{\p}_t^*\|_2^2\right] \nonumber \\
    & \le C\sum_{t=1}^n\frac{m}{t}\log t \log t \nonumber \\
    & \le  Cm\log n\log\log n \label{pSum1}
\end{align} 
where the constant $C$ is the coefficient of the dual convergence.

For the second summation in (\ref{pSum0}), we apply Lemma \ref{d-p} and our previous analysis of the process $\bm{d}_t,$
\begin{align}
    \E \left[\sum_{t=1}^\tau \|\tilde{\p}_t^*-\p^*\|_2^2\right]& \le \frac{1}{\lambda^2\lambda_{\min}^2}\E \left[ \sum_{i\in I_B}\sum_{t=1}^\tau (d_{it}-d_i)^2\right] \nonumber \\
    & \le \frac{1}{\lambda^2\lambda_{\min}^2}\alpha m^2\log n \log \log n
    \label{pSum2}
\end{align}
where the first line comes from Lemma \ref{d-p} and the second line comes from plugging in the result (\ref{key-d}) -- noting that $d_{it}=d_{it}'$ when $t\le \tau$.
Plugging (\ref{pSum1}) and (\ref{pSum2}) into (\ref{pSum0}), we obtain
\begin{equation}
    \E \left[\sum_{t=1}^\tau \|\p_t-\p^*\|_2^2\right] \le \left(Cm+ \frac{1}{\lambda^2\lambda_{\min}^2}\alpha m^2\right) \log n \log \log n
    \label{part1Log}
\end{equation}
for all $n$ and $\mathcal{P} \in \Xi.$ Thus we complete the analysis of the second part of the generic upper bound.

Now, we analyze the \textbf{third part of the generic upper bound}. From the definition of $\tau_i,$ we know
$$b_{in} \le (d_i+\delta_d)(n-\tau_i)$$
for $i\in I_B.$ Consequently,
$$\E\left[b_{in}\right] \le \E\left[(d_i+\delta_d) (n-\tau_i)\right]\le \frac{\bar{d}\alpha}{\delta_d^2}  m\log n \log \log n+\bar{d}\left(\frac{\bar{a}}{\underline{d}}+1\right),$$
\begin{equation}
    \E\left[\sum_{i\in I_B}b_{in}\right] \le  \frac{\bar{d}\alpha}{\delta_d^2}  m^2\log n \log \log n +m\bar{d}\left(\frac{\bar{a}}{\underline{d}}+1\right)
    \label{part3Log}
\end{equation}
for all $n>0$ and $\mathcal{P} \in \Xi.$

Combining (\ref{part1Log}), (\ref{part2Log}), and (\ref{part3Log}) with Corollary \ref{coro1}, we complete the proof.

\end{proof}

\subsubsection{Lemmas and Inequalities in the Proof of Theorem \ref{Algo2}}

\begin{lemma}
\label{d-p}
Under Assumption \ref{assump1} and \ref{assump3}, for any $\hat{\bm{d}}, \tilde{\bm{d}} \in \Omega_d$, let 
\begin{align*}
    \hat{\bm{p}}^* &\in \argmin_{\hat{\bm{p}}\ge \bm{0}}  \hat{f}(\bm{p}) \coloneqq \hat{\bm{d}}^\top \hat{\bm{p}} + \E\left[(r-\bm{a}^\top \hat{\bm{p}})^+\right], \\
    \tilde{\bm{p}}^* &\in \argmin_{\tilde{\bm{p}}\ge \bm{0}}  \tilde{f}(\tilde{\bm{p}}) \coloneqq  \tilde{\bm{d}}^\top \tilde{\bm{p}} + \E\left[(r-\bm{a}^\top \tilde{\bm{p}})^+\right]
\end{align*}
be the optimal solution to the according optimization problem.
Then,
$$\|\hat{\bm{p}}^* -\tilde{\bm{p}}^*\|_2^2 \le \frac{1}{\lambda^2\lambda_{\min}^2} \|\hat{\bm{d}}-\tilde{\bm{d}}\|_2^2.$$
In addition, if $\hat{f}$ and $\tilde{f}$ define the same binding and non-binding dimensions, $I_B$ and $I_N$ respectively, and the binding and non-binding dimensions are strictly complimentary to each other (as Assumption \ref{assump3} (c)), then
$$\|\hat{\bm{p}}^* -\tilde{\bm{p}}^*\|_2^2 \le \frac{1}{\lambda^2\lambda_{\min}^2}\sum_{i \in I_B} (\hat{d}_i-\tilde{d}_i)^2.$$
\end{lemma}

\begin{proof}
From Proposition \ref{strConvex}, we know 
$$\hat{f}(\tilde{\p}^*)- \hat{f}(\hat{\p}^*)
\ge\frac{\lambda\lambda_{\min}}{2} \|\tilde{\p}^*-\hat{\p}^*\|_2^2$$
$$\tilde{f}(\hat{\p}^*)- \tilde{f}(\tilde{\p}^*)
\ge\frac{\lambda\lambda_{\min}}{2} \|\tilde{\p}^*-\hat{\p}^*\|_2^2.$$
This is because the first-order term in Proposition \ref{strConvex} is non-negative, and both Assumption \ref{assump3} (b) and Proposition \ref{strConvex} hold for all $\hat{\bm{d}}, \tilde{\bm{d}}\in \Omega_d$. Adding up the two inequalities, 
we have
$$\hat{\bm{d}}^\top \tilde{\bm{p}}^* -\hat{\bm{d}}^\top \hat{\bm{p}}^* +\tilde{\bm{d}}^\top \hat{\bm{p}}^* - \tilde{\bm{d}}^\top  \tilde{\bm{p}}^*\ge\lambda\lambda_{\min} \|\tilde{\p}^*-\hat{\p}^*\|_2^2$$
where the expectation terms in $f$ and $\tilde{f}$ are cancelled out.
Then,
\begin{align*}
  \lambda\lambda_{\min} \|\tilde{\p}^*-\hat{\p}^*\|_2^2 & \le (\tilde{\bm{d}}-\hat{\bm{d}})^\top(\hat{\bm{p}}^*-\tilde{\bm{p}}^*) \\
  & \le  \|\hat{\bm{d}}-\tilde{\bm{d}}\|_2  \cdot \|\tilde{\p}^*-\hat{\p}^*\|_2.
\end{align*}
In addition,
\begin{align*}
  \lambda\lambda_{\min} \|\tilde{\p}^*-\hat{\p}^*\|_2^2 & \le (\tilde{\bm{d}}-\hat{\bm{d}})^\top(\hat{\bm{p}}^*-\tilde{\bm{p}}^*) \\
  & \le \sqrt{\sum_{i \in I_B} (\hat{d}_i-\tilde{d}_i)^2}  \cdot \|\tilde{\p}^*-\hat{\p}^*\|_2
\end{align*}
where only the binding dimensions remain in the last line because both $\tilde{\bm{p}^*}$ and $\hat{\bm{p}}^*$ are zero-valued on the non-binding dimensions.

\end{proof}

\begin{lemma}
Recall that the binding and non-binding dimensions specified by the stochastic program (\ref{asymProblem}) with parameter $\bm{d}=(d_1,...,d_m)^\top$ as $I_B$ and $I_N$. Under Assumption \ref{assump1} and \ref{assump3}, there exists a constant $\delta_d>0$ such that for all $\tilde{\bm{d}} \in \mathcal{D} = \bigotimes_{i=1}^m [d_i-\delta_d, d_i+\delta_d] \subset \Omega_d$, the stochastic program (\ref{asymProblem}) specified with the parameter $\tilde{\bm{d}}$ share the same binding and non-binding dimensions. 
\label{delta_d}
\end{lemma}

\begin{proof}
From Lemma \ref{d-p}, we know that for any $\tilde{\bm{d}}\in \Omega_d$ and the corresponding optimal solution $\tilde{\bm{p}}^*$, 
we have
$$\|\bm{p}^* -\tilde{\bm{p}}^*\|_2^2 \le \frac{1}{\lambda^2\lambda_{\min}^2} \|\bm{d}-\tilde{\bm{d}}\|_2^2.$$
If the conclusion in the lemma does not hold, there are two cases: 
\begin{itemize}
    \item[(i)] There is an index $i\in I_B$ but $i$ is a non-binding constraint for the stochastic program specified by $\tilde{\bm{d}}$, i.e., $p_i^*>0$ and $\tilde{p}_i^*=0$. 
    \item[(ii)] There is an index $i \in I_N$ but $i$ is a binding constraint for the stochastic program specified by $\tilde{\bm{d}}$, i.e., 
    \begin{equation}
    \label{d_Condition}
    \begin{array}{cc}
d_i& >\E[a_{ij}I(r_j>\bm{a}_j^\top \p^*)], \\
\tilde{d}_i &  =\E[a_{ij}I(r_j>\bm{a}_j^\top \tilde{\p}^*)].
    \end{array}
\end{equation}
    
\end{itemize}

For case (i), denote $\underline{p} = \min\{p_i^*: i\in I_B\}$. Then, we must have
$$  \|\bm{d}-\tilde{\bm{d}}\|_2^2\ge \lambda^2\lambda_{\min}^2\|\bm{p}^* -\tilde{\bm{p}}^*\|_2^2 \ge \underline{p}^2\lambda^2\lambda_{\min}^2.$$

For case (ii), from the inequality (\ref{g0P}) in the proof of Theorem \ref{Algo1}, we know that
\begin{equation}
  |\E[a_{ij}I(r_j>\bm{a}_j^\top \tilde{\p}^*)]-\E[a_{ij}I(r_j>\bm{a}_j^\top {\p}^*)]|\le \bar{a}^2 \mu \|\tilde{\p}^* -\p^*\|_2. \label{continuity_p}  
\end{equation}
Recall that from Assumption \ref{assump3} (c), 
$$d_i > \E[a_{ij}I(r_j>\bm{a}_j^\top {\p}^*)]$$
for $i\in I_N.$ Denote $$\gamma =\min_{i\in I_N}\left\{d_i- \E[a_{ij}I(r_j>\bm{a}_j^\top {\p}^*)]\right\}.$$
From Assumption \ref{assump3} (c), we know $\gamma>0.$
From \eqref{d_Condition} and \eqref{continuity_p}, we know that if $|d_i-\tilde{d}_i|\le \frac{\gamma}{2}$, then we must have 
\begin{align*}
  \|\tilde{\p}^* -\p^*\|_2&\ge \frac{1}{\bar{a}^2 \mu }  |\E[a_{ij}I(r_j>\bm{a}_j^\top \tilde{\p}^*)]-\E[a_{ij}I(r_j>\bm{a}_j^\top {\p}^*)]| \\
  & = \frac{1}{\bar{a}^2 \mu }  |\tilde{d}_i-\E[a_{ij}I(r_j>\bm{a}_j^\top {\p}^*)]| \ \ \text{ (from the condition \eqref{d_Condition} in case (ii))}\\
  & \ge \frac{1}{\bar{a}^2 \mu } \left( |d_i-\E[a_{ij}I(r_j>\bm{a}_j^\top {\p}^*)]| - |d_i-\tilde{d}_i| \right) \\
  & \ge \frac{\gamma}{2\bar{a}^2\mu},
\end{align*}
and consequently,
$$  \|\bm{d}-\tilde{\bm{d}}\|_2^2\ge \lambda^2\lambda_{\min}^2\|\bm{p}^* -\tilde{\bm{p}}^*\|_2^2 \ge \frac{\gamma^2\lambda^2\lambda_{\min}^2}{4\bar{a}^4\mu^2}.$$
Combining the two aspects, we know that when 
$$  \|\bm{d}-\tilde{\bm{d}}\|_2^2 < \min\left\{{\underline{p}^2}{\lambda^2\lambda_{\min}^2}, \frac{\gamma^2\lambda^2\lambda_{\min}^2}{4\bar{a}^4\mu^2}, \frac{\gamma^2}{4}\right\}.$$
Then we have the two stochastic programs specified by $\bm{d}$ and $\tilde{\bm{d}}$ must share the same binding and non-binding dimensions. And the existence of $\delta_d$ is implied by the last inequality.
\end{proof}

\begin{lemma}
\label{sum_Z}
If a sequence $\{z_{t}\}_{t=0}^n$ satisfies
$$
z_{t+1} 
=
z_{t} 
+
\frac{\sqrt{m\log \log t}\sqrt{z_{t}}}{(n-t-1)\sqrt{t}}+\frac{1}{(n-t-1)^2},
$$
and $z_0 = 0.$
Then we have
$$\sum_{t=1}^n z_{t} \le 32 \log n \log \log n$$
holds for all $n\ge 3.$
\end{lemma}

\begin{proof}
First, we only need to prove that if 
\begin{equation}
   z_{t+1} 
=
z_{t} 
+
\frac{\sqrt{z_{t}}}{(n-t-1)\sqrt{t}}+\frac{1}{(n-t-1)^2},\label{z_simplified} 
\end{equation}
then $$\sum\limits_{t=1}^{n}z_{t}\leq 32\log n.$$
This is because we could change the variable by introducing $z_t'=z_t \cdot m \log\log t$ and prove the above results for $z_t'.$ So we analyze the sequence $z_t$ under (\ref{z_simplified}) now.
Note that if $z_t\leq \frac{4t}{(n-t-1)^2}$,
$$
z_{t+1}
\leq
\frac{4t}{(n-t-1)^2}
+
\frac{3}{(n-t-1)^2}
\leq
\frac{4(t+1)}{(n-t-2)^2}.
$$
Also, we have $z_0=0$. Thus, when $t\leq n/2$, an induction argument leads to
$$z_{t}\leq  \frac{4t}{(n-t-1)^2}$$ 
and specifically, $z_{n/2+1} \leq 10/n$. 

Now, we analyze the case of $t>n/2.$
Note that if $$z_t \leq \frac{16}{n-t-1}$$ and $t\geq n/2$, we have
$$
z_{t+1}
\leq
\frac{16}{n-t-1}
+
\frac{1}{(n-t-1)^2}
+
\frac{4}{(n-t-1)^{3/2}\sqrt{t}}
\leq
\frac{16}{n-t-2}.
$$
Also, we have 
$$z_{n/2+1} \le \frac{10}{n} \le \frac{16}{n-\frac{n}{2}-1}.$$
Thus, with an induction argument for $t\geq n/2$, we obtain $$z_t\leq \frac{16}{n-t-2}.$$ Now, combining the two parts together,
\begin{equation*}
\begin{split}
    \sum\limits_{i=1}^{n-3}z_t 
    &\leq
    \sum\limits_{i=n/2}^{n-3}
    \frac{16}{n-t-2}
    +
    \sum\limits_{i=1}^{n/2}
    \frac{4(t+1)}{(n-t-3)^2}\\
    &\leq
    32\log n
    +
    \frac{n}{n-n/2-1}
    +
    \log(n-1-n/2)
    +
    \frac{1}{n-3}\\
    &\leq
    32\log n,
\end{split}
\end{equation*}
where $n>3$ and the second inequality is obtained by approximating the sum by integral.

\end{proof}

\section{One-Constraint Case and Lower Bound}

\renewcommand{\thesubsection}{D\arabic{subsection}}

\subsection{One-Constraint Case}
In this section, we discuss the OLP regret lower bound by relating the OLP problem with a statistical estimation problem. Specifically, we consider a one-constraint LP where $m=1$ in LP (\ref{primalLP}). We set $a_{1j} = 1$ for $j=1,...,n.$ Then the optimization problem becomes
\begin{align}
    \max\ & \sum_{j=1}^n r_j x_j \label{1DLP} \\
    \text{s.t.}\ & \sum_{j=1}^n x_j \le nd, \ \   x_j \in [0,1]. \nonumber
\end{align}
This one-constraint case of OLP has been discussed extensively and is known as the multi-secretary problem \citep{kleinberg2005multiple, arlotto2019uniformly, bray2019does}. 
To simplify our discussion, we assume $nd$ to be an integer. There exists an integer-valued optimal solution of (\ref{1DLP}), given by
$$x_t^* = \begin{cases}
1, & r_t \ge \hat{Q}_n(1-d) \\
0, & r_t < \hat{Q}_n(1-d)
\end{cases}$$
where $\hat{Q}_n(\eta)$ defines the sample $\eta$-quantile of $\{r_j\}_{j=1}^n$, i.e. $\hat{Q}_n(\eta) = \inf\left\{v\in \mathbb{R}: \frac{\sum_{j=1}^n I(v>r_j)}{n}\ge \eta \right \}.$ The sample quantile $\hat{Q}_n(1-d)$ is indeed the dual optimal solution $p_n^*$ in the general setting. Intuitively, this optimal solution allocates resources to the proportion of orders with highest returns. We restate Assumption \ref{assump1} and \ref{assump2} in this one-constraint case as follows.

\begin{assumption}
\label{r_t}
Assume $d \in (0,1)$ and $\{r_j\}_{j=1}^n$ is a sequence of i.i.d. random variables supported on $[0,1]$. Assume it has a density function $f_r(x)$ s.t. $\lambda' \le f_r(x) \le \mu'$ for $x\in[0,1]$ with $\lambda',\mu'>0$. Denote the set of all distributions $\mathcal{P}_r$ satisfying the above assumptions as $\Xi_r.$
\end{assumption}

We restrict our attention to a class of thresholding policies as discussed in Section \ref{thresPolicy}. At each time $t$, we compute a dual price from the history inputs, $p_t = h_t(r_1, x_1,...,r_{t-1}, x_{t-1})$ and if the constraint permits, set 
$$x_t = \begin{cases} 
1,& r_t > p_t.\\
0,& r_t \le p_t.
\end{cases}$$
If $\sum_{j=1}^t x_t = nd$ for some $t$, we require $p_s = h_s(r_1, x_1,...,r_{s-1}, x_{s-1}) =1$ for all $s>t.$ In this way, all the future orders will be automatically rejected. Similar to the general OLP setting, an online algorithm/policy in this one-constraint problem can be specified by the sequence of functions $h_t$'s, i.e., $\bm{\pi} = (h_1,...,h_n).$

\subsection{Lower Bound}

Theorem \ref{p-p} establishes a regret lower bound for the one-constraint problem. The inequality tells that if we view the threshold $p_j$ at each step as an statistical estimator, the regret of this one-constraint problem is no less than the cumulative estimation error of a certain quantile of the distribution $p^*=Q_{\tau}(1-d)$. The significance of this inequality lies in the fact that the term on its right-hand side does not involve the constraint. In an online optimization/learning setting, a violation of the binding constraint will potentially improve the reward and reduce the regret; however, it does not necessarily help decrease the term on the right-hand side. Therefore, while studying the lower bound, we can focus on the right-hand side and view it as the estimation error from an unconstrained problem. 

\begin{theorem}
The following inequality holds
\begin{align}
    \E_{\mathcal{P}_r} \left[R_n^*-{R}_n(\bm{\pi})\right] \ge \frac{\lambda'}{2} \cdot \E_{\mathcal{P}_r}\left[\sum_{t=1}^n (p_t - p^*)^2 \right] +  \E_{\mathcal{P}_r}\left[R_n^*\right] - ng(p^*) \label{regretLower}
\end{align}
holds for any the thresholding policy $\bm{\pi}$ and any distribution $\mathcal{P}_r \in \Xi_r.$
Here $p_t=h_t(r_1,x_1,...,r_{t-1},x_{t-1})$ is specified by the policy $\bm{\pi},$ and $p^* = Q_r(1-d)$ is the $(1-d)$-quantile of the random variable $r_j$ with ${Q}_r(\eta) \coloneqq \inf\left\{v\in \mathbb{R}: \prob(r\le v) \ge \eta \right\}.$ The parameter $\lambda'$ comes from Assumption \ref{r_t}. The function $g(\cdot)$ is the same as defined in Section 4.2.
\label{p-p}
\end{theorem}

The proof of Theorem \ref{p-p} follows the same approach as the derivation of the upper bound in Theorem \ref{representation}. We establish $ng(p^*)$ as an upper bound for $\E\left[R_n^*\right]$ and compare $\E\left[R_n(\pi)\right]$ against $ng(p^*)$. Therefore, the gap between $\E_{\mathcal{P}_r}\left[R_n^*\right]$ and $ng(p^*)$ will affect the tightness of this lower bound. Fortunately, this gap should be small in many cases; for example, in the proof of Theorem \ref{lower_bound}, we use the dual convergence result 
in Theorem \ref{expectation} and show that 
$$\E_{\mathcal{P}_r}\left[R_n^*\right] - ng(p^*) \ge -C\log \log n$$
for this one-constraint problem. The following corollary presents the lower bound in a more similar form as the upper bound. The right hand side of the lower bound in Corollary \ref{coro2} also involves $\E[n-\tau_0].$ Note that the terms $\E[n-\tau_0]$ and $\E\left[\sum_{i\in I_B}b_{in}\right]$ symmetrically capture the overuse and underuse of the constraints, and therefore they should be on the same order.  Corollary \ref{coro2} tells that the upper bound given in Theorem \ref{representation} is rather tight, and that a stable control of the resource consumption is indispensable because the term $\E[n-\tau_0]$ also appears in the lower bound of the regret.

\begin{corollary}
The following inequality holds,
\begin{align*}
    \E_{\mathcal{P}_r} \left[R_n^*-{R}_n(\bm{\pi})\right] \ge \frac{\lambda'}{2} \cdot \E_{\mathcal{P}_r}\left[\sum_{t=1}^{\tau_0} (p_t - p^*)^2 \right] + \frac{\lambda'}{2}\left(1-p^*\right)^2\E[n-\tau_0] + \E_{\mathcal{P}_r}\left[R_n^*\right] - ng(p^*)
\end{align*}
holds for any the thresholding policy $\bm{\pi}$ and any distribution $\mathcal{P}_r \in \Xi_r.$
Here $p_t=h_t(r_1,x_1,...,r_{t-1},x_{t-1})$ is specified by the policy $\bm{\pi},$ and $p^* = Q_r(1-d)$.
The stopping time 
$$\tau_0=\min\{n\} \cup \left\{t: \sum_{j=1}^t x_j =  nd\right\}$$
represents the first time that the resource is exhausted.
\label{coro2}
\end{corollary}

Based on Theorem \ref{p-p}, we can derive a lower bound for the OLP problem by analyzing the right-hand-side of (\ref{regretLower}). The idea is to find a parametric family of distributions and we relate $p^*$ with the parameter that specifies the distribution $\mathcal{P}_r$. The $t$-th term on the right-hand side then can be viewed as the approximation error of the parameter $p^*$ using the first $t-1$ observations. Theorem \ref{lower_bound} states the lower bound for the OLP problem. The proof considers a family of truncated exponential distributions that satisfies Assumption \ref{r_t}, and it mimics the derivation of lower bounds in \citep{keskin2014dynamic, besbes2013implications}. 
The core part is the usage of van Trees inequality \citep{gill1995applications} -- a Bayesian version of the Cramer-Rao bound. 

\begin{theorem*}
There exist constants $\underline{C}$ and $n_0>0$ such that
$$\Delta_{n}(\bm{\pi}) \ge \underline{C} \log n$$
holds for all $n \ge n_0$ and any dual-based policy $\bm{\pi}.$
\end{theorem*}

Theorem \ref{lower_bound} indicates that Algorithm \ref{alg:HDLA} is an asymptotically near-optimal algorithm for the OLP problem under the fixed-$m$ and large-$n$ regime. The lower bound $O(\log n)$ is also consistent with the lower bound of the unconstrained online convex optimization problem \citep{abernethy2008optimal}.

\subsection{Proof of Theorem \ref{p-p}}

\begin{proof}
First,
\begin{align*}
    \E \left[{R}_n(\bm{\pi})\right] & = \E \left[\sum_{t=1}^n r_t x_t\right] \\
    & \stackrel{(a)}{\le} \E \left[\sum_{t=1}^n r_t x_t - \left(nd -\sum_{t=1}^n x_t\right) p^*\right] \\
    & = \E \left[\sum_{t=1}^n\left( r_tx_t + dp^* - x_tp^*\right)\right] \\
    & =  \sum_{t=1}^n \E \left[ r_tx_t + dp^* - x_tp^*\right] \\
    & \stackrel{(b)}{=}  \sum_{t=1}^n \E \left[g(p_t)\right].
\end{align*}
where the expectation is taken with respect to $r_t\sim \mathcal{P}_r.$ (a) comes from that the constraint must be satisfied and (b) comes from the definition of $g(\cdot).$ We do not need to consider the stopping time because after the resource is exhausted, the setting of $p_s=1$ is consistent with the enforcement of $x_s=0$. 

Then, 
\begin{align*}
    ng(p^*) - \E \left[{R}_n(\bm{\pi})\right] & \ge \sum_{t=1}^n \E[g(p^*) - g(p_t)] \\
    & = \sum_{t=1}^n \E \left[(r_t - p^*)I(r_t>p^*) - (r_t - p^*)I(r_t>p_t) \right] \\ 
    & \ge\frac{\lambda'}{2}\sum_{t=1}^n \E \left[(p_t-p^*)^2\right]
\end{align*}
where the expectation is taken with respect to $r_t\sim \mathcal{P}_r$ and the last line comes from Assumption \ref{r_t}.

\end{proof}

\subsection{Proof of Corollary \ref{coro2}}
\begin{proof}
Since the algorithm enforces $p_t=1$ for $t>\tau_0.$ The result follows by splitting the summation on the right-hand-side of (\ref{lower_bound}).
\begin{align*}
    \E_{\mathcal{P}_r} \left[R_n^*-{R}_n(\bm{\pi})\right] & \ge \frac{\lambda'}{2} \E_{\mathcal{P}_r}\left[\sum_{t=1}^n (p_t - p^*)^2 \right] +  \E_{\mathcal{P}_r}\left[R_n^*\right] - ng(p^*) \\
    & = \frac{\lambda'}{2} \E_{\mathcal{P}_r}\left[\sum_{t=1}^{\tau_0} (p_t - p^*)^2 \right] + \frac{\lambda'}{2}\E_{\mathcal{P}_r}\left[\sum_{t=\tau_0+1}^{n} (p_t - p^*)^2 \right] +  \E_{\mathcal{P}_r}\left[R_n^*\right] - ng(p^*) \\
    & = \frac{\lambda'}{2}  \E_{\mathcal{P}_r}\left[\sum_{t=1}^{\tau_0} (p_t - p^*)^2 \right] + \frac{\lambda'}{2}\left(1-p^*\right)^2\E[n-\tau_0] + \E_{\mathcal{P}_r}\left[R_n^*\right] - ng(p^*).
\end{align*}
\end{proof}

\subsection{Proof of Theorem \ref{lower_bound}}

We first introduce the van Tree inequality and refer its proof to \citep{gill1995applications}.
\begin{lemma}[\cite{gill1995applications}]
Let $(\mathcal{X}, \mathcal{F}, P_{\theta}: \theta\in \Theta)$ be a dominated family of distributions on some sample space $\mathcal{X}$; denote the dominating measure by $\mu$. The parameter space $\Theta$ is a closed interval on the real line. Let $f(x|\theta)$ denote the density of $P_\theta$ with respect to $\mu.$ Let $\lambda(\theta)$ denote the density function of $\theta$. Suppose that $\lambda$ and $f(x|\cdot)$ are both absolutely continuous, and that $\lambda$ converges to zero at the endpoints of the interval $\Theta.$ Consider $\phi: \Theta \rightarrow \mathbb{R}$ a first-order differentiable function. Let $\hat{\phi}(X)$ denote any estimator of $\phi(\theta)$. Then, 
$$\E [\hat{\phi}(X) - \phi(\theta)]^2 \ge \frac{[\E \phi'(\theta)]^2}{\E[\mathcal{I}(\theta)]+ \mathcal{I}(\lambda)}$$
where the expectation on the left hand side is taken with respect to both $X$ and $\theta,$ and the expectation on the right hand side is taken with respect to $\theta.$ $\mathcal{I}(\theta)$ and $\mathcal{I}(\lambda)$ denote the Fisher information for $\theta$ and $\lambda$, respectively, 
$$\mathcal{I}(\theta) \coloneqq \E \left[\left(\log f(X|\theta)'\right)^2\Big |\theta\right]$$
$$\mathcal{I}(\lambda)  \coloneqq \E \left[\left(\log \lambda(\theta)'\right)^2\right].$$
\label{vanT}
\end{lemma}

Now, we proceed to prove Theorem \ref{lower_bound}.
\begin{proof}[Proof of Theorem \ref{lower_bound}.]
First, we analyze the gap between $\E[R_n^*]$ and $ng(p^*)$.
\begin{align}
    ng(p^*) - \E[R_n^*] & \stackrel{}{\le}   \sum_{t=1}^n\left( \E\left[\left(p^*-p_n^*\right)I(p^*\ge r_t> p_n^*)\right] + \E\left[\left(p_n^*-p^*\right)I( p^*<r_t\le p_n^*)\right]\right) \nonumber  \\
    & = \sum_{t=1}^n\left( \E\left[\left(p^*-p_n^*\right)\prob(p^*\ge r_t> p_n^*|p_n^*)\right] + \E\left[\left(p_n^*-p^*\right)\prob( p^*<r_t\le p_n^*|p_n^*)\right]\right) \nonumber \\
    & \stackrel{}{\le} n \mu' \E \left[|p^*-p_n^*|_2^2\right] \nonumber \\
    & \le C \mu' \log\log n  \label{l_b_p1}
\end{align}
where the constant $C$ is the coefficient of dual convergence in Theorem \ref{expectation}. Here the first line comes from the proof of Lemma \ref{gP} and the third line comes from Assumption \ref{r_t}.
The last comes for applying the dual convergence result to this special one-constraint case. Next, we derive a lower bound for $$\E\left[\sum_{j=1}^n(p_j-p^*)^2\right]$$
with the help of the van Tree's inequality in Lemma \ref{vanT}.
For the lower bound, we only need to derive under a specific distribution. Consider a truncated exponential distribution for $r_j$'s
$$f(r|\theta) = \frac{\theta e^{-\theta r} I(r\in[0,1])}{1-e^{-\theta}}$$
and the Beta distribution as the priori for the parameter $\theta$
$$\lambda(\theta) = (\theta-1)^2(2-\theta)^2$$
with the support $\Theta = [1,2].$ Let $d=1/2$ and then
$$p^* = Q_{\frac{1}{2}}(r) = \phi(\theta) \coloneqq \frac{1}{\theta} \log \left(\frac{1}{2}+\frac{1}{2} \theta \right).$$ Additionally,
$$\left[\E \phi'(\theta)\right]^2 = \left[\int_{1}^2 \phi'(\theta) \lambda(\theta)d \theta \right]^2 \coloneqq c_1 \approx 0.006.$$
The Fisher information $\mathcal{I}(\theta)$ and $\mathcal{I}(r;\lambda)$ can be computed according to the definition.
\begin{align*}
    \E\left[\mathcal{I}(r;\theta) \right]& = \E\left[\E \left[\left(\log f(r|\theta)'\right)^2\Big |\theta\right]\right] \\
    & = \int_{1}^{2} \int_{0}^{1} \left(\log f(r|\theta)'\right)^2 f(r|\theta) \lambda(\theta)drd\theta \coloneqq c_2>0. \\
    \mathcal{I}(\lambda) & = \E \left[\left(\log \lambda(\theta)'\right)^2\right] \\
    & = \int_{0}^{1} \left(\log \lambda(\theta)'\right)^2 \lambda(\theta) d\theta \coloneqq c_3>0.
\end{align*}
In above, $c_1,$ $c_2$ and $c_3$ are deterministic real numbers that can be computed from the corresponding integrals. 

Then, $p_j$ can be viewed as an estimator of $p^* = \phi(\theta)$ with the first $j-1$ samples. Consider the fact that
$$\E\left[\mathcal{I}(r_{1:(j-1)};\theta) \right] = (j-1) \E\left[\mathcal{I}(r;\theta) \right].$$
We apply Lemma \ref{vanT} and obtain,
\begin{equation}
    \sup_{\theta} \E\left[\sum_{j=1}^n(p_j-p^*)^2\right] \ge \sum_{j=2}^n \frac{c_1}{c_2(j-1)+c_3} \ge c_4\log n \label{l_b_p2}
\end{equation}
for all $n>0$ with some constant $c_4$ dependent on $c_1,$ $c_2,$ and $c_3.$ Combining (\ref{l_b_p1}) and (\ref{l_b_p2}) with Theorem \ref{p-p}, there exist a $\theta$ and a distribution $\prob(r|\theta)$ such that
$$\E R_n^* - \E R_n(\bm{\pi}) \ge \frac{c_4\lambda'}{2}\log n - C\mu' \log \log n.$$ 
Set $n_0 = \min\left\{n\ge 0: c_4\lambda'\log n \le 4 C\mu' \log \log n\right\}$ and $\underline{C} = \frac{c_4\lambda'}{4}$. The lower bound result follows. 

\end{proof}

\subsection{More Discussions on the Algorithms}
\label{numeric_m}

In the paper, we present three different algorithms and derive corresponding regret bounds with the help of Theorem \ref{representation} and Corollary \ref{coro1}. As an summary, Table \ref{tab:summary_proof} presents the regret upper bounds of the algorithms separately with respect to the three components in Theorem \ref{representation}. Table \ref{tab:summary} summarizes the prior knowledge and the computational cost of the three algorithms. Algorithm \ref{alg:Distribution} uses $\bm{p}^*$ as the dual price and it serves for a benchmark purpose. Strictly speaking, Algorithm \ref{alg:Distribution} is not an OLP algorithm in that it assumes the knowledge of the distribution $\mathcal{P}$ and that to compute $\p^*$ exactly is hardly practical.  Algorithm \ref{alg:IDLA} is both a learning version of Algorithm \ref{alg:Distribution} and a simplified version of the dynamic learning algorithm proposed in \citep{agrawal2014dynamic}. Algorithm \ref{alg:HDLA} introduces a history-action-dependent mechanism to stabilize the constraint consumption. It is an adaptive version of Algorithm \ref{alg:IDLA} and an extension of the re-solving technique (in network revenue management literature) to a learning and more general context. 

\begin{table}[ht!]
    \centering
    \begin{tabular}{c|c|c|c|c}
    \toprule
         Algorithm & $\E\left[\sum_{t=1}^\tau \|\p_t-\p^*\|_2^2\right]$ &  $\E\left[n-\tau\right]$ & $\E \left[\sum_{i\in I_B} b_{in}\right]$ & \text{Regret} \\\midrule
        Algorithm 1  & 0 & $\tilde{O}(\sqrt{n})$ & $\tilde{O}(\sqrt{n})$ & $\tilde{O}(\sqrt{n})$\\
         Algorithm 2 & $\tilde{O}(\log n)$ & $\tilde{O}(\sqrt{n})$ & $\tilde{O}(\sqrt{n})$ & $\tilde{O}(\sqrt{n})$ \\
        Algorithm 3  & $\tilde{O}(\log n)$ & $\tilde{O}(\log n)$ &$\tilde{O}(\log n)$ & $\tilde{O}(\log n)$\\
    \bottomrule
    \end{tabular}
    \caption{Three components of the upper bound in Theorem \ref{representation}/Corollary \ref{coro1}.} 
    \label{tab:summary_proof}
\end{table}

\begin{table}[ht!]
    \centering
    \begin{tabular}{c|c|c|c}
    \toprule
         Algorithm &  Prior Knowledge & Computational Cost & Regret \\\midrule
        Algorithm 1  & $(d_1,...,d_m),\mathcal{P}, \p^*$&$ O(1)$ & $O(\sqrt{n})$ \\
         Algorithm 2  & $(d_1,...,d_m)$ & $O(\log n)$ & $O(\sqrt{n}\log n)$  \\
        Algorithm 3  &$(d_1,...,d_m)$, n &$O(n)$& $O(\log n \log \log n)$ \\
    \bottomrule
    \end{tabular}
    \caption{Algorithm summary and comparison}
    \label{tab:summary}
\end{table}

The table shows that the bottleneck for Algorithm \ref{alg:Distribution} and \ref{alg:IDLA} lies in the control of constraint consumption. Intuitively, these two algorithms are not adaptive enough and the constraint consumption in each period is ``independent'' of the current constraint level. This causes a fluctuation of $O(\sqrt{n})$ after $n$ periods, and that is essentially the reason why the last two terms are $O(\sqrt{n})$ for these two algorithms. Moreover, Algorithm \ref{alg:HDLA}, in contrast with the geometrically updating scheme in Algorithm \ref{alg:IDLA}, updates the dual price after every period. The analysis of Algorithm \ref{alg:HDLA} indicates the goal of this more frequent updating scheme is not to further reduce the approximation error of $\p^*$, but to stabilize the constraint consumption.

In Table \ref{tab:summary}, the computational cost is measured by the number of LPs or optimization problems that need to be solved throughout the process. Algorithm \ref{alg:Distribution} utilizes the distribution knowledge and thus only needs to optimize once.  Algorithm \ref{alg:IDLA} is notably more computationally efficient than Algorithm \ref{alg:HDLA}. However, the $O(n)$ computational cost of Algorithm \ref{alg:HDLA} can be significantly curtailed in practice by using $\p_t$ as the initial point while solving the optimization problem for $\p_{t+1}$ (as in papers \citep{gupta2014experts, agrawal2014fast, devanur2019near}). The dual convergence result tells us that $\p_t$ and $\p_{t+1}$ are close to each other and this makes $\p_t$ as a good warm start for $\p_{t+1}$. 

Also, we further illustrate the performance of the algorithms with respect to the number of constraints $m$. Table \ref{tab:Regret_m} reports the performance under the model of Random Input I and Random Input II with fixed $n=500$ but different values of $m$. From the table, we conjecture that under Random Input I, the regret increases sublinearly as $m$ grows, while under Random Input II, the regret grows linearly as $m$ grows.

\begin{table}[ht!]
    \centering
    \small
    \begin{tabular}{c|ccc|ccc}
    \toprule
    Model  & \multicolumn{3}{c|}{Random Input I} & \multicolumn{3}{c}{Random Input II} \\ \midrule
      Algorithm   & A1& A2& A3 &A1& A2& A3 \\ \midrule 
      m = 5, n = 500   &90.64&109.07&\textbf{31.89}&36.87& 120.27 &\textbf{8.73}\\
      m = 10, n = 500   &135.55& 152.28&  \textbf{38.23}& 62.65& 174.04  &\textbf{25.52}\\
      m = 50, n = 500  &228.43& 255.74&  \textbf{56.22}&853.98& 647.83 &\textbf{369.99}\\
      m = 100, n = 500    &251.95& 296.96 &\textbf{70.34}&2189.28& 1732.12 &\textbf{1197.30}\\
      m = 200, n = 500    &281.72& 319.55&\textbf{76.51}&4975.79& 4291.45&\textbf{3351.86}\\
      \bottomrule
    \end{tabular}
    \caption{Regret performance: A1, A2, and A3 stand for Algorithm 1 (No-need-to-learn), Algorithm 2 (Simplified Dynamic Learning), and Algorithm 3 (Action-history-dependent), respectively.}
    \label{tab:Regret_m}
\end{table}

\end{document}